\documentclass[12pt]{article}
\usepackage[utf8]{inputenc}
\usepackage{latexsym}
\usepackage{soul}
\usepackage{amssymb,amsmath, amsthm}
\usepackage{pifont}
\usepackage{listings}
\usepackage{mathrsfs}
\usepackage[pdftex]{graphicx}
\usepackage{tikz}
\usetikzlibrary{intersections, calc}
\usetikzlibrary{patterns}
\usetikzlibrary{3d}
\usepackage{changepage}
\providecommand{\keywords}[1]{\par\vspace{0.5em}\noindent\textbf{#1}\par\vspace{0.5em}}
\usepackage{setspace}
\usepackage{geometry}  
\usepackage{hyperref}  
\usepackage[utf8]{inputenc}
\usepackage[affil-it]{authblk} 
\usepackage{algorithm}
\usepackage{algorithmicx}
\usepackage{algcompatible}
\usepackage{threeparttable}
\usepackage{pdflscape}

\usepackage{etoolbox}
\usepackage{lmodern}
\usepackage[most]{tcolorbox}
\usepackage{pst-node}
\usepackage{tikz-cd} 
\usepackage{empheq}
\usepackage{cancel}
\usepackage{ dsfont }
\usepackage{enumitem}
\usepackage{chngcntr}
\usepackage{ mathrsfs }
\usepackage{thmtools}
\usepackage{bbm}
\usepackage{booktabs}
\usepackage{subcaption}

\allowdisplaybreaks

\usepackage{titling}

\usepackage{etoolbox}
\AtBeginEnvironment{proof}{\setlength{\parindent}{10pt}}

\usepackage{natbib}
\theoremstyle{definition}

\newtheorem{remark}{Remark}

\theoremstyle{plain}
\newtheorem{assumption}{Assumption}
\newtheorem{definition}{Definition}

\newtheorem{theorem}{Theorem}

\newtheorem{lemma}{Lemma}

\DeclareMathOperator{\sgn}{{sgn}}

\DeclareMathOperator{\1}{\mathbbm{1}}

\usepackage{array} 
\newcolumntype{C}[1]{>{\centering\arraybackslash}p{#1}}

\usepackage[most]{tcolorbox}

\tcbset{
    breakable,
    enhanced,
    arc=0mm,
    outer arc=0mm,
    boxrule=0.3mm 
}

\makeatletter
\DeclareRobustCommand{\varamalg}{%
  \mathbin{\mathpalette\var@malg\perp}%
}
\newcommand{\succprec}{\mathrel{\mathpalette\succ@prec{\succ\prec}}}
\newcommand{\precsucc}{\mathrel{\mathpalette\succ@prec{\prec\succ}}}

\newcommand{\succ@prec}[2]{\succ@@prec#1#2}
\newcommand{\succ@@prec}[3]{%
  \vcenter{\m@th\offinterlineskip
    \sbox\z@{$#1#3$}%
    \hbox{$#1#2$}\kern-0.4\ht\z@\box\z@
  }%
}
\newcommand\var@malg[2]{%
  \rlap{$\m@th#1#2$}\mkern6mu{#1#2}%
}
\makeatother

\usepackage{float}

\hypersetup{hidelinks}

\title{Estimating Social Norm Complementarities\thanks{We thank seminar participants at Harvard and Princeton QSSC for helpful comments. Sebastian Garcia-Torres and Alessandro Fusari provided excellent research assistance. }}

\date{\today}

\author{Eliana La Ferrara\thanks{Harvard Kennedy School, NBER, CEPR and LEAP.} \and Cheaheon Lim\thanks{Department of Economics, Harvard University.} \and Davide Viviano\thanks{Department of Economics, Harvard University}}

\begin{document}

\maketitle

\begin{abstract}

We develop  a model of choice over social norms that allows for complementarities along two dimensions: \textit{technological}, analogous to complementarities between consumption goods, and \textit{social}, capturing  returns from conformity. Together, these determine whether two norms are complements, substitutes, or independent, as defined by how the equilibrium prevalence of one norm responds to a marginal shift in the utility of another. We estimate the model using repeated cross-sections from Sierra Leone and Nigeria, focusing on female genital cutting, polygyny, and child marriage. Social returns are significant across all specifications. For female genital cutting and child marriage, we find evidence of complementarities, especially strong in Sierra Leone. For polygyny and child marriage,  we find evidence of social substitutability, particularly in Nigeria. We interpret these differences using insights from anthropology. Finally, we iterate the model forward to study policy counterfactuals, assessing the potential effects of legal reforms and social interventions.
\end{abstract}

\keywords{\textit{Keywords:} Social norms, Social interactions, norm complementarities, social spillovers}

\newpage 

\onehalfspacing


\section{Introduction} 

A large literature in economics documents the influence of social norms on individual behavior. This spans a wide array of norms, including norms of cooperation and reciprocity \citep[e.g.,][]{fehr1999theory, fehr_2000, kandori_1992}, trust and social capital \citep[e.g.,][]{alesina2002trusts, guiso2004role}, fertility and cultural transmission \citep[e.g.,][]{bisin2001economics, fernandez2009culture}, risk-sharing and solidarity \citep[e.g.,][]{townsend1994risk, jakiela2016does}, and a range of harmful gender practices, such as child marriage \citep[e.g.,][]{field2008early, corno2020age, buchmann2023signal}, female genital cutting \citep[e.g.,][]{efferson2015fgc, Gulesci_etal_AER_2025}, and restricting female labor force participation \citep[e.g.,][]{jayachandran2021microentrepreneurship}. Despite the breadth of these contributions, the standard approach studies norms one at a time: the equilibrium, persistence, and policy responsiveness of each norm are analyzed as if compliance with a given norm is unaffected by compliance with other norms.\footnote{A recent exception is \citet{Bargain_2026}, who study how mixtures of ancestral practices relate to women's empowerment across countries.}

This paper argues that norms may be \textit{interdependent}: the returns from adhering to one norm can depend on the prevalence of another, similar to how the utility of consuming one good can depend on the level of consumption of another good.  In the case of \textit{complementarities}, the returns from complying with two norms simultaneously exceed the sum of the returns from complying with each norm in isolation. This is likely the case, for example, with norms of cooperation and enforcement, which intuitively reinforce each other \citep[e.g.,][]{kandori_1992, greif1993contract, henrich2001search, tabellini_2008scope}. In the case of \textit{substitutabilities}, complying with one norm reduces the returns from complying with another.  Norms about rewarding excellence and ``don't stand out" norms is one natural setting where substitutabilities can emerge  \citep[e.g.,][]{jakiela2016does, kandel1992peer, ichino2000work, almaas2020cutthroat, cappelen2020fair}. 

Modeling the interdependence among different norms can be quite consequential, from both a theoretical and a policy perspective. If two norms are complements, a shock that lowers adherence to one (e.g., a policy reform) will propagate to the other through both technological and social channels, amplifying the overall policy response. If they are substitutes, targeting one norm with the aim of decreasing its prevalence may induce reallocation toward the other, thus offsetting---at least in part---the intended policy goal. The sign and magnitude of these cross-norm spillovers are central to the design of interventions aimed at changing norm compliance.

We propose a model of choice over \textit{bundles} of norms where utility depends on an individual's intrinsic valuation of the norm (or bundle of norms), as well as a social payoff that varies with the fraction of others who comply with the norm (or bundle of norms). This model extends the framework of \citet{Brock_Durlauf_Restud_2001} from binary to bundled choice and formalizes two channels through which complementarities and substitutabilities can emerge among norms. The first channel is \textit{technological}, which captures the degree to which the intrinsic utility of adhering to both norms exceeds the sum of the intrinsic utilities of adhering to each separately. This mirrors the standard notion of complementarity between consumption goods \citep[e.g.,][]{Gentzkow_AER_2007}. The second channel is \textit{social}, capturing whether the returns to social conformity are super-additive or sub-additive in joint adoption. The former, for instance, may correspond to a setting where adherents of both norms impose sanctions on non-adopters that exceed the sum of the sanctions that single-norm adherents impose independently on non-adopters. That is, we assume that individuals can choose among four options: norm $A$ only,  norm $B$ only, both norms $AB$, or  neither $\emptyset$. Each choice generates two kinds of utility: an intrinsic payoff from the choice itself and a social payoff from conforming to what others do. As described earlier, the former reflects the technological channel, while the latter reflects the social channel. 

\medskip
Our main theoretical result shows that, at a stable equilibrium, whether two norms are complements, substitutes, or independent depends on the joint direction of these two forces. Two norms are complements under two conditions. First, joint adoption is intrinsically more valuable, which is captured by the  utility of joint adoption exceeding the sum of individual utilities. Second, because of the equilibrium effects, either the social return to joint conformity is at least as large as the sum of the separate conformity returns, or intrinsic payoffs are exactly additive but social returns are strictly super-additive. Conversely, two norms are substitutes when joint adoption is intrinsically less valuable and the social return to joint conformity is weakly sub-additive,  or when intrinsic payoffs are exactly additive but social returns are strictly sub-additive. The norms are independent if both intrinsic and social payoffs from joint adoption are exactly additive. This result nests the characterization in Gentzkow (2007) as the special case in which social spillovers are absent.

We take a dynamic version of the model to the data by estimating it using repeated cross-sections from various rounds of the Demographic and Health Surveys (DHS) in Sierra Leone (2008-2019) and Nigeria (2003--2024). For identification, we exploit the imposition of child marriage bans in the two countries, which we use as an exogenous shifter in the utility of child marriage conditional on covariates, fixed effects, and lagged adoption shares.\footnote{For Nigeria, we additionally use exposure to droughts during teenage years, which has been shown to increase incentives for child marriage in societies practicing bride price \citep{corno2020age}. We do not use droughts in Sierra Leone because they turn out not to be predictive of child marriage in that context.} Our key identifying assumption is that the ban shifts the baseline utility of child marriage, but not the utility of the other norms, conditional on included covariates, group fixed effects, and lagged adoption shares.

Our main empirical findings can be summarized as follows. We start by focusing on the bundle ``female genital cutting, child marriage" (FGC, CMA) in Sierra Leone. We find positive social multipliers  for both norms, with the size of multiplier being much larger for FGC---consistent with the fact that FGC is a highly coordinated practice, performed as part of a secret society initiation (``Bondo''). Our estimates suggest that FGC and CMA are strong \textit{complements} in Sierra Leone, which aligns with the fact that Bondo initiation happens during adolescence and is considered a segway into (early) marriage.  When we perform the analogous exercise in Nigeria, we find that the degree of complementarity, albeit positive, is significantly smaller than in Sierra Leone.\footnote{This is consistent with the fact that the decision to comply with FGC and CMA is temporally disjoint in Nigeria, with FGC being performed at much younger ages.}


We also consider two policy counterfactuals in Sierra Leone that reduce, respectively, (i) the intrinsic utility, or (ii) the social returns, of a given norm. We find that policies of type-(i) are very effective at reducing the prevalence of \textit{both} norms when they target the intrinsic utility of FGC, and much less so when they target CMA. Type-(ii) policies are also effective when they take the form of reducing the FGC social multiplier: they reduce both FGC-only and joint FGC-CMA adoption, consistent with the two norms being complements. Indeed, our simulations suggest that policy interventions that target social spillovers can be as effective as intervening on intrinsic utilities.

Turning to the bundle ``polygyny, child marriage" (PGY, CMA), we start by focusing on Nigeria. As in the previous cases, we estimate positive and significant social multipliers  for both norms. However, we find statistically insignificant technological complementarities and sub-linear social returns, in which case our characterization result implies that the two norms are \textit{substitutes}.\footnote{We perform the analogous exercise for Sierra Leone in the Appendix, finding similar results (i.e., the two norms are substitutes).}  Because we focus on the decisions that families make for women near the time of marriage, we classify individuals as being in a polygynous union only if the wife is the second or later wife. In this context, the relatively lower status enjoyed by non-first wives likely conflicts with the incentives parents have to marry their daughters as children. As a consequence, the impact of counterfactual policy interventions is smaller than in the case of FGC and CMA in Sierra Leone. With insignificant technological complementarities and decreasing social returns from joint adoption, the equilibrium spillovers that arise from intervening on any one norm are limited.  


\paragraph{Related Literature} Our paper contributes to several strands of literature. The first is the literature on discrete choice models with social interactions \citep[e.g.,][]{Brock_Durlauf_Restud_2001, Gulesci_etal_AER_2025, bhattacharya2024demand, jackson2007diffusion, guiteras2019demand, kline2020econometric}. Relative to this literature, we allow individuals to choose \textit{bundles}, modeling social norms as actions that are not mutually exclusive. This has the key implication that we can model returns from joint adoption and estimate complementarities and substitutabilities among norms.\footnote{Because we investigate the implications of complementarities in equilibrium, our results can be viewed as characterizing what happens to the equilibrium of an evolutionary game where agents act according to the best response dynamic \citep[see, e.g.,][]{sandholm_2010}.}

A second strand is the microeconomics literature that estimates complementarity or substitutability among consumption goods  \citep[e.g.,][]{Gentzkow_AER_2007, goolsbee2004consumer, berry2014structural, athey1998empirical}. We contribute to this literature by extending the notion of complementarity/substitutability to settings with social spillovers, allowing us to decompose complementarity into technological and social terms. We consider this extension to be of value well beyond the analysis of social norms: economic models of individual behavior increasingly recognize the importance of ``social payoffs" when studying individual consumption decisions, for example in the context of status goods and social media consumption \citep[e.g.,][]{imas2024superiority, bursztyn2018status, bursztyn_socialmedia2025}.

Finally, our work speaks to the literature on the persistence and change of gender norms \citep[for recent reviews, see, e.g.,][]{LaFerrara&Yanagizawa_2026, Field_etal2026}. By focusing on bundles, we shed light on sources of persistence or change related to norms \textit{other} than the one being studied. This opens the way for policy solutions that may target adjacent complementary norms, either to amplify the effects or to speed up the transition.

The remainder of the paper is organized as follows. Section 2 presents the model and derives conditions to characterize complements and substitutes in our setting. Section 3 develops the dynamic model specification, discusses identification, and illustrates the policy counterfactual methodology. Section 4 provides institutional background on the gender norms studied and their cultural significance in Sierra Leone and Nigeria. Section 5 describes the data and presents some descriptive analysis. Section 6 contains our main empirical results and the policy counterfactuals, and Section 7 concludes.


\section{Model} \label{sec:model_main}

Our model focuses on two non-exclusive norms $A$ and $B$ in a setting of static choice. The decision maker selects a vector $v \in \mathcal{V}\equiv\{\emptyset, A, B, AB\}$, where $A$ corresponds to adopting only norm $A$ (resp.\ $B$), $AB$ corresponds to adopting both norms, and $\emptyset$ is the outside option. This vector representation explicitly models the multidimensional nature of norm adherence: the decision to adhere to a norm is not a binary decision between the norm and the outside option, and it allows for $A$ and $B$ to be non-exclusive. 

Let $p \equiv (p_{A},p_{B},p_{AB},p_{\emptyset})$ denote the proportion of the population adhering to each norm vector such that  $\sum_{v \in \mathcal{V}} p_v = 1$ and $p_v \geq 0$ for all $v \in \mathcal{V}$. Define the prevalence of each norm as
\[
Q_{A}(p)\equiv p_{A}+p_{AB},\qquad Q_{B}(p)\equiv p_{B}+p_{AB}.
\]

\subsection{Utilities} \label{sec:baseline}

We model the (latent) utility of choosing each norm $j\in\{A,B\}$ using an intrinsic term $\delta_j$, a social term, and an idiosyncratic shock $z_j$. The corresponding utilities for each norm vector are given by $(u_{\emptyset},u_A,u_B,u_{AB})$. As discussed below, the social term represents sanctions between different-type agents, or, equivalently, returns from conformity.

\paragraph{Sanctions/conformity representations.}
Following \cite{Brock_Durlauf_Restud_2001} and  \cite{Gulesci_etal_AER_2025}, we assume that disagreement over a single norm dimension $j \in \{A,B\}$ induces a non-conformity penalty $s_j \in \mathbb{R}_+$. Specifically, adherents of $v \in \{A, \emptyset\}$ (resp. $v \in \{B, \emptyset\}$)  incur the penalty $s_B$ (resp. $s_A$)  when interacting with adherents of $v' = B$ (resp. $v' = A$).\footnote{Such costs may arise, for instance, from shunning or ostracism \citep[e.g.,][]{shell2011dynamics}}  When social interactions along each dimension are separable, adherents of $v=  AB$ penalize $A$-adherents by $s_B$, $B$-adherents by $s_A$, and $\emptyset$-adherents by $s_{AB} = s_A + s_B$. We refer to this case as \emph{linear sanctions:} the total penalty imposed by $AB$-adherents equals the sum of the dimension-specific sanctions. 

However, in many settings, the sanctions may not be linear. Consider, for example, a context in which both $A$ and $B$ are restrictive gender norms associated with purity and chastity. Violating both norms simultaneously may provoke a disproportionately severe response among the most ``traditional" members of society (the $AB$-adherents) and sanctions on $\emptyset$-adherents will be \emph{super-linear}, such that $s_{AB} > s_A + s_B$. Conversely, we say that sanctions are \emph{sub-linear} if $s_{AB} < s_A + s_B$. We accommodate both forms of non-linearities in the most general version of the model.

Table ~\ref{tab:non_linear_sanctions} (left panels) describes the social penalty of choosing $v \in \mathcal{V}$ as a function of the decision maker’s vector (row) and the peer’s vector (column). The top  matrix corresponds to the case of linear sanctions, while the bottom matrix allows for non-linear sanctions. The right panels reparameterize the model, interpreting the social utility term as benefits from conformity instead. These parameterizations are without loss of generality (they only differ by a normalization $u_{\emptyset}=0$).

\begin{table}[ht]
 \caption{Social interactions matrices with linear  and non-linear sanctions.} 
  \centering
  \begin{subtable}[t]{0.48\textwidth}
    \centering
    \begin{tabular}{lcccc}
      \toprule
       & $\emptyset$ & A & B & AB \\
      \midrule
      $\emptyset$ & $0$  & $-s_A$    & $-s_B$  & $-s_A - s_B$   \\
      A           & $0$  & $0$     & $-s_B$  & $-s_B$    \\
      B           & $0$  & $-s_A$    & $0$   & $-s_A$    \\
      AB          & $0$  & $0$     & $0$   & $0$     \\
      \bottomrule
    \end{tabular}
    \subcaption{Sanctions representation: linear}
  \end{subtable}\hfill
  \begin{subtable}[t]{0.48\textwidth}
    \centering
    \begin{tabular}{lcccc}
      \toprule
       & $\emptyset$ & A & B & AB \\
      \midrule
      $\emptyset$ & $0$ & $0$ & $0$ & $0$ \\
      A           & $0$ & $s_A$ & $0$ & $s_A$ \\
      B           & $0$ & $0$ & $s_B$ & $s_B$ \\
      AB          & $0$ & $s_A$ & $s_B$ & $s_A + s_B$ \\
      \bottomrule
    \end{tabular}
    \subcaption{Conformity representation: linear}
  \end{subtable} 
  
  \vspace{10 mm} 
  \begin{subtable}[t]{0.48\textwidth}
    \centering
    \begin{tabular}{lcccc}
      \toprule
       & $\emptyset$ & A & B & AB \\
      \midrule
      $\emptyset$ & $0$  & $-s_A$    & $-s_B$  & $-s_{AB} $   \\
      A           & $0$  & $0$     & $-s_B$  & $s_A - s_{AB}$    \\
      B           & $0$  & $-s_A$    & $0$   & $s_B - s_{AB}$    \\
      AB          & $0$  & $0$     & $0$   & $0$     \\
      \bottomrule
    \end{tabular}
    \subcaption{Sanctions representation: non-linear}
  \end{subtable}\hfill
  \begin{subtable}[t]{0.48\textwidth}
    \centering
    \begin{tabular}{lcccc}
      \toprule
       & $\emptyset$ & A & B & AB \\
      \midrule
      $\emptyset$ & 0  & $0$    & $0$  & $0$   \\
      A           & $0$  & $s_A$    & $0$& $s_A$  \\
      B           & $0$  & $0$   & $s_B$  & $s_B$  \\
      AB          & $0$  & $s_A$    & $s_B$  & $s_{AB}$   \\
      \bottomrule
    \end{tabular}
    \subcaption{Conformity representation: non-linear }
  \end{subtable}
\label{tab:non_linear_sanctions}

 \begin{tablenotes}[flushleft]
   \footnotesize 
    \item \emph{Note:} Rows represent the decision maker's choice and columns the peer's choice.
  \end{tablenotes}
  
\end{table}

\medskip

\paragraph{Utility of $v \in \{A,B,\emptyset\}$} We proceed with the conformity representation and normalize $u_\emptyset = 0$. With this parameterization, the overall utility of  adopting each $v \in \{A,B\}$ are, respectively,
\begin{equation} \label{eqn:utility_basics}
    u_{A}(p) = \delta_A + s_A \,\underbrace{(p_A + p_{AB})}_{= Q_A} + z_A, \qquad 
    u_{B}(p) = \delta_B + s_B\,\underbrace{(p_B + p_{AB})}_{= Q_B} + z_B. 
\end{equation} 
Assuming that each agent interacts with a single peer randomly drawn from the population, the social payoff term reflects the agent's expected social payoff.  
We impose the following regularity assumption on the idiosyncratic shocks $z_A, z_B$.

\begin{assumption}[Shock distribution]\label{ass:z_density}
Let $z \equiv (z_{A},z_{B})^\top \sim R$ with density $r(z)$ that is continuously differentiable and strictly positive for all $z \in \mathbb{R}^2$. The density $r(z)$ is not a function of $(\delta_A, \delta_B, s_A, s_B, p_A, p_B, p_{AB})$, and $\mathbb{E}[z_A] = \mathbb{E}[z_B] = 0$.  Assume, in addition, that there exist constants \(C<\infty\) and \(\eta>0\) such that
$ 
r(z)\le C(1+\|z\|)^{-2-\eta}$ for all $z\in\mathbb R^2.$ 
\end{assumption}

\noindent Assumption \ref{ass:z_density} imposes smoothness and a mild regularity condition on the tails of the density. It holds, for example, if we assume $(z_A,z_B)\sim N(\mathbf{0},\Sigma)$, with  $\Sigma_{AA}=1$, $\Sigma_{BB} \in (0,\infty)$, and $\rho \equiv \frac{\Sigma_{AB}}{\sqrt{\Sigma_{BB}}} \in (-1,1)$, which coincides with the distribution of shocks in \cite{Gentzkow_AER_2007}. 

\paragraph{Utility of $AB$.} Next, we specify the utility of adhering to the norm vector $AB$. Following \cite{Gentzkow_AER_2007}, we allow for the non-social (technological) complementarities in addition to social complementarities defined above: 
\begin{equation} \label{eqn:AB_1}  
\begin{aligned} 
    u_{AB}(p) & \;=\;  \underbrace{\delta_A + \delta_B + \Gamma}_{\text{intrinsic utility of AB}} + \underbrace{s_A p_A+ s_B p_B + s_{AB}\,p_{AB}}_{\text{social spillovers for AB}} + \underbrace{(z_A+z_B)}_{\text{shocks}} \\
             &= \underbrace{u_A(p) + u_B(p) + \Gamma}_{\text{utility $AB$ with linear spillovers}} + \underbrace{p_{AB}(s_{AB} - s_A - s_B)}_{\text{correction for non-linear spillovers}}.
             \end{aligned} 
\end{equation}

Since the non-social and non-idiosyncratic component of utility is given by $\delta_A+ \delta_B + \Gamma$, the intrinsic utility of adhering to both norms is super-linear if $\Gamma > 0$ and sub-linear if $\Gamma < 0$.
This specification differs from standard social-interaction models such as that of  \cite{Gulesci_etal_AER_2025}, where harmful norms with low baseline value can persist due to coordination based on conformity. In our setting, high $AB$ adoption can arise even when $s_A, s_B, s_{AB} = 0$ provided $\Gamma\gg 0$. 

Our framework also generalizes that of \cite{Gentzkow_AER_2007} by introducing the social component of utility (both in utilities with a single norm $A$ or $B$ and in utilities for $AB$) to model sanctions and/or returns from conformity. Together, our model allows both technological and social complementarities to shape the choices of agents.

To gain further insight in the interpretation of $\Gamma$, suppose that we are in the linear sanctions regime with $s_{AB} = s_A + s_B$, as in the top panel of Table \ref{tab:non_linear_sanctions}(a) and \ref{tab:non_linear_sanctions}(b).  In this case, 
$u_{AB}(p) = u_A(p) + u_B(p) + \Gamma$, and  $\Gamma$ can be defined as the difference between the utility of adopting $AB$ and the utilities of adopting $A$ and $B$, i.e., $\Gamma \equiv u_{AB}(p) - u_A(p) - u_B(p)$. \cite{Gentzkow_AER_2007} thus interprets $\Gamma$ as capturing ``technological" complementarities. However, even when social effects are linear, they play an important role in determining the equilibrium dynamics (and consequently the interpretation of the parameters) of our model that we study below.  




\subsection{Equilibrium analysis}
\label{sec:equil}

We now provide an equilibrium analysis of the system induced by our model. For each $v \in \mathcal{V}$, let the proportion of norm vector adherents be defined as
\begin{equation} \label{eqn:p_star1}
    \tilde{p}_{v}(p) \;=\; \mathbb{E}_{z}\left[ \mathbbm{1}\!\left\{ u_{v}(p) \,\ge\, \max_{v'\in\mathcal{V}} u_{v'}(p) \right\} \, \right],
\end{equation} 
where  $\mathbb{E}_{z}[\cdot]$ indicates the expectation operator with respect to idiosyncratic shocks $z\sim R$. In other words, $\tilde{p}_v(p)$ is the probability that an individual randomly drawn from the population chooses the vector $v \in \mathcal{V}$ when the population proportion of norm adherents is given by $p$.

Any $p^* \equiv (p_A^*, p_B^*,p_{AB}^*)$ satisfying the following  fixed point equations constitutes an equilibrium: 
\begin{equation} \label{eqn:q_star}
p_j^* = \tilde{p}_j(p^*), \qquad \forall j \in \{A,B, AB\}. 
\end{equation} 
Under Assumption \ref{ass:z_density}, existence of $p^*$ follows directly from the Brouwer fixed-point theorem since $\tilde{p}$ is a continuous map from the simplex onto itself. However, uniqueness of the equilibrium is not guaranteed.

\begin{definition}[Stability] \label{ass:stability2} Denote in vector form $\tilde{p}(p) = \Big(\tilde p_A(p), \tilde p_B(p), \tilde p_{AB} (p)\Big)^\top$ and let 
\begin{equation} \label{eqn:F_tilde}
\tilde{\Lambda}(p^*)  = \frac{\partial \tilde{p}(p)}{\partial p}\Big|_{p = p^*}.
\end{equation} 
Given equilibrium quantities $p^*$ such that $p^* =  \tilde{p}(p^*)$, an equilibrium is stable if the absolute value of all the eigenvalues of $\tilde{\Lambda}(p^*)$ is strictly below one.  
\end{definition}

Intuitively, a stable equilibrium is a locally attracting fixed point of the dynamic system
\(p_{t+1}=\tilde p(p_t)\): if the initial condition is sufficiently close to that equilibrium,
then the sequence \(p_t\) converges to it. This local condition does not rule out the
existence of other equilibria. The existence of stable equilibria typically requires restrictions on the magnitude of social spillovers (see Remark \ref{rem:Lambda-interpretation} at the end of this section for a formal discussion). We will focus on stable equilibria in the remaining discussion, as their robustness to small perturbations in the baseline utilities of each norm forms the basis of our definition of complements and substitutes in equilibrium.

\begin{definition}[Complements and substitutes in stable equilibria] \label{defn:Q_star}
In any stable equilibrium $p^*$ satisfying Equation \eqref{eqn:q_star}, 
the norms $A$ and $B$ are complements if $\frac{\partial Q_A(p^*)}{\partial \delta_B} > 0$, independent if $\frac{\partial Q_A(p^*)}{\partial \delta_B} =0$, and substitutes if $\frac{\partial Q_A(p^*)}{\partial \delta_B} < 0$. 
\end{definition}

Definition \ref{defn:Q_star} states that two norms are complements (resp. substitutes) if the overall change in the adoption of one norm  is positive (resp. negative) when the baseline utility of the other norm is perturbed. This reflects the usual notion of complements and substitutes between goods defined at the level of cross-price derivatives.  The following theorem provides a formal characterization of complementary effects through the model's parameters. 

\begin{theorem} \label{prop_gamma_fixed_point_extended}
Suppose Assumption \ref{ass:z_density} holds, $s_A, s_B, s_{AB} \in (0,\infty)$, and $|\Gamma| < \infty$.
    \begin{align*} 
    \small 
     & \text{If }    \begin{cases}
             (\Gamma > 0 \text{ and } s_{AB} \ge s_A + s_B) \text{ or } (\Gamma = 0 \text{ and } s_{AB} > s_A + s_B) \\
             (\Gamma  = 0 \text{ and } s_{AB} = s_A + s_B)  \\
             (\Gamma < 0 \text{ and }  s_{AB} \le s_A + s_B) \text{ or }  (\Gamma = 0 \text{ and }  s_{AB} < s_A + s_B)
        \end{cases},   \nonumber \\ 
      & \hspace{180pt}  \text{ then norms A and B are } \quad  \begin{cases}
            \text{complements} \\
            \text{independent} \\
            \text{substitutes}
        \end{cases}  \nonumber
    \end{align*}
    at every stable equilibrium $p^*$.
\end{theorem}
\begin{proof}
    See Supplemental Appendix \ref{proof_gamma_fixed_point_extended}.
\end{proof}

Theorem~\ref{prop_gamma_fixed_point_extended} shows that with linear social interactions ($s_{AB}=s_A+s_B$), whether norms are complements or substitutes is determined solely by the sign of $\Gamma$. This result nests the characterization of $\Gamma$ in Proposition 1 of  \cite{Gentzkow_AER_2007} as a special case where $s_A= s_B =s_{AB} =0$, i.e., the case without social spillovers. With social interactions, perturbations in the baseline utilities of each norm also operate through equilibrium adoption shares; the derivation differs due to the presence of equilibrium dynamics in our context.

With non-linear social returns, both the technology $\Gamma$ and social terms $s_{AB} - s_A - s_B$ matter; clear comparative statics require that they share the same sign. 
The intuition mirrors the income–substitution decomposition in consumer theory (although the formalisms and derivations differ). When the technological and social terms reinforce each other, the effect of shifting one norm’s utility on the other is unambiguous. When they conflict, the net effect can flip. For example, even if $A$ and $B$ are technological substitutes ($\Gamma<0$), a reduction in $\delta_B$ can lower adherence to $A$ if social complementarity is strong enough ($s_{AB}\gg s_A + s_B$): the policy reduces $p_B$ and $p_{AB}$, weakening the social component of $u_{AB}$ and thereby reducing the appeal of co-adoption.

Finally, we remark that in this model, complementarities may arise even in the absence of technological ones. That is, it is possible that $\Gamma = 0$, but norms are not independent due to the non-linear sanction component $s_{AB}$. We  can thus interpret non-linear $s_{AB}$ as ``social complementarities'' in the absence of technological effects.

\begin{remark}[Properties of stable equilibria]\label{rem:Lambda-interpretation} To gain intuition on the properties of a stable equilibrium, consider 
\[
\begin{pmatrix}
\bar u_A(p)\\ \bar u_B(p)\\ \bar u_{AB}(p)
\end{pmatrix}
\equiv 
\begin{pmatrix}
\delta_A+s_A(p_A+p_{AB})\\
\delta_B+s_B(p_B+p_{AB})\\
\delta_A+\delta_B+\Gamma+s_Ap_A+s_Bp_B+s_{AB}p_{AB}
\end{pmatrix}.
\]
These define the utilities without idiosyncratic shocks. 
Under Assumption~\ref{ass:z_density}, it follows that 
$\tilde p(p)=F(\bar u(p))$ for some continuously differentiable $F:\mathbb R^3\to[0,1]^3$. Hence, at any $p^*$,
\[
\tilde\Lambda(p^*)
=\frac{\partial \tilde p(p)}{\partial p}\Big|_{p=p^*}
=\underbrace{\frac{\partial F(\bar u)}{\partial  \bar u}\Big|_{\bar u=\bar u(p^*)}}_{G(p^*)}
\underbrace{\frac{\partial \bar u(p)}{\partial p}\Big|_{p=p^*}}_{\displaystyle B}
\quad\text{where}\quad
B=
\begin{pmatrix}
s_A & 0   & s_A\\
0   & s_B & s_B\\
s_A & s_B & s_{AB}
\end{pmatrix}.
\]

The matrix $B$ depends only on the social penalties $(s_A,s_B,s_{AB})$.
Therefore, a locally stable equilibrium implicitly imposes restrictions on the magnitude of social spillovers through restrictions on the matrix $B$ and $G(p^*)$. Intuitively, it assumes that social spillovers are not ``explosive''. In particular, 
the stability requirement on $\tilde{\Lambda}(p^*)$ ensures that the best-response dynamics
$p_{t+1}=\tilde p(p_t)$ is locally contracting around $p^*$ at $t \rightarrow \infty$. \qed
\end{remark}

\section{Bringing the model to the data}

We discuss identification and estimation of the model using repeated cross-sections, while deferring specifics on our empirical implementation to the next section.

\subsection{Dynamic model specification}

In each period $t$, we sample repeated cross sections of non-overlapping $n_t$ individuals, each period indexed by $i \in \{1, \dots, n_t\}$. All individuals in the population are organized into groups that are time-invariant and non-overlapping, with $g(i)$ denoting the group of individual $i$. In our application, these groups capture the ethnic and geographic proximity between individuals. During the observed time window running from $t = 1$ to $t = T$, we consider dynamic equilibrium paths where each individual $i$ at time $t$, with observable characteristics $X_{it} \in \mathbb{R}^p$, has utilities 
\begin{equation} \label{eqn:dynamic1}
\begin{aligned} 
u_{Ait} & = \tilde{\delta}_{Ait} + s_A (p_{A(t-1)}^{g(i)} + p_{AB(t-1)}^{g(i)}) + \tilde{z}_{Ait}, \\ 
u_{Bit} & = \tilde{\delta}_{Bit} + s_B (p_{B(t-1)}^{g(i)} + p_{AB(t-1)}^{g(i)}) + \tilde{z}_{Bit}, \\  
u_{ABit} & = u_{Ait} + u_{Bit} + \Gamma + (s_{AB} - s_A - s_B)\, p_{AB(t-1)}^{g(i)},  
\end{aligned} \quad \begin{bmatrix} 
\tilde{z}_{Ait} \\ 
\tilde{z}_{Bit}
\end{bmatrix} \Big| X_{it} \sim \mathcal{N}\left(0, \begin{bmatrix}
1 &  \rho \\ 
\rho & 1  
\end{bmatrix} \right),
\end{equation} 
where $\tilde{\delta}_{jit} = X_{it}^\top \beta_j$ and  $p_{jt}^g$ is the probability that an individual adopts the norm $j$ at time $t$ within the group $g$. The model absorbs the time-varying $X_{it}$ through mean changes in the baseline utilities, with $\tilde{\delta}_{jit} = X_{it}^\top \beta_j$. It assumes that the unobserved shocks $\tilde{z}_{it}$ are independent of the covariates. 
We standardize $\tilde{z}_{jit}$'s variance to be equal to one assuming homoskedasticity and set $\rho \in (-1,1)$ to denote their correlation coefficient.\footnote{In practice, some of these assumptions can be relaxed by e.g., accounting for additional dependence with clustered standard errors, or allowing for forms of heteroskedastic noise.}

We let utilities be functions of proportions of adherents in the previous period. This assumption simplifies estimation, while maintaining the same interpretation of the model where an equilibrium is defined as the point of attraction of the dynamic system.  The model assumes that social interactions occur at the level of groups $g(i)$, while assuming common parameters between the different groups.\footnote{Extensions to heterogeneous parameters are possible but omitted, with heterogeneous parameters the same exercise can be carried over separately for each group at the expense of a smaller sample size.}

\paragraph{Independence under repeated cross-sections} In our application, we observe repeated and non-overlapping cross-sections of surveyed individuals. By taking a sampling based perspective where individuals are independently drawn from a large (infinite) population over different periods, we treat the vector 
$$
(z_{Ai,1}, \cdots, z_{Ai,T}, z_{Bi,1}, \cdots, z_{Bi,T}, X_{i1}, \cdots, X_{iT})
$$
as independent across different individuals, though not necessarily across time for the same individual, with $X_{it} \sim F_X(t)$ potentially non-stationary over time. We therefore implicitly condition on aggregate shocks (which may affect the non-stationary distribution of covariates $F_X(t)$), such as a change in state-level legislation. On the other hand, we assume that the distribution of idiosyncratic shocks $\tilde{z}$ is stationary over time, therefore implicitly requiring the exogeneity of $\tilde{z}$.

\paragraph{Stable equilibria} We interpret stable equilibria as the point of attraction of a dynamic system, and study the set of stable equilibria corresponding to the long-run outcome of a dynamic system as $t \rightarrow \infty$ where, for simplicity, we fix the distribution of covariates in the last period, so that $F_X(t) = F_X(T)$ for all $t \ge T$.  
The observed adoption shares and covariates' distribution determine equilibrium selection by providing the initial conditions for the long-run dynamic system. Formally, this implies that, taking $t \rightarrow \infty$ in the dynamic system, we can map our model with covariates to a set of stable equilibria under our model in Section \ref{sec:model_main}, where, for each group $g$,  the corresponding $z$ in Assumption \ref{ass:z_density} coincides with the vector $\tilde{z}_{iT} + X_{iT}^\top \beta - \mathbb{E}[X_{iT}^\top \beta]$, $\delta_j = \mathbb{E}[X_{iT}^\top\beta_j]$ with $g(i) = g$, and assuming $\Gamma$ and $s$ are homogeneous between different groups $g(i)$. 

\subsection{Identifying restrictions}

Without exclusion restrictions, the model is not necessarily identified. We consider exclusion restrictions that induce shifts in the utility of one norm, but not the other. Let $\mathcal{I}_A$ and $\mathcal{I}_B$ denote index sets for columns of $X_{it} \in \mathbb{R}^p$ (of fixed dimension over time-periods to simplify notation), and let $\beta_j^{(\mathcal{I})}$ denote the subvector of $\beta_j$ corresponding to indices in $\mathcal{I}$. Define $\mathcal{I}_A, \mathcal{I}_B$ the set of indeces such that 
\begin{align} 
\beta_A^{(\mathcal{I}_A)} = 0, \quad \beta_B^{(\mathcal{I}_B)} = 0; \label{eqn:validity}  \\  \label{eqn:relevance}
\beta_B^{(\mathcal{I}_A)} \neq 0, \quad \beta_A^{(\mathcal{I}_B)} \neq 0. 
\end{align}

As in \citet{Gentzkow_AER_2007}, the  first condition imposes that certain covariates do not shift the utility of one of the two norms, similar to standard validity restrictions in instrumental variables. The second requires these covariates to shift the utility for the other norm, similar to standard relevance conditions.  It is not necessary that both sets are non-empty: one of them may be empty (e.g., $\mathcal{I}_A = \emptyset$), as long as the other is not. Identification requires that the Fisher information criterion of the profiled likelihood for the relevant parameters of interest ($\Gamma, s_A, s_B, s_{AB}$) is invertible under Equation \eqref{eqn:validity}. In our specification, we will use country-level regulations that affect the utility of CMA but do not change the regulations and utility of the other norm as our main exclusion restriction.\footnote{In the different context of panel data, other forms of identification are also possible.}  Supplemental Appendix \ref{sec:stacked_se} provides additional intuition behind identification.

\subsection{Estimation and inference}

Define $\theta = (\beta_A,\beta_B,s_A,s_B,s_{AB},\Gamma,\rho)$  as the vector of parameters to be estimated, and let
\[
\hat{p}_{it}^{g(i)}(\theta) \equiv \big( \hat{p}_{it}^{g(i)}(v;\theta) \big)_{v \in \{\emptyset, A, B, AB\}}
\]
denote the vector of model-implied choice probabilities for individual $i$ in group $g(i)$ at time $t$. Let $v_{it}$ denote the norm chosen by individual $i$ at time $t$, which is observed in the data. We estimate $\theta$ by maximizing the log-likelihood
\[
\ell(\theta) \equiv \sum_{i,t} \sum_{v \in \{\emptyset, A, B, AB\}} 
\mathbf{1}\{ v_{it} = v \} \log\!\Big(\hat{p}_{it }^{g(i)}(v;\theta) \Big),
\]
subject to the validity restrictions in Equation~\eqref{eqn:validity}. To facilitate numerical optimization, we smooth the indicator functions that define the choice regions in the integral for $p_{it}^g$ using a sigmoid approximation,
$
\mathbf{1}\{x \geq y\} \approx \sigma_\varepsilon(x-y) \equiv \big(1 + e^{-(x-y)/\varepsilon}\big)^{-1},
$ 
and approximate the resulting Gaussian integrals by quadrature.\footnote{\cite{Gentzkow_AER_2007} rationalizes the smoothing as adding a small additional iid idiosyncratic noise to the utilities component. Here, we interpret the smoothing as a  valid approximation (as we take $\varepsilon \rightarrow 0$) that facilitates optimization. Because we use a regular $200 \times 200$ grid over $[-5,5]^2$ for $(z_{A},z_{B})$ for computing the likelihood and its gradient, we take $\epsilon = 0.05$. The likelihood gradient is calculated via JAX and maximum likelihood estimation is performed using multiple random-seeded runs of the IPOPT optimizer.} 

We conduct inference via standard inference for the methods of moments, with moments generated by the scores of the likelihood. In addition to the uncertainty generated by idiosyncratic shocks and covariates, inference must account for the additional dependence induced by using survey-based share estimates $\hat p$ in our regression (a generated-regressor problem). We present the formal derivation in Supplemental Appendix~\ref{sec:stacked_se}. 



\subsection{Policy counterfactuals} \label{subsec:policy_counterfactuals}

Given the estimated parameter and standard errors, we study policy counterfactuals by simulating shocks to the utilities (which may, for example, occur with the passage of legislation banning a certain norm) and propagating the model forward. With an abuse of notation, let $\tilde{p}_t(p; \theta)$ denote the transition function in Definition \ref{ass:stability2}, where we now explicitly index $\tilde{p}$ by the model parameters $\theta$ and time $t$. The time index corresponds to the distribution of covariates $X_{it} \sim F_X(t)$. We will fix the distribution of covariates $X_{it} \sim F_X(T)$ for all unobserved periods $t \ge T$. 

Denote by $p_t^{base}$ the model-implied shares of adopters of $(A,B,AB)$ at time $t\ge T$ under the estimated parameters $\hat{\theta}$. Denote by $p_t^{pol}$ the corresponding predicted shares under a counterfactual parameter vector $\theta^{pol}$, obtained from $\hat{\theta}$ by, for example, lowering the baseline utility of a given norm or altering spillover coefficients. Intuitively, the difference $p_t^{pol}-p_t^{base}$ summarizes the policy effect of changing the environment from $\hat{\theta}$ to $\theta^{pol}$. We formalize these definitions below.

\begin{definition}[Policy counterfactuals with utilities shocks] \label{defn:policy1} For a given observed repeated cross-sections of observations observed over $T$ periods, define a $h$-periods ahead policy counteractual as 
    \begin{equation} \label{eqn:delta}
\Delta^h(\theta^{\text{pol}}) \equiv p_{T+h}^{\text{pol}} -  p_{T+h}^{\text{base}}, \qquad h=1,2,\dots.
\end{equation} 
where 
$$  
p_{t}^{\text{pol}} = \begin{cases} 
\tilde{p}_T\big(p_{t-1}^{\text{pol}}; \theta^{\text{pol}}\big),  &  t > T \\ 
p_{t}^{\text{base}}, & t \le T 
\end{cases}, \qquad p_{t}^{\text{base}} = \begin{cases} 
 \tilde{p}_{\min\{t,T\}}(p_{t-1}^{\text{base}}; \hat{\theta}), &  t > 1 \\ 
 p_{t=1} & t = 1
 \end{cases} 
$$
with $p_{t=1}$ denoting the observed share of adopters in the first period of observations and $\hat{\theta}$ corresponding to the estimated parameter.
\end{definition}

Definition \ref{defn:policy1} defines counterfactuals where we replace $\theta$ with a counterfactual parameter vector $\theta^{\text{pol}}$ that incorporates the policy induced changes while keeping the remaining parameters at their estimated values.
We construct the policy counterfactual by iterating the model forward. 
 In our application, we consider the following two classes of policy counterfactuals:

 \paragraph{Shifts to baseline utilities} In this case $\theta^{pol}$ coincides with $\hat{\theta}$ except for intercepts $\delta_A$ and/or $\delta_B$, to which we add an additive shift capturing a shock to the baseline utilities of adopting norm $A$ and / or $B$. In practice, we calibrate the magnitude of this shift using evidence on the impact of bans or legal reforms observed in other contexts.

\paragraph{Shifts to social complementarities} To separate the role of social interactions from purely technological or preference effects, we also consider counterfactuals in which the terms of social interaction are smaller (for instance, set $s_A' < s_A$ as a social multiplier for norm $A$ under the counterfactual intervention $\theta^{pol}$). These interventions may capture information campaigns that may change the perception of social complementarities of a given norm, or weakening of social institutions traditionally responsible for norm enforcement.

\section{Application: Gender norms in West Africa}
\label{sec:norms_background}

Our framework is general and can be applied to a variety of norms in multiple domains. In our application, we will focus on a subset of gender norms that are relatively prevalent in the African context: female genital cutting, child marriage, and polygyny.

According to the World Health Organization, \textit{female genital cutting} (FGC) is ``the partial or total removal of the external female genitalia, or other injury to the female genital organs for non-medical reasons.'' The WHO estimates that over 230 million girls and women are currently cut, and over 4 million girls are at risk of cutting every year.\footnote{Source: https://www.who.int/news-room/fact-sheets/detail/female-genital-mutilation}  While the health consequences of FGC vary depending on the ``type" of cut (from type 1, or clitoridectomy, to type 3, or infibulation), severe risks have been documented both at the time of cutting (e.g., bleeding, infection) and later in life (e.g., urinary and gynecological problems, heightened risks at childbirth, etc.). 
\textit{Child marriage} is defined by the United Nations as ``any formal marriage or informal union between a child under the age of 18 and an adult or another child." About one in five girls globally marry before age 18.\footnote{Source: https://www.unicef.org/protection} 
\textit{Polygyny} is a marital arrangement where a man is permitted to have multiple wives.

\begin{figure}[!h]
\caption{Prevalence of gender norms across African countries} 
\centering

  \begin{subfigure}[t]{0.4\textwidth}
    \centering
    \includegraphics[width=\linewidth]{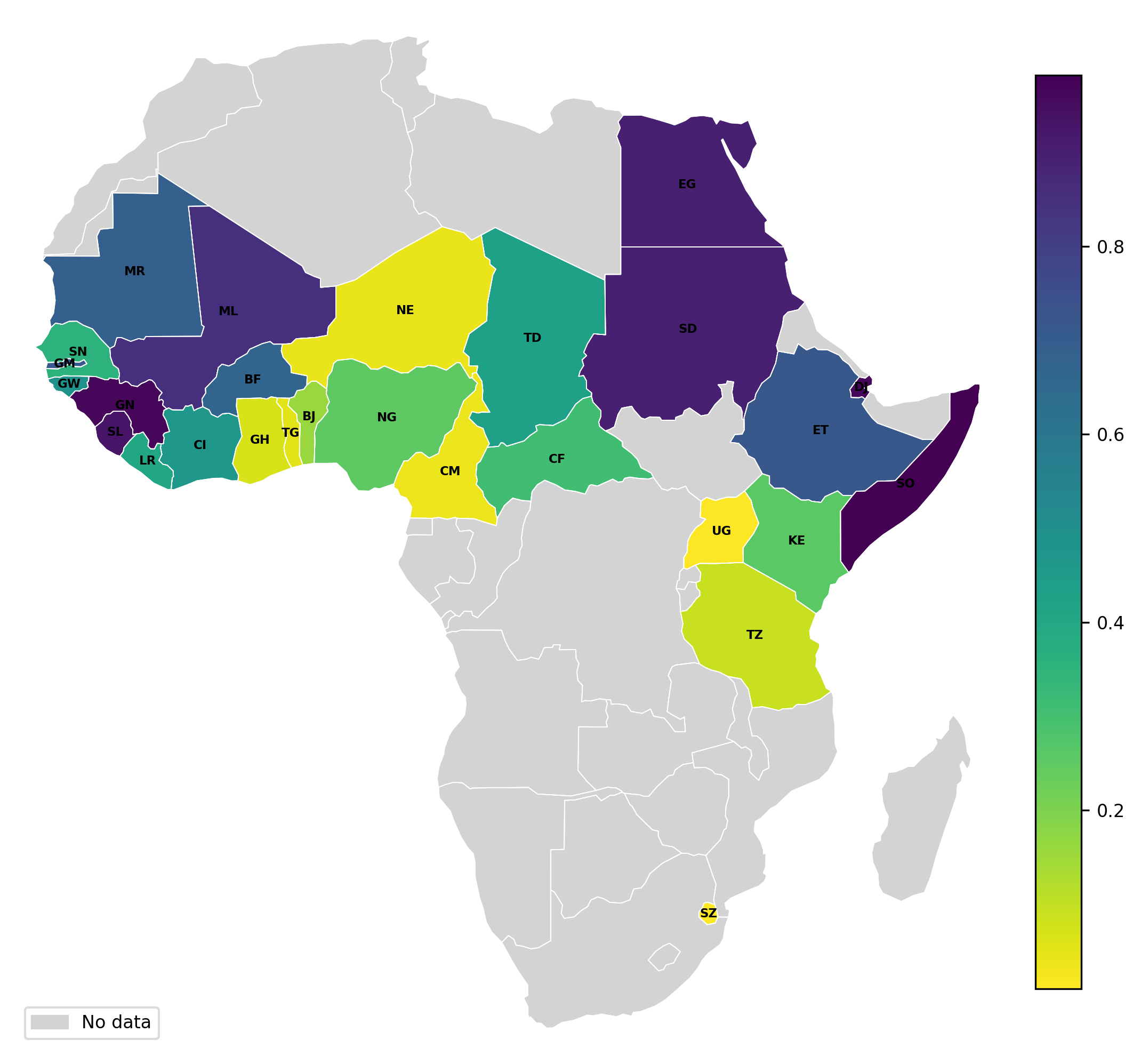}
    \captionsetup{justification=centering,singlelinecheck=false}
    \caption{Female Genital Cutting (FGC)}
  \end{subfigure}
  \hfill
  \begin{subfigure}[t]{0.4\textwidth}
    \centering
    \includegraphics[width=\linewidth]{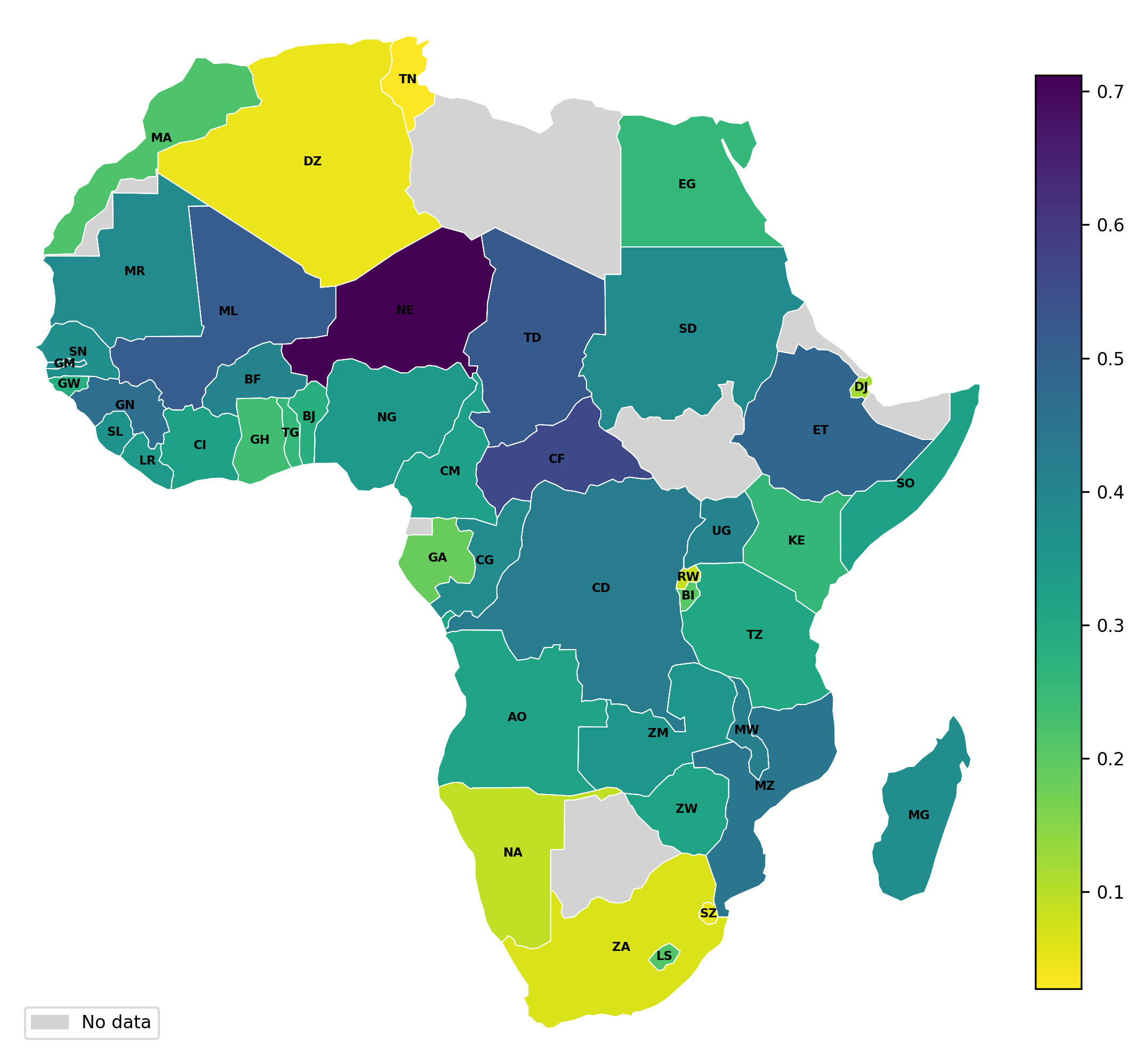}
    \captionsetup{justification=centering,singlelinecheck=false}
    \caption{Child Marriage (CMA)}
  \end{subfigure}

  \vspace{1em}

  \begin{subfigure}[t]{0.4\textwidth}
    \centering
    \includegraphics[width=\linewidth]{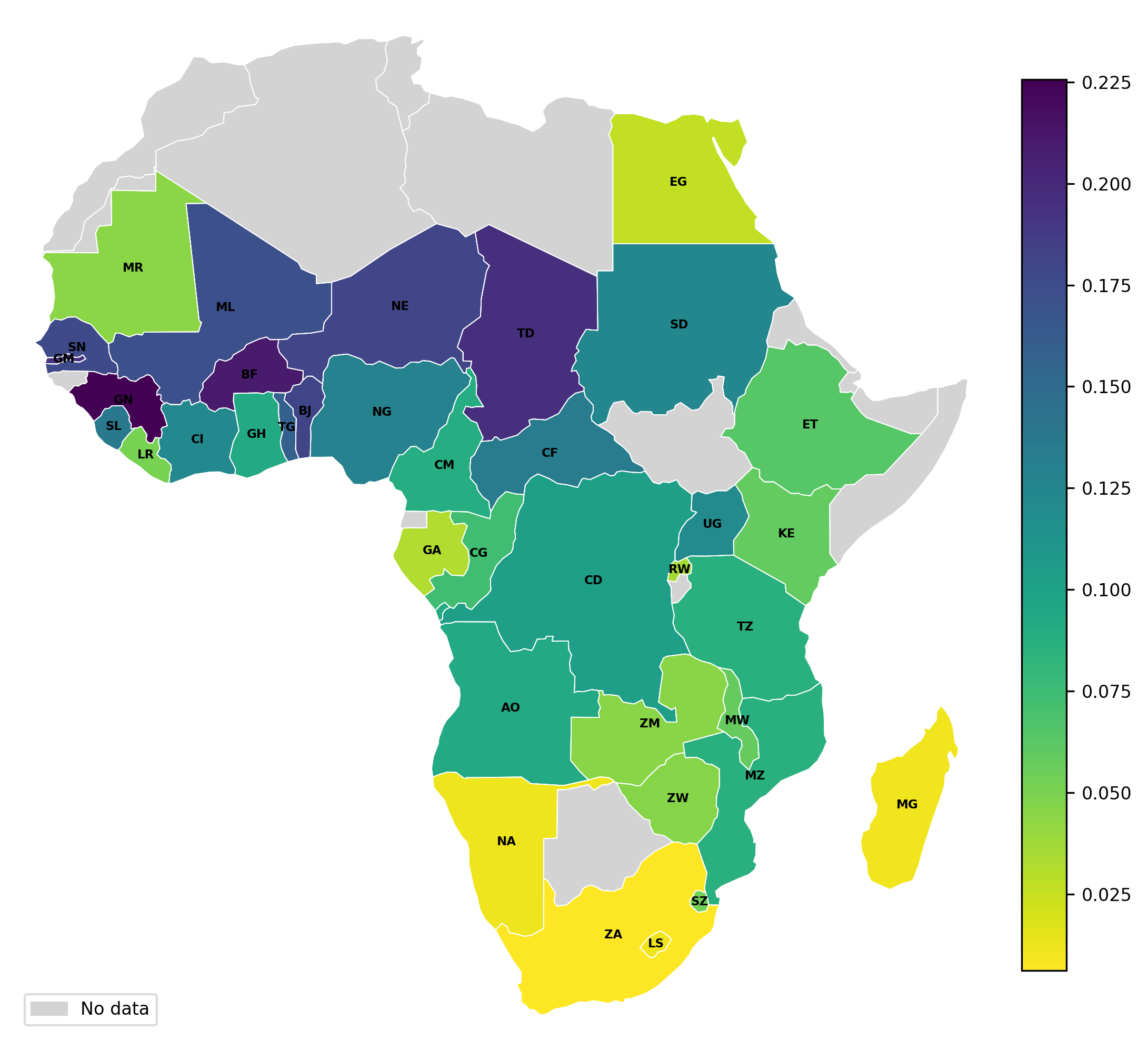}
    \captionsetup{justification=centering,singlelinecheck=false}
    \caption{Polygyny (PGY)}
  \end{subfigure}

  \begin{flushleft}
  \footnotesize \emph{Notes:} Each subfigure reports the country-level prevalence of the respective norm: (a) for FGC, (b) for CMA, and (c) for PGY. Estimates are computed by the authors using the most recent available DHS or MICS survey for each country.
  \end{flushleft}

\label{fig:heatmaps}

\end{figure}

Figure 1 displays the cross-country variation in the prevalence of these norms across the African continent, estimated using data from the Demographic and Health Surveys (DHS) and---for panel (a)---also from the Multiple Indicator Cluster Surveys (MICS).\footnote{DHS data are available at \url{https://dhsprogram.com/data/} and MICS data are available at \url{https://mics.unicef.org}.} Panel (a) shows the prevalence of FGC, panel (b) of child marriage, and panel (c) of polygyny. 
FGC rates are particularly high in West Africa (notably in Mali, Guinea and Sierra Leone), as well as in Egypt, Sudan and in the horn of Africa. Child marriage is broadly prevalent but highest in the Sahel; polygyny is concentrated in West and Central Africa.

Figure~\ref{fig:timetrend} documents the evolution over time of all three norms, plotting the prevalence of the norms among birth cohorts in 20 African countries. The series display significant co-movement in many of them: this could be suggestive of complementarity, or it could simply reflect secular trends or correlation with common determinants of the various norms. Our methodology will attempt to isolate the underlying parameters by capturing interdependence.

\begin{figure}[!h]
\caption{Prevalence of gender norms over time} 
\centering
    \includegraphics[width=0.95\linewidth]{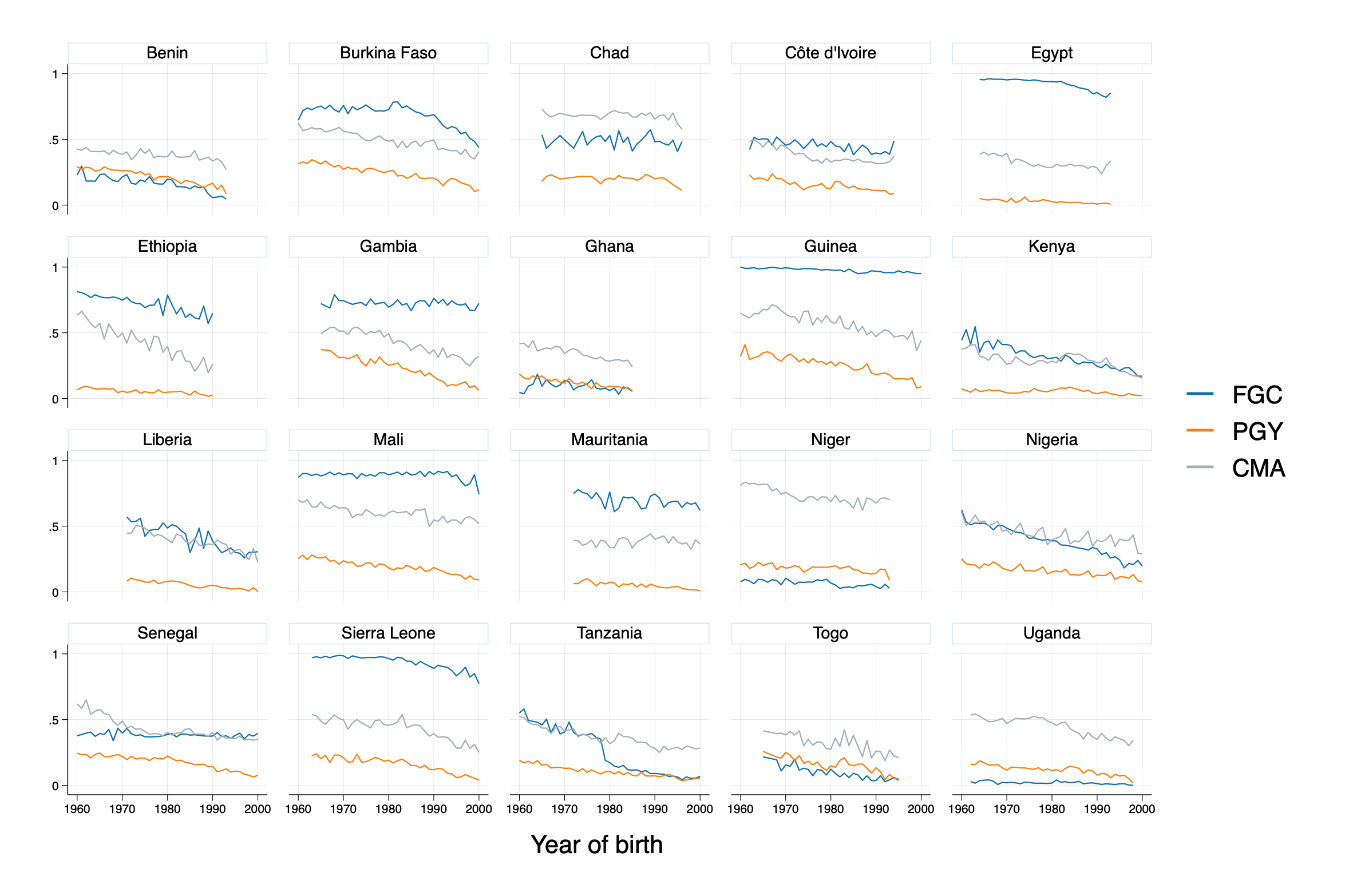}
 \begin{flushleft}
  \footnotesize \emph{Notes:} Each panel reports the share of norm adopters by birth cohort for a given country. The sample includes all available DHS waves and cohorts born between 1960 and 2000. Country–cohort cells with fewer than 100 observations are excluded.
   \end{flushleft}
\label{fig:timetrend}
\end{figure}

Since we will estimate the complementarity/substitutability of these norms in two  specific countries, Nigeria and Sierra Leone, we now briefly describe the configurations that these norms take in these two countries.

\paragraph{Female genital cutting} The timing and cultural significance of FGC differs between the two countries in our analysis. In Nigeria, FGC is typically performed in infancy and in an ``individualized" setting, in the sense that the family of the young girl agrees with the circumciser on when the practice will be performed. Drivers of the decision here include the belief that cutting will help control women's sexual desire and increase fidelity in marriage, and in some cases the belief that uncut female genitalia are aesthetically less `beautiful'. It should be noted that FGC was prohibited in Nigeria through the Violence Against Persons Prohibition Act in 2015, though enforcement remains weak.

In Sierra Leone, cutting is part of an initiation ritual known as ``Bondo''. It is performed at puberty, at the end of a month-long period that the girl spends in the bush together with other girls and with traditional circumcisers (``soweis''). During this period, girls are initiated into adulthood and taught how to be good wives and good mothers. Membership in a Bondo society (i.e., a secret society) is extremely valuable from a societal standpoint: the girls who are initiated at the same time remain close to each other for life, and constitute a mutual support network. Furthermore, since Bondo initiation marks the transition to adulthood, it is typically seen as a sequential pre-requisite for marriage. Unlike Nigeria, FGC is legal in Sierra Leone.

\paragraph{Child Marriage} In Nigeria, the prevalence of child marriage has a marked geographic gradient. Although most states in the South have domesticated the Childs Rights Act of 2003, which established 18 as the minimum age for marriage, many Northern states operate Sharia courts where marriage is permitted after puberty. Drivers of child marriage include poverty: in particular, the bride price (or ``mahr" among Muslims) and the incentive to reduce the number of dependents may induce parents to marry off young daughters in times of financial stress. Additionally, low returns to girls' education (actual or perceived) reduce the opportunity cost of early marriages.

Sierra Leone is also a country of Muslim-majority, but in this case there is no sharp regional divide when it comes to marriage patterns or legislation. The Child Rights Act of 2007 established age 18 as the minimum age for marriage country-wide. 
Although the role of bride price as a potential incentive for early marriage is similar to that in Nigeria, the initiation through Bondo societies plays a symbolic and almost ``institutional" role in shaping the timing of the marriage. In fact, once a girl returns from the Bondo bush, she is understood by the community to be an adult woman and ready for marriage. Furthermore, the fact that other girls initiated in the same cohort may get married can trigger peer effects (e.g., girls who remain unmarried may feel marginalized).

\paragraph{Polygyny.} Within the African continent, polygyny is most common in West and Central Africa, and, in fact, Nigeria has one of the  highest polygyny rates in the world. The drivers of this practice include signaling social status (wealthier men can afford to support more wives), a desire to have more children either in order to have more farm labor or for support in old age, and---in some circumstances---strategic alliances among clans or lineages. The relative roles of senior and junior wives are generally governed by custom in both Nigeria and Sierra Leone, with senior wives typically having greater authority over household decisions.


\section{Data and descriptive analysis}

\subsection{Data}
\label{sec:data}

We use repeated cross-sections from the Demographic and Health Surveys (DHS), which provide nationally representative microdata aimed at monitoring population and health worldwide.\footnote{The DHS employs a stratified two-stage cluster sampling design based on the most recent census frame available for each country. The microdata are publicly available from \url{https://dhsprogram.com/data/}.} 
We focus our analysis on Sierra Leone and Nigeria. For Sierra Leone, three survey rounds are available, namely 2008, 2013, and 2019. For Nigeria, the available rounds are 2003, 2008, 2013, 2018, 2024.

We use individual women's questionnaires administered to women aged 15-49, but restrict the sample for our main analysis to those aged 18-49 at the time of the interview. The reason is that we are interested in studying, among other norms, child marriage, and respondents younger than 18 would still be at risk of child marriage. 
The sample contains 34,437 women in Sierra Leone and 134,083 women in Nigeria.\footnote{The estimation samples differ between norm pairs due to
the outcome-specific missingness for FGC, PGY, and CMA.} 

Among other questions, the DHS women’s questionnaire contains harmonized modules on female genital cutting, marriage  history (which we use to construct measures of child marriage), and co-wives indicating polygyny of the husband. These are the three sets of norms on which we focus our main analysis. We next describe how the relevant variables are constructed.

\paragraph{Child Marriage (CMA).} The DHS collects retrospective information on a woman’s age at first marriage, asking respondents to report the age, month, and year at which they first got married. Our individual-level measure of child marriage is a dummy equal to one if a woman reports entering her first marriage before age 18.
We then aggregate these dummies at the birth-year $\times$ region $\times$ ethnic group level to construct relevant cohort-level proportions for our analysis.

\paragraph{Female Genital Cutting (FGC).} In the DHS, respondents are first asked whether they know what FGC means. We code women who report not knowing what FGC means as missing (that is, we restrict the sample to those individuals who likely had a choice between FGC and other norms; Supplemental Appendix \ref{sec:fgc_robustness} reports results where these observations are not excluded from the sample). Among women who respond affirmatively, the survey then asks whether they have been circumcised and, if so, the type of circumcision. We code the individual-level variable ``FGC'' as equal to one if a woman reports having undergone any form of circumcision, and as zero if a woman reports never having been circumcised. 
As for CMA, individual-level dummies are then aggregated into cohort-level proportions (at the birth-year $\times$ region $\times$ ethnic group level). 

\paragraph{Polygyny (PGY).} The DHS collects information on co-wives among currently married women, asking respondents to report the number of co-wives they have. Our individual-level measure of polygyny is a dummy equal to one if a woman reports having at least one co-wife and is not the first-ranked wife in the union, and equal to zero otherwise (i.e., zero includes women who are not currently married, who are in a monogamous marriage, and who are first-ranked wives). Restricting to non-first-ranked wives ensures that the variable captures women who entered a marriage that was already polygynous at union formation, meaning that both the woman and her family knew that she would be joining a polygynous household, unlike first-ranked wives whose families may have not anticipated future co-wives.

\paragraph{Covariates.} In our analysis, we control for a number of covariates taken from the DHS data. 
    \begin{enumerate}[label=(\roman*)]
        \item \textit{Education}. We use the highest level attained by respondents and classify it into three categories: less than primary education, primary education, and secondary education or higher. 
        \item \textit{Urban residence}. We include an indicator for urban residence, based on the DHS urban–rural classification drawn from official national records.
        \item \textit{Religion}. We construct a variable denoting whether the individual is Muslim, Christian or something else. 
        \item \textit{Household wealth}. We use the DHS wealth index, which is constructed through a principal component analysis of ownership of household assets and housing characteristics, including durable goods, household materials, and access to basic services such as water and sanitation. 
        We construct five dummies indicating whether individuals are in different quintiles of the wealth distribution (separately for each country).
        \item \textit{Ethnicity}. We use a self-reported categorical variable in which respondents identify their ethnic group from a country-specific list. For Sierra Leone, we recode this variable to capture the largest ethnic groups, resulting into three categories: Mende, Temne, or Other. In the case of Nigeria, we aggregate ethnic groups into Hausa-Fulani, Igbo, Yoruba, and Other.
        \item \textit{Region}. We use the subnational administrative divisions provided by the DHS, which typically correspond to the first administrative level below the national level (ADM1). For Sierra Leone, these regions correspond to four areas: the Northern province, the Western and North Western provinces, the Eastern province, and Southern province.\footnote{Regions correspond to Sierra Leone’s provinces, which we refer to as North, South, East, and West. Because the present-day Western Area is split into two provinces, as Northwest and West, but older observations cannot always be uniquely assigned to one of the two, we merge Northwest and West into a single “West” category, hence obtaining four instead of the current five provinces.} For Nigeria, the regions are North West, North East, North Central, South West, South East and South South.
    \end{enumerate}

We report the covariate composition in Supplemental Appendix Table \ref{tab:sl_summary1} for Sierra Leone and Supplemental Appendix Table \ref{tab:ng_summary1} for Nigeria. 


\paragraph{Child marriange bans} We exploit the Child Rights Act (2007) in Sierra Leone, which prohibits child marriage, to construct a cohort-based exposure indicator for cohorts born in and after 1990. This cutoff year corresponds to the cohort that was seventeen or younger at the time of the ban.
In Nigeria, the Child Rights Act was passed at the federal level in 2003 and sets the minimum legal age of marriage at 18. In our main specification, we use the country-level ban; in our robustness checks, we also show that results are robust as we control for state-level adoption of the ban.\footnote{Under Nigeria's federal system, the Child Rights Act becomes legally
binding upon state-level adoption and gubernatorial assent. However, variation in adoption years across states may be endogenous, and for this reason we prefer to rely on the national level ban (i.e., the first date in which the ban was enacted) in our main analysis. For the robustness checks that rely on state-level variation, we use state-specific assent years from the Partners West Africa Nigeria Child Rights Act Tracker and construct a state-by-cohort exposure indicator. Assent dates range from 2003 to 2023 across states, and at least one state has not formally assented to the Act (see Supplemental Appendix Table \ref{cra:list}).}

\paragraph*{Droughts} We construct an annual $0.5^{\circ} \times 0.5^{\circ}$ grid-cell--level drought indicator from monthly precipitation data from the Climatic Research Unit (CRU). We aggregate monthly precipitation totals to the annual level. To define droughts, we follow Corno et al. (2020) and define a drought shock as annual rainfall at a given location falling below the 15th percentile of its own historical distribution.\footnote{Under i.i.d.\ rainfall realizations, this definition implies that each grid cell has the same ex ante probability of experiencing a shock in any given year. As a result, the shock measure is orthogonal to time-invariant local characteristics, and identification relies on within-location variation driven by the timing of rainfall shortfalls over time.} We use DHS information on the geo-location of clusters to assign respondents to the corresponding weather grid. Clusters typically comprise 20--30 households, and the number of clusters varies by country and survey wave.\footnote{Cluster coordinates are displaced by the DHS for confidentiality —up to 2 km in urban areas and up to 5 km (with a small fraction up to 10 km) in rural areas—but this displacement is small relative to the $0.5^{\circ} \times 0.5^{\circ}$ grid size, limiting measurement error in the cluster-to-grid assignment.} We classify a woman as having been ``treated'' by a drought if her assigned grid cell experiences at least one drought year between ages 12 and 15. We choose this time window since this is the most common timing for child marriage.

\subsection{Descriptive analysis}
 
 Table \ref{tab:desc_stats} in the Supplemental Appendix \ref{S_Appendix_descriptives} reports aggregate adoption rates for FGC, CMA, PGY, and their joint prevalence across major socio-demographic sub-groups in Nigeria and Sierra Leone.



\begin{figure}[!ht]
 \caption{Observed and residualized shares of norm adopters in Sierra Leone}
    \centering
   \includegraphics[scale = 0.5]{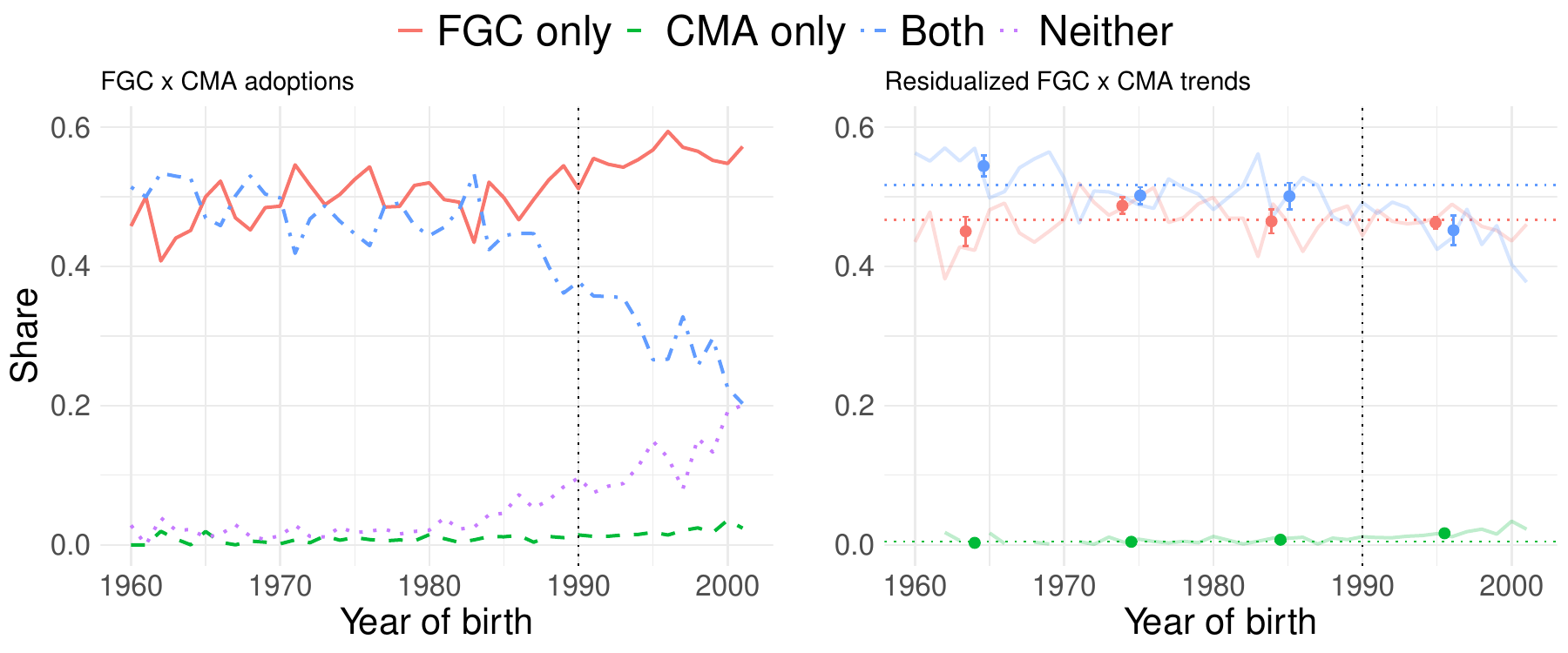}
   \includegraphics[scale = 0.5]{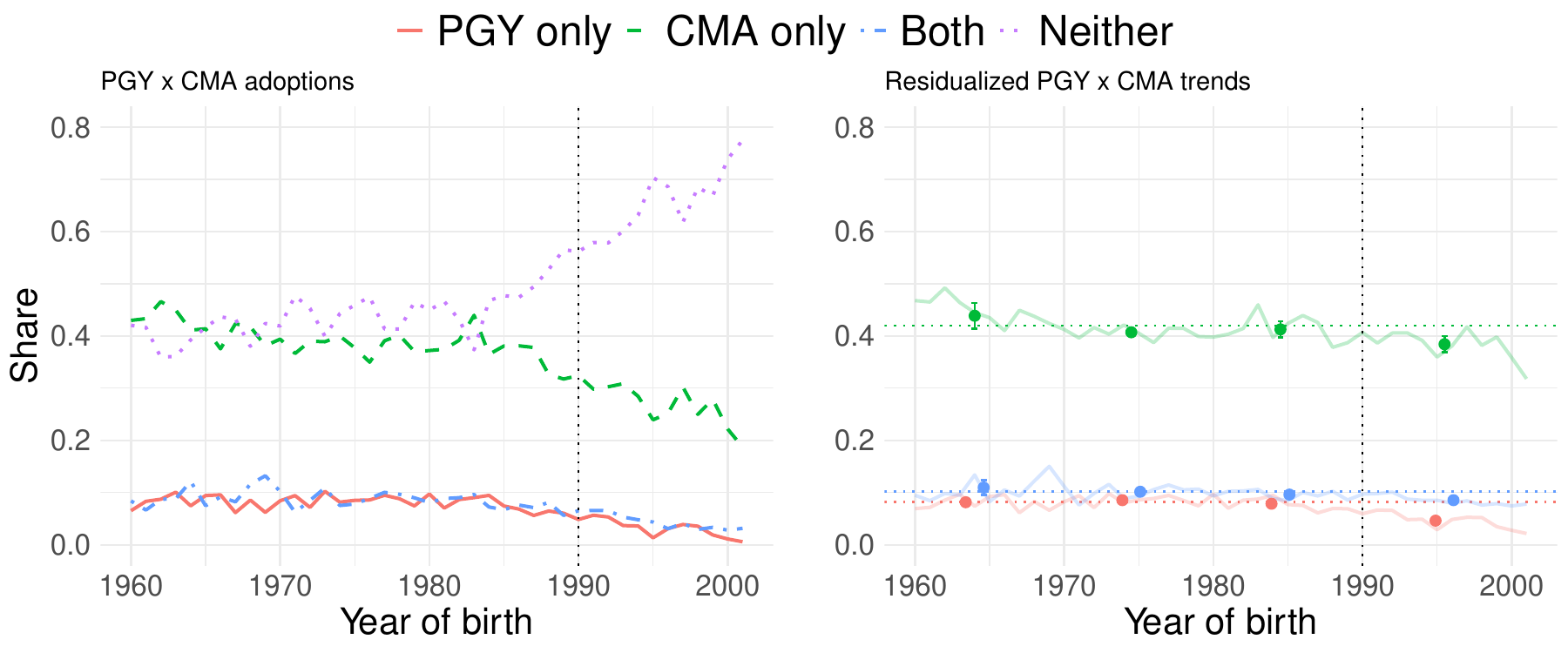} 
 
  \begin{flushleft}
  \footnotesize \emph{Notes:}  Left-top (resp. left-bottom) panel reports the share of surveyed individuals in Sierra Leone who had CMA only, FGC only or both, by birth cohorts (resp. CMA only, PGY only or both). Right-top (resp. right-bottom) panel reports the residualized share of of surveyed individuals in Sierra Leone who had CMA only, FGC only or both, by birth cohorts (resp. CMA only, PGY only or both), after residualizing by region and ethnicity fixed effects, education, religion, wealth and urban/rural indicators. The dashed lines report the residualized share and the bins report the average within four time windows (1958  - 1969, 1970 - 1979, 1980 - 1989, 1990 - 2001) with correspondig 95\% confidence intervals. 
   \end{flushleft}
       
    \label{fig:illustrative1}
\end{figure}

Figure~\ref{fig:illustrative1} (left panels) plots cohort profiles of norm adoption by respondents' year of birth: FGC only, CMA only, and joint adoption in the top row (panel a), and PGY only, CMA only, and joint adoption in the bottom row (panel c). The vertical line corresponds to the first cohort affected by the child marriage ban, i.e., the oldest girls who were not yet 18 when the Child Rights Act was passed in 2007. 
These cohort patterns may reflect multiple forces that may drive trends before the CMA ban, including (i) changes in covariate composition over time, in particular education, (ii) dynamic propagation through social spillovers, and (iii) complementarities across norms.

As a descriptive exercise, Figure~\ref{fig:illustrative1} (right panels) reports adoption shares after residualizing for region-by-ethnicity fixed effects and individual covariates. Specifically, we run a regression
\[
Y_{jit} \;=\; \alpha_{g(i)} \;+\; X_{i,t} \phi_j \;+\; \mathbbm{1}\{t \ge 1990\}\tau_j \;+\; \varepsilon_{jit}, \qquad \mathbb{E}[\varepsilon_{jit} | X_{i,t}] = 0
\]
where $Y_{jit}$ is a binary indicator for whether individual $i$ in cohort $t$ adopts bundle $j$ (e.g., FGC only, CMA only, or both; analogously for DMV--CMA). Here, $\alpha_{g(i)}$ denotes region-by-ethnicity fixed effects and $\tau_j$ captures a post-1990 shift in adoption; $X_{it}$ is a vector including current education level, religion, urban area and wealth. We do not assign a causal interpretation to this regression; rather, it serves as a descriptive device to summarize cohort trends net of observed composition. The residualized cohort shares plotted in Figure~\ref{fig:illustrative1} are
\[
\bar{Y}_{jt}^{\mathrm{res}} \;\equiv\; \frac{1}{n_t}\sum_{i:\,t(i)=t}\hat{\varepsilon}_{jit} \;+\; \mathbbm{1}\{t \ge 1990\}\hat{\tau}_j,
\]
where $\hat{\varepsilon}_{jit}$ and $\hat{\tau}_j$ are the OLS residuals and post-1990 coefficient estimate from the regression above, and $n_t$ is the number of respondents in cohort $t$ and $\{i: t(i) = t\}$ is the set of respondents of cohort $t$.

The thin lines show cohort-by-cohort residualized shares. The binned markers report averages over 1958--1969, 1970--1979, 1980--1989, and 1990--2001, with Gaussian 95\% confidence intervals. The dotted horizontal lines indicate pre-ban historical means (cohorts born before 1990). The binned residualized means are close to these historical means in the pre-1990 bins; in contrast, cohorts born after 1990 exhibit a sharp downward shift in joint adoption ($AB$), while FGC adoption remains essentially flat. 


 \begin{figure}[!ht]
    \centering
    \caption{Observed and residualized shares of norm adopters in Nigeria.}
   \includegraphics[width = 15cm]{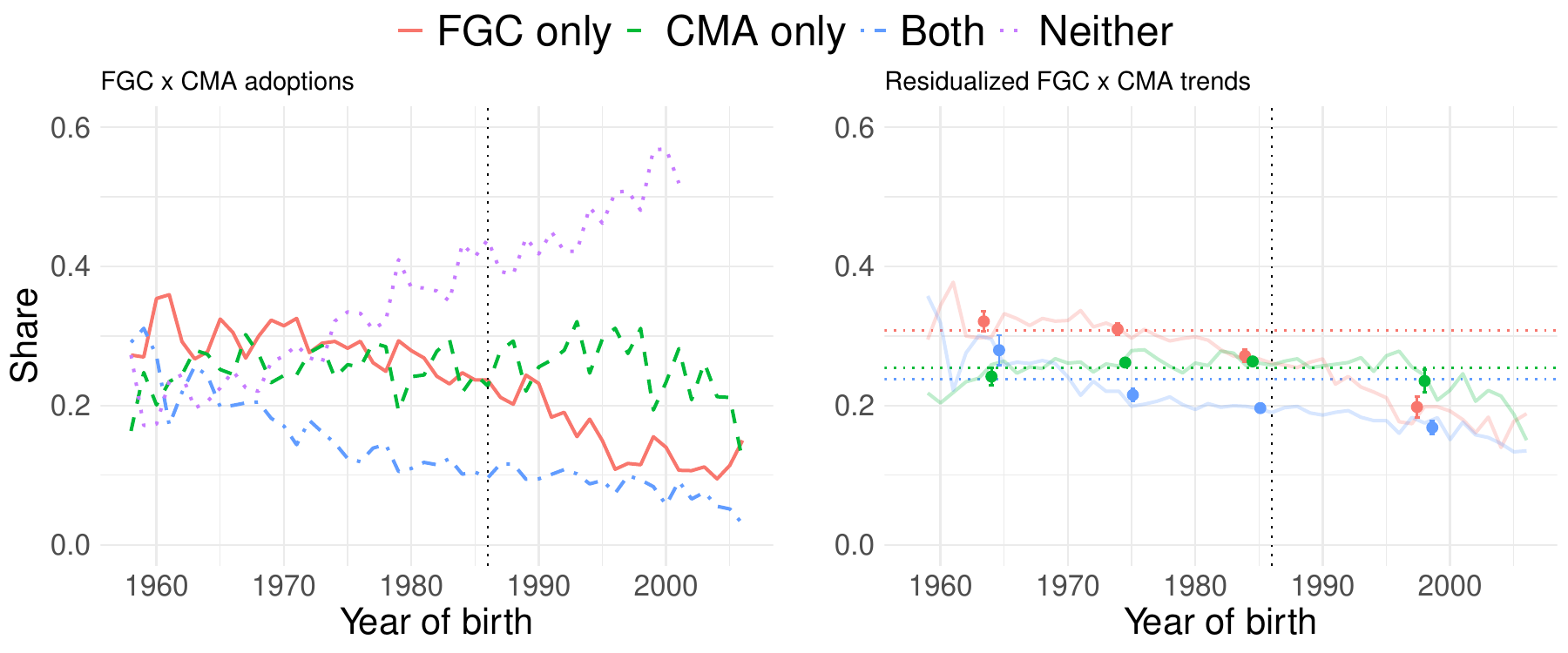}
   \includegraphics[width = 15cm]{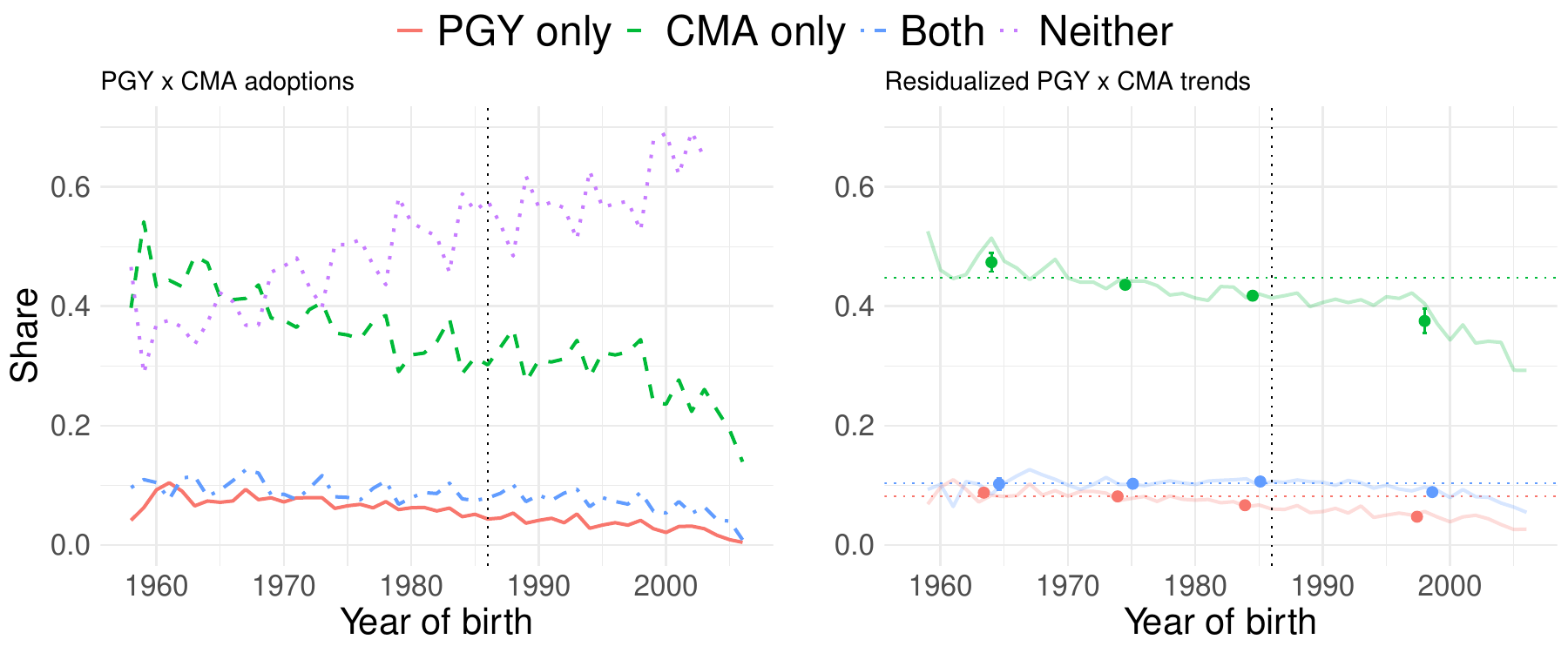}

     \begin{flushleft}
     \footnotesize \textit{Notes:} Left-top (resp. left-bottom) panel reports the share of surveyed individuals in Nigeria who had CMA only, FGC only or both, by birth cohorts (resp. CMA only, PGY only or both). Right-top (resp. right-bottom) panel reports the residualized share of of surveyed individuals in Nigeria who had CMA only, FGC only or both, by birth cohorts (resp. CMA only, PGY only or both), after residualizing by region and ethnicity fixed effects, education, religion, wealth and urban/rural indicators. The dashed lines report the residualized share and the bins report the average within four time windows (1958  - 1969, 1970 - 1979, 1980 - 1989, 1990 - 2001) with correspondig 95\% confidence intervals. 
    \end{flushleft}
    \label{fig:illustrative_Nigeria1}
\end{figure}

Figure~\ref{fig:illustrative_Nigeria1} shows the corresponding cohort trends in adoption shares for Nigeria, both in raw form and after residualizing for the observable composition and fixed effects for Nigeria. As before, the dotted vertical line corresponds to the first cohort affected by the child marriage ban. 
Unlike Sierra Leone, the residualized series exhibits pre-trends, suggesting that cohort dynamics are not well captured by a simple level shift at the national cutoff. These pre-trends may reflect persistence through lagged adoption (dependence on $p_{t-1}$) and dynamics over time. This motivates our main specification where we control for lagged adoption controls and, for Nigeria, also use drought-driven variation alone instead of the CMA ban (see Section \ref{sec:nigeria}).

 



\section{Empirical results}

In this section, we present our main results using the model in Equation \eqref{eqn:dynamic1}. Reference groups for social interactions are defined by at the region $\times$ ethnicity level, and we control for primary and secondary education, religion, urban residence, wealth, ethnicity and region fixed effects. 

\subsection{Female genital cutting and child marriage}
 
We start by studying the adoption of $(\text{FGC},\text{CMA})$ in Sierra Leone, our main specification. Our model treats adoption of each norm, as well as joint adoption, as a static choice. In the FGC--CMA setting, this choice is naturally interpreted as a parental decision: FGC typically occurs around age 12 and CMA around age 15 on average (see Supplemental Appendix Figure~\ref{SL:agedist}). 

 The estimated model parameters are shown in Table \ref{tab:main1_nodrought_fgc}. We report two variants of the model: (i) a restricted specification that estimates $(s_A,s_B)$ and imposes additivity of sanctions, $s_{AB}=s_A+s_B$ (columns 1-2); and (ii) an unrestricted specification that additionally estimates $s_{AB}$ (columns 3-4).\footnote{Note that, even though they are displayed in two separate columns, the parameters in columns 1-2 pertain to one estimation, and so do the parameters in columns 3-4.}

 \paragraph{Exclusion restriction: CMA ban} Our main exclusion restriction in Sierra Leone exploits the introduction of the child-marriage ban. In the model, this policy enters only the utility of child marriage (CMA), providing a shifter for CMA relative to the other norm.\footnote{Adding drought controls leaves the estimates unchanged (the drought coefficient is statistically insignificant and close to zero), potentially due to limited within-country variation in drought exposure in Sierra Leone (see Figure~\ref{SL:drought}).}  The child marriage ban enters negatively in the CMA equation in both specifications with a 0.07 standard deviation effect that is statistically significant. 

\paragraph{Social multipliers} In column 1, the estimated social multiplier for FGC ($s_A$) is large (2.46) and statistically significant at the 1 percent level. In contrast, the CMA social multiplier ($s_B$ in column 2) is much smaller (0.36), although also significant at the 1 percent level. The comparatively large multiplier for FGC is consistent with the view that FGC in Sierra Leone is a highly socially coordinated practice, tied to social inclusion and group membership through the Bondo secret society.  
 
 \paragraph{Test for linear social interactions} Allowing for non-additive sanctions (columns 3-4) does not change the conclusions: 
 we find $s_{AB} = 2.83$, while $s_A + s_B = 2.91$. Using a standard Wald test, we cannot reject the null hypothesis that $s_{AB} = s_A + s_B$ (p-value is $67\%$). 

\paragraph{Complementarities} The estimated parameter $\Gamma$ is positive and statistically significant in both specifications, with $\Gamma \approx 1$. 
Thus, by Theorem \ref{prop_gamma_fixed_point_extended}, the data suggests that CMA and FGC are complements. Such complementarity is easy to rationalize in the context of Sierra Leone because, as discussed in section \ref{sec:norms_background}, the fact that the initiation to Bondo occurs in adolescence creates a natural segway into (early) marriage. Another way to view this is that, if a family wants to marry off a daughter, she will first need to undergo Bondo initiation in order to be recognized as a woman.

 \begin{table}[!htbp]
\caption{Sierra Leone, main results: estimated complementarity and social multipliers (FGC$\times$CMA).}
\label{tab:main1_nodrought_fgc}
\centering
\resizebox{\ifdim\width>\linewidth\linewidth\else\width\fi}{!}{
\begin{tabular}{lcccc}
\toprule
 \multicolumn{1}{c}{\textbf{Sierra Leone}}  & \multicolumn{4}{c}{FGC $\times$ CMA} \\
\cmidrule(l{3pt}r{3pt}){2-5}
 & \multicolumn{2}{c}{ $s_A, s_B$ } & \multicolumn{2}{c}{$s_A, s_B, s_{AB}$} \\
\cmidrule(l{3pt}r{3pt}){2-3} \cmidrule(l{3pt}r{3pt}){4-5}
 & FGC  & CMA  &  FGC &  CMA \\
 & (1) & (2) & (3) & (4) \\
\midrule
CMA Ban & \multicolumn{2}{c}{$-0.07^{***}$} & \multicolumn{2}{c}{$-0.07^{***}$} \\
 & \multicolumn{2}{c}{$(0.02)$} & \multicolumn{2}{c}{$(0.02)$} \\
$\Gamma$ & \multicolumn{2}{c}{$1.05^{***}$} & \multicolumn{2}{c}{$1.08^{***}$} \\
 & \multicolumn{2}{c}{$(0.07)$} & \multicolumn{2}{c}{$(0.08)$} \\
$s$ & $2.46^{***}$ & $0.36^{***}$ & $2.47^{***}$ & $0.44^{*}$ \\
 & $(0.24)$ & $(0.13)$ & $(0.25)$ & $(0.24)$ \\
 $s_{AB}$ & \multicolumn{2}{c}{} & \multicolumn{2}{c}{$2.83^{***}$} \\
& \multicolumn{2}{c}{} & \multicolumn{2}{c}{$(0.28)$} \\
$\rho$ & \multicolumn{2}{c}{$-0.70^{***}$} & \multicolumn{2}{c}{$-0.70^{***}$} \\
 & \multicolumn{2}{c}{$(0.04)$} & \multicolumn{2}{c}{$(0.04)$} \\
 \hline
Primary Education & $-0.16^{***}$ & $-0.05^{**}$ & $-0.16^{***}$ & $-0.05^{**}$ \\
 & $(0.04)$ & $(0.02)$ & $(0.04)$ & $(0.02)$ \\
Secondary Education & $-0.31^{***}$ & $-0.54^{***}$ & $-0.31^{***}$ & $-0.54^{***}$ \\
 & $(0.04)$ & $(0.02)$ & $(0.04)$ & $(0.02)$ \\
Muslim & $0.23$ & $-0.03$ & $0.23$ & $-0.03$ \\
 & $(0.17)$ & $(0.10)$ & $(0.17)$ & $(0.10)$ \\
Christian & $-0.19$ & $-0.04$ & $-0.19$ & $-0.04$ \\
 & $(0.17)$ & $(0.11)$ & $(0.17)$ & $(0.11)$ \\
Urban & $-0.04$ & $-0.06^{***}$ & $-0.04$ & $-0.06^{***}$ \\
 & $(0.03)$ & $(0.02)$ & $(0.03)$ & $(0.02)$ \\
\midrule
Ethnicity FE & Yes & Yes & Yes & Yes \\
Region FE & Yes & Yes & Yes & Yes \\
Wealth FE & Yes & Yes & Yes & Yes \\
\midrule  p-value for $H_0: s_{AB} = s_A + s_B$ & \multicolumn{2}{c}{ } &  \multicolumn{2}{c}{ 0.67  } \\
\midrule
Sample size & \multicolumn{2}{c}{33161 } & \multicolumn{2}{c}{ 33161 } \\
\bottomrule
\end{tabular}
}
 
\vspace{0.4cm}

\begin{tablenotes}
    \footnotesize \item \textit{Note:} Maximum-likelihood estimates from the two-norm adoption model for FGC$\times$CMA. Columns labeled $s_A,s_B$ impose additive sanctions ($s_{AB}=s_A+s_B$); columns labeled $s_A,s_B,s_{AB}$ additionally estimate $s_{AB}$. The exclusion restriction is the Child Marriage (CMA) ban for cohorts born after 1990, which enters only the CMA utility but not the FGC utility. Reference groups for social interactions are defined by region$\times$ethnicity. All specifications control for education, religion, urban residence, and ethnicity, region, and wealth fixed effects. Standard errors in parentheses.
\end{tablenotes}
\end{table}

\paragraph{Idiosyncratic taste shocks.} Finally, we find that   
 the estimated $\rho$ is negative and significant. A negative $\rho$ indicates negative correlations in idiosyncratic taste shocks, which  may occur in the presence of excluded variables that affect the utility of each norm with different signs \citep{Gentzkow_AER_2007}.

\paragraph{Policy counterfactuals}  We consider two main policy counterfactuals. The first examines the effect of reducing the baseline utility of one or both norms, which can be interpreted, for example, as the impact of a ban or other measures that directly lower the private returns to adoption. The second considers an intervention that reduces the social spillover component $s_j$, which may capture policies such as information campaigns that weaken perceived sanctions, or weakening of social institutions traditionally responsible for norm enforcement. (see Section~\ref{subsec:policy_counterfactuals} for a formal definition).


Table \ref{tab:main1_counterfactual} collects policy counterfactual for interventions to baseline utilities. 
Several patterns stand out. First, policies that lower the baseline utility of FGC (i.e., reduce $\delta_A$) generate large and increasing declines in the overall prevalence of either norm, as measured by $\Delta Q_{\mathrm{tot}}\equiv \Delta(p_A+p_B+p_{AB})$. For example, a $-0.1$ shock to FGC-utility ($\delta_A$) implies $\Delta Q_{\mathrm{tot}}=-0.019$ at horizon $h=1$ and $-0.083$ at $h=10$, while a $-0.2$ shock yields $-0.041$ at $h=1$ and $-0.202$ at $h=10$, where each period corresponds to one year (so $h = 10$ is a ten year horizon). In contrast, policies that lower the baseline utility of CMA (i.e., reduce $\delta_B$) produce smaller reductions in overall adoption even at longer horizons (e.g., $\Delta Q_{\mathrm{tot}}=-0.029$ at $h=10$ for $\Delta\delta_B=-0.1$, and $-0.058$ for $\Delta\delta_B=-0.2$). This suggests that FGC-directed interventions are more effective than CMA-directed interventions in reducing the prevalence of either norm.

Second, the composition effects differ across interventions. The reduction in utility of CMA ($\delta_B$) mainly reduces joint adoption ($\Delta p_{AB}<0$) while partially reallocating mass to FGC-only adoption ($\Delta p_A>0$), while the reduction of utility of FGC reduces both FGC-only adoption and joint adoption, with only a small offset increase in CMA-only adoption.

Finally, joint reductions in the baseline utilities of both FGC and CMA generate the largest overall declines, though the incremental gains relative to an FGC-only intervention vary significantly. Depending on the horizon and shock magnitude, joint interventions reduce total prevalence between 10 and 50 percent more than FGC-only interventions.

\begin{table}[!ht]
\caption{Sierra Leone (FGC$\times$CMA): baseline-utility counterfactuals.}
\label{tab:main1_counterfactual}
\centering

\resizebox{\textwidth}{!}{%
\begin{tabular}{lccccccccc}
\toprule
\multicolumn{1}{c}{A: FGC, B: CMA } & \multicolumn{3}{c}{$h=1$} & \multicolumn{3}{c}{$h=5$} & \multicolumn{3}{c}{$h=10$} \\
\cmidrule(l{3pt}r{3pt}){2-4} 
\cmidrule(l{3pt}r{3pt}){5-7} 
\cmidrule(l{3pt}r{3pt}){8-10}
 & $\Delta p_A$ & $\Delta p_B$ & $\Delta p_{AB}$ 
 & $\Delta p_A$ & $\Delta p_B$ & $\Delta p_{AB}$ 
 & $\Delta p_A$ & $\Delta p_B$ & $\Delta p_{AB}$\\
\midrule
\multicolumn{10}{l}{\textit{Changes in $\delta_A$ (constantA)}}\\
\addlinespace
$\Delta\delta_A=-0.100$ 
& -0.012 & 0.002 & -0.009 
& -0.036 & 0.005 & -0.027 
& -0.055 & 0.006 & -0.034\\
$\Delta\delta_A=-0.200$ 
& -0.026 & 0.005 & -0.020 
& -0.089 & 0.012 & -0.063 
& -0.133 & 0.016 & -0.085\\
$\Delta\delta_A=-0.500$ 
& -0.077 & 0.013 & -0.057 
& -0.360 & 0.039 & -0.210 
& -0.498 & 0.045 & -0.256\\

\midrule
\multicolumn{10}{l}{\textit{Changes in $\delta_B$ (constantB)}}\\
\addlinespace
$\Delta\delta_B=-0.100$ 
& 0.026 & -0.001 & -0.033 
& 0.022 & -0.001 & -0.041 
& 0.015 & -0.001 & -0.043\\
$\Delta\delta_B=-0.200$ 
& 0.050 & -0.003 & -0.063 
& 0.040 & -0.002 & -0.079 
& 0.026 & -0.002 & -0.082\\
$\Delta\delta_B=-0.500$ 
& 0.109 & -0.007 & -0.141 
& 0.071 & -0.006 & -0.168 
& 0.039 & -0.006 & -0.171\\

\bottomrule
\end{tabular}
}

\vspace{0.4cm}

\begin{tablenotes}
    \footnotesize \item \textit{Notes:} Model-implied changes in adoption shares at horizons $h=1,5,10$ following permanent decreases in baseline utilities $\delta_A$ (FGC), $\delta_B$ (CMA). Entries report changes in the share adopting only $A$ ($\Delta p_A$), only $B$ ($\Delta p_B$), and both norms ($\Delta p_{AB}$), relative to the baseline predicted path under the estimated model.
\end{tablenotes}

\end{table}

Figure~\ref{fig:projection1} plots the predicted dynamics under a mild and a stringent intervention ($\Delta\delta=-0.1$ and $\Delta\delta=-0.5$), targeting either FGC (left panels) or CMA (right panels). The figure shows that an FGC-directed intervention generates a large decline in both FGC-only adoption and joint adoption, accompanied by only a modest increase in CMA-only adoption. By contrast, a CMA-only intervention operates primarily through joint adopters: it reduces joint adoption, induces an initial reallocation toward FGC-only adoption, which then gradually returns toward the original steady state. Overall, the figure highlights why FGC interventions can be more effective in equilibrium.

\begin{figure}[!h]
    \centering
    \caption{Ten-year projections under baseline-utility interventions (FGC$\times$CMA) in Sierra Leone}
    \includegraphics[width = 15cm]{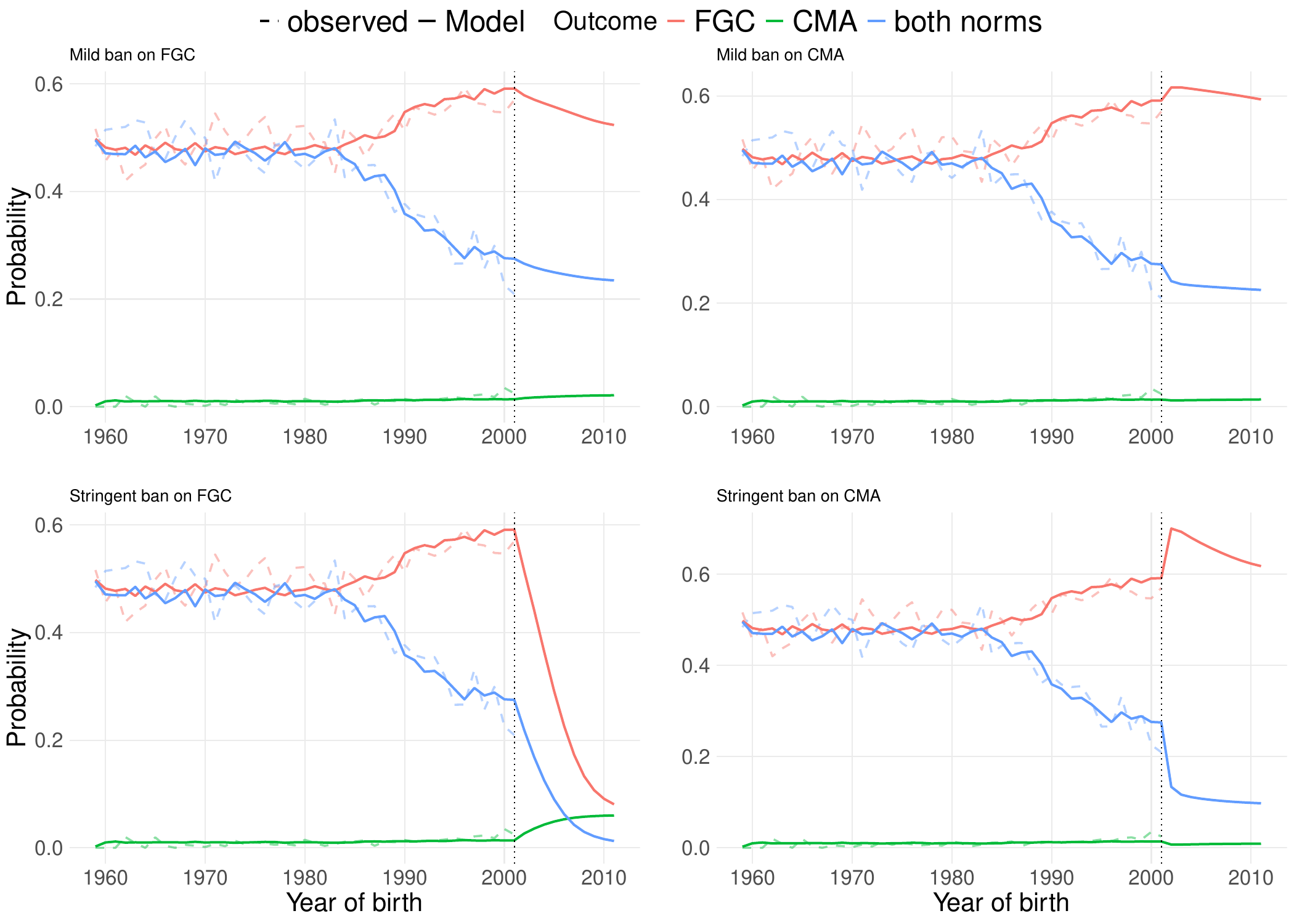}

\begin{flushleft}
    \footnotesize \textit{Notes:} Observed cohort adoption shares (dashed) and model-implied shares (solid) for adopting FGC only (red), CMA only (green), or both norms (blue). The vertical dotted line marks the start of the projection window (last observed cohort). Panels report counterfactual paths obtained by iterating the estimated model forward after applying one-time utility shocks: a mild FGC-only intervention (top left, $\Delta\delta_{\mathrm{FGC}}=-0.1$), a mild intervention on CMA (top right, $\Delta\delta_{\mathrm{CMA}}=-0.1$), a stringent FGC-only intervention (bottom left, $\Delta\delta_{\mathrm{FGC}}=-0.5$), and a stringent intervention on CMA (bottom right, $\Delta\delta_{\mathrm{CMA}}=-0.5$). Shocks are in standard-deviation units of the idiosyncratic utility shocks; covariates are held fixed at their last-observed distribution.
\end{flushleft}

    \label{fig:projection1}
\end{figure}

Table~\ref{tab:main2_counterfactual} considers counterfactuals for policies that aim to weaken social interactions. Holding all other parameters and initial conditions fixed, we scale down the FGC social multiplier to $s_A' \in \{0.9,0.7,0.5\}s_A$, and equivalently reducing $s_{AB}$ by $\{10,30,50\}\% s_A$. The effects are substantial even for modest reductions in $s_A$: a 10\% decrease lowers overall adoption $Q_{\mathrm{tot}}$ by $-0.04$ at $h=1$ and $-0.16$ by $h=10$, while halving $s_A$ reduces $Q_{\mathrm{tot}}$ by $-0.33$ at $h=1$ and $-0.62$ by $h=10$. These changes are driven primarily by sharp declines in FGC-only and joint adoption, with only a small offsetting increase in CMA-only adoption. Consistent with complementarities, weakening FGC spillovers also reduces CMA adoption through the joint-adoption margin.

\begin{table}[!htbp]
\caption{Sierra Leone (FGC$\times$CMA): social-interaction counterfactuals.}
\label{tab:main2_counterfactual}
\centering
\resizebox{\textwidth}{!}{%
\begin{tabular}{lccccccccc}\toprule\multicolumn{1}{c}{ } & \multicolumn{3}{c}{$h=1$} & \multicolumn{3}{c}{$h=5$} & \multicolumn{3}{c}{$h=10$} \\\cmidrule(l{3pt}r{3pt}){2-4} \cmidrule(l{3pt}r{3pt}){5-7} \cmidrule(l{3pt}r{3pt}){8-10} & $\Delta p_A$ & $\Delta p_B$ & $\Delta p_{AB}$ & $\Delta p_A$ & $\Delta p_B$ & $\Delta p_{AB}$ & $\Delta p_A$ & $\Delta p_B$ & $\Delta p_{AB}$\\
$s_A\leftarrow 0.90\,s_A$ & -0.027 & 0.005 & -0.022 & -0.078 & 0.012 & -0.060 & -0.103 & 0.014 & -0.074\\$s_A\leftarrow 0.70\,s_A$ & -0.106 & 0.018 & -0.079 & -0.316 & 0.037 & -0.200 & -0.362 & 0.039 & -0.218\\$s_A\leftarrow 0.50\,s_A$ & -0.215 & 0.033 & -0.145 & -0.424 & 0.043 & -0.242 & -0.422 & 0.043 & -0.240\\\bottomrule\end{tabular} }
\vspace{0.4cm}

\begin{tablenotes}
    \footnotesize \item \textit{Note:} Model-implied changes in adoption shares at horizons $h=1,5,10$ when weakening social interactions for norm $A$ (FGC) by setting $s_A' \in \{0.9,0.7,0.5\}s_A$, holding all other parameters and initial conditions fixed. Entries report changes in $\Delta p_A$, $\Delta p_B$, and $\Delta p_{AB}$ relative to the baseline predicted path.
\end{tablenotes}
\end{table}

Relative to the baseline-utility counterfactuals in Table~\ref{tab:main1_counterfactual}, spillover interventions can be at least as powerful in equilibrium. For example, a mild reduction in baseline FGC utility lowers $Q_{\mathrm{tot}}$ by $-0.019$ at $h=1$ and $-0.083$ at $h=10$, whereas a 10\% reduction in $s_A$ lowers $Q_{\mathrm{tot}}$ by $-0.04$ at $h=1$ and $-0.16$ at $h=10$. More generally, these counterfactuals highlight two implications of the estimated model: (i) complementarities imply that targeting one norm can shift adoption of the other through joint adoption, and (ii) policies that reduce social spillovers can generate large changes in equilibrium adoption shares. These dynamics are also illustrated in Figure \ref{fig:projection2}, which reports the predicted paths over a ten-period horizon.

\begin{figure}[!ht]
    \centering
    \caption{Ten-year projections under reduced social spillovers for FGC (FGC$\times$CMA) in Sierra Leone}
    \includegraphics[width=15cm]{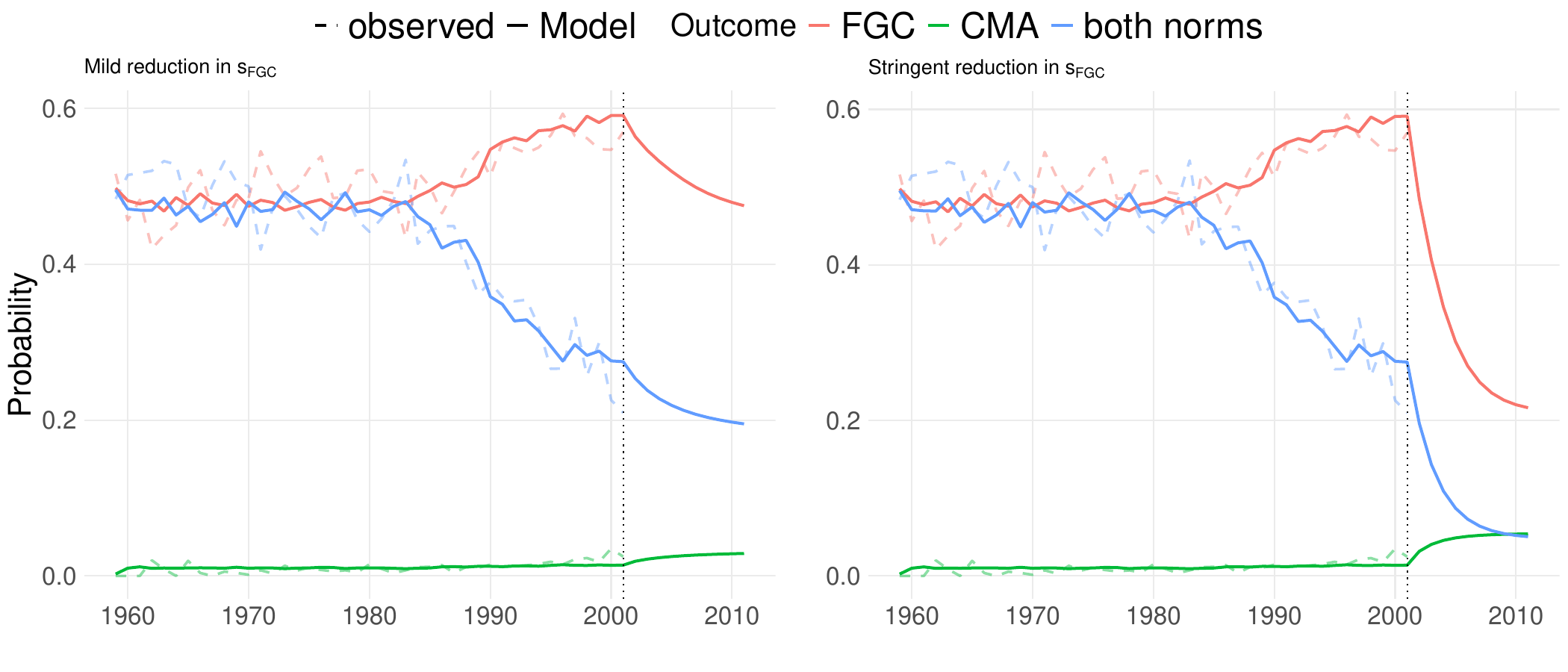}
\begin{flushleft}
    \footnotesize \textit{Notes:} Observed cohort adoption shares (dashed) and model-implied shares (solid) for adopting FGC only (red), CMA only (green), or both norms (blue). The vertical dotted line marks the start of the projection window (last observed cohort). Counterfactual paths are generated by iterating the estimated model forward after weakening the FGC social multiplier: mild reduction (left, $s_{\mathrm{FGC}}'\!=0.9\,s_{\mathrm{FGC}}$) and stringent reduction (right, $s_{\mathrm{FGC}}'\!=0.5\,s_{\mathrm{FGC}}$), holding all other parameters and initial conditions fixed and keeping covariates at their last-observed distribution.
\end{flushleft}
    \label{fig:projection2}
\end{figure}

\paragraph{FGC and CMA in Nigeria}  \label{sec:nigeria}

In Appendix \ref{appendix_section_NG_FGC_CMA}, we complement the  Sierra Leone analysis by performing the analogous exercise in Nigeria. A distinctive feature of Nigeria  is the pronounced asymmetry in timing: FGC is typically realized at younger ages than CMA. Specifically, average age of FGC is three years and average CMA age is fifteen, see Figure \ref{NG:age}, while the support of FGC timing decision overlaps with the support of CMA allowing us to identify complementarities in those settings where decisions are joint. This timing difference suggests that complementarities and idiosyncratic taste correlations between FGC and CMA should be smaller than in Sierra Leone, where the two decisions are closer in age and more likely to be jointly determined, providing a natural ``test'' for our model.\footnote{In Supplemental Appendix \ref{sec:fgc_robustness} we also report robustness checks where we drop individuals for which FGC was performed early in their life (before five years old).}  

In addition, in Nigeria we also use drought for the exclusion restriction with or without the CMA ban. We find significantly smaller complementary effects and weaker shock correlations compared to Sierra Leone ($\Gamma = -0.09, \rho = -0.1$).


\paragraph{Robustness  checks}

In Section \ref{sec:fgc_robustness}, we report several robustness checks for both Sierra Leone and Nigeria. For both countries, we re-estimate the model after excluding individuals exposed to FGC before age five and after treating respondents who report not knowing what FGC is as non-adopters. We also report specifications that use drought as an additional exclusion restriction, including in Sierra Leone. For Nigeria, we further consider specifications that control for the Sharia law passed in 1999 as an additional covariate. The results are qualitatively unchanged for all these exercises, with similar point estimates and statistical significance for all specifications.

\subsection{Child marriage and polygyny}

We next turn to the joint adoption of $(\text{PGY},\text{CMA})$ in Nigeria. 
 The estimated model parameters are reported in Table~\ref{tab:ng_main1_nodrought_pgy}. We again report two variants: (i) a restricted specification that estimates $(s_A,s_B)$ and imposes additivity of sanctions, $s_{AB}=s_A+s_B$ (columns 1-2); and (ii) an unrestricted specification that additionally estimates $s_{AB}$ (columns 3-4). 
As before, identification exploits the child-marriage ban, which in the model enters only the CMA utility and provides an exogenous shifter for CMA relative to PGY. The ban leads to significant negative effects on both specifications, with estimates around $-0.04$ (s.e.\ 0.01).\footnote{In Supplemental Appendix \ref{sec:pgy_robustness}, we show that the results remain robust when drought is used as an additional or only exclusion restriction.} 

\paragraph{Lack of technological complementarities.} Turning to technological complementarities, the estimate of $\Gamma$ is close to zero. Under additive sanctions, $\Gamma=-0.06$ (s.e.\ 0.03) and marginally significant, while under the unrestricted specification $\Gamma=0.05$ (s.e.\ 0.05) and not statistically different from zero. Taken together, these estimates provide evidence of lack of non-social (technological) complementarities between PGY and CMA in Nigeria. 
This is not an obvious result because, if one assumed that CMA and PGY were mutually reinforcing practices used to perpetuate a patriarchal culture, one would expect a positive $\Gamma$. More insights on the relationship between the two norms will emerge from considering the social return parameters, to which we turn next.

\paragraph{Social returns.} First, we note that the estimated social multipliers are large for both norms. Under the unrestricted specification, $s_A=2.45$ (s.e.\ 0.09) for PGY and $s_B=1.26$ (s.e.\ 0.07) for CMA. In contrast to the (FGC, CMA) bundle, the data provide strong evidence against additive social interactions for the PGY--CMA pair. Under the additive specification, $s_A+s_B=3.51$, while the unrestricted estimate is $s_{AB}=3.01$ (s.e.\ 0.11). The implied gap is statistically significant, and we reject $H_0:s_{AB}=s_A+s_B$ (p-value $=0$ from a Wald test). Thus, while joint adoption is socially reinforced (since $s_{AB}$ is large and positive), the joint social return is \emph{sub-additive} relative to the sum of the marginal conformity incentives. 

Recall from Theorem 1 in section \ref{sec:equil} that, when $\Gamma \approx 0$, $s_{AB}<s_A+s_B$ implies that the two norms are social (and overall) \textit{substitutes}. 
To understand this result, it may be useful to highlight the asymmetry in the ``status returns" accruing to the husband and to the wife in a polygynous union. For the husband, having an extra wife unambiguously signals wealth and -potentially- virility (especially if the new wife is young). For the bride, the signal is more mixed. Although entering a wealthy family can be status-enhancing for the bride and for her family, this effect depends on the rank the wife will have within the new household. The role of a junior wife is subordinate to that of the first wife. Therefore, when parents consider the options for marrying a young daughter (who is highly valued in the marriage market for her young age), they will be less likely to accept a position of second or third wife for her. As a result, substitution effect may arise in our context, as we construct polygynous unions based on the wife not being the first wife (see the variable description in section \ref{sec:data}).\footnote{The reason for restricting to non-first ranked wives in our definition is that we aim to capture decisions where parents know that the girl will enter a polygynous household at the moment of deciding, as opposed to capturing cases where additional wives are added later on and could not have been anticipated at the time of marriage.}

In Appendix \ref{appendix_section_SL_PGY_CMA}, we report the analysis for Sierra Leone, showing similar results (i.e., the two norms are substitutes) except that the magnitude of social returns is smaller than in Nigeria. This may reflect the fact that in  Sierra Leone,  most social pressure is centered around the ritual of FGC (``Bondo''), relative to other gender norms.

\begin{table}[!h]
\caption{Nigeria (PGY$\times$CMA): Estimated complementarity and social multipliers.}
\label{tab:ng_main1_nodrought_pgy}
\centering
\resizebox{\ifdim\width>\linewidth\linewidth\else\width\fi}{!}{
\begin{tabular}{lcccc}
\toprule
\multicolumn{1}{c}{\textbf{Nigeria}} & \multicolumn{4}{c}{ PGY $\times$ CMA } \\
\cmidrule(l{3pt}r{3pt}){2-5}
 & \multicolumn{2}{c}{$s_A,s_B$} & \multicolumn{2}{c}{$s_A,s_B,s_{AB}$} \\
\cmidrule(l{3pt}r{3pt}){2-3} \cmidrule(l{3pt}r{3pt}){4-5}
 & PGY  & CMA  & PGY & CMA \\
 & (1) & (2) & (3) & (4) \\
\midrule
CMA Ban &  & $-0.04^{***}$ &  & $-0.04^{***}$ \\
 &  & $(0.01)$ &  & $(0.01)$ \\
$\Gamma$ & \multicolumn{2}{c}{$-0.06^{*}$} & \multicolumn{2}{c}{$0.05$} \\
 & \multicolumn{2}{c}{$(0.03)$} & \multicolumn{2}{c}{$(0.05)$} \\
$s$ & $2.27^{***}$ & $1.24^{***}$ & $2.45^{***}$ & $1.26^{***}$ \\
 & $(0.11)$ & $(0.07)$ & $(0.09)$ & $(0.07)$ \\
$s_{AB}$ &  &  &  \multicolumn{2}{c}{$3.01^{***}$} \\
 &  &  &  \multicolumn{2}{c}{$(0.11)$} \\
 \hline
$\rho$ & \multicolumn{2}{c}{$0.11^{***}$} & \multicolumn{2}{c}{$0.07$} \\
 & \multicolumn{2}{c}{$(0.03)$} & \multicolumn{2}{c}{$(0.04)$} \\
Primary Education & $-0.13^{***}$ & $-0.12^{***}$ & $-0.13^{***}$ & $-0.12^{***}$ \\
 & $(0.02)$ & $(0.01)$ & $(0.01)$ & $(0.01)$ \\
Secondary Education & $-0.58^{***}$ & $-0.84^{***}$ & $-0.57^{***}$ & $-0.83^{***}$ \\
 & $(0.02)$ & $(0.01)$ & $(0.02)$ & $(0.01)$ \\
Muslim & $0.38^{***}$ & $0.26^{***}$ & $0.38^{***}$ & $0.26^{***}$ \\
 & $(0.02)$ & $(0.01)$ & $(0.02)$ & $(0.01)$ \\
Christian & $0.00$ & $0.00$ & $0.00$ & $0.00$ \\
 & $(0.02)$ & $(0.02)$ & $(0.02)$ & $(0.02)$ \\
Urban & $-0.16^{***}$ & $-0.14^{***}$ & $-0.16^{***}$ & $-0.14^{***}$ \\
 & $(0.01)$ & $(0.01)$ & $(0.01)$ & $(0.01)$ \\
\midrule
Ethnicity FE & Yes & Yes & Yes & Yes \\
Region FE & Yes & Yes & Yes & Yes \\
Wealth FE & Yes & Yes & Yes & Yes \\
\midrule  p-value for $H_0: s_{AB} = s_A + s_B$ & \multicolumn{2}{c}{} & \multicolumn{2}{c}{ $0.00^{***}$ } \\
\midrule
Sample size & \multicolumn{2}{c}{ 128785  } & \multicolumn{2}{c}{128785 } \\
\bottomrule
\end{tabular}
}
\vspace{0.4cm}

\begin{tablenotes}
    \footnotesize \item  \textit{Notes:} Maximum-likelihood estimates from the two-norm adoption model for PGY$\times$CMA. Columns labeled $s_A,s_B$ impose additive sanctions ($s_{AB}=s_A+s_B$); columns labeled $s_A,s_B,s_{AB}$ additionally estimate $s_{AB}$. The exclusion restriction is the Child Marriage (CMA) ban for cohorts born after 1990, which enters only the CMA utility but not the PGY utility. Reference groups for social interactions are defined by region$\times$ethnicity. All specifications control for education, religion, urban residence, and ethnicity, region, and wealth fixed effects. Standard errors in parentheses.
\end{tablenotes}
\end{table}

\paragraph{Policy counterfactuals} 
We replicate the two counterfactual exercises for Nigeria using the PGY$\times$CMA specification. Relative to the Sierra Leone (FGC$\times$CMA) bundle, the model-implied policy effects are smaller in Nigeria at comparable horizons. This difference is consistent with two features of the Nigerian estimates: (i) limited evidence of technological complementarities ($\Gamma$ is close to zero), and (ii) decreasing social returns to joint adoption ($s_{AB}<s_A+s_B$), which dampens the equilibrium spillovers.

Table~\ref{tab:ng_pgy_counterfactual_delta_h1h5h10} reports counterfactuals for permanent shocks to the baseline utilities. Baseline-utility interventions remain the primary driver of changes in overall adoption, but the implied equilibrium responses are modest in absolute magnitude. For instance, a $-0.1$ shock to the PGY baseline utility ($\Delta\delta_A=-0.1$) yields a small decline in total adoption of either norm, with $\Delta Q_{\mathrm{tot}}\equiv \Delta(p_A+p_B+p_{AB})\approx -0.007$ at $h=1$ and $-0.010$ at $h=10$. In contrast, a comparable shock to the CMA baseline utility ($\Delta\delta_B=-0.1$) produces larger declines in total adoption, with $\Delta Q_{\mathrm{tot}}\approx -0.022$ at $h=1$ and $-0.035$ at $h=10$ given the weaker social spillovers (so that more individuals react to the utility shock). Compositionally, the CMA-directed intervention operates primarily through reductions in CMA-only adoption ($\Delta p_B<0$) together with a mild decline in joint adoption ($\Delta p_{AB}<0$).

\begin{table}[!ht]
\caption{Nigeria (PGY$\times$CMA): baseline-utility counterfactuals.}
\label{tab:ng_pgy_counterfactual_delta_h1h5h10}
\centering

\resizebox{\textwidth}{!}{%
\begin{tabular}{lccccccccc}
\toprule
\multicolumn{1}{c}{ } & \multicolumn{3}{c}{$h=1$} & \multicolumn{3}{c}{$h=5$} & \multicolumn{3}{c}{$h=10$} \\
\cmidrule(l{3pt}r{3pt}){2-4} 
\cmidrule(l{3pt}r{3pt}){5-7} 
\cmidrule(l{3pt}r{3pt}){8-10}
 & $\Delta p_A$ & $\Delta p_B$ & $\Delta p_{AB}$ 
 & $\Delta p_A$ & $\Delta p_B$ & $\Delta p_{AB}$ 
 & $\Delta p_A$ & $\Delta p_B$ & $\Delta p_{AB}$\\
\midrule
\multicolumn{10}{l}{$\delta_A$ \textit{counterfactuals}}\\
\addlinespace
$\Delta\delta_A=-0.100$ 
& -0.007 & 0.005 & -0.005 
& -0.010 & 0.008 & -0.008 
& -0.010 & 0.008 & -0.008\\
$\Delta\delta_A=-0.200$ 
& -0.013 & 0.010 & -0.010 
& -0.018 & 0.014 & -0.014 
& -0.018 & 0.014 & -0.014\\
$\Delta\delta_A=-0.500$ 
& -0.026 & 0.022 & -0.022 
& -0.030 & 0.025 & -0.025 
& -0.030 & 0.025 & -0.025\\

\midrule
\multicolumn{10}{l}{$\delta_B$ \textit{counterfactuals}}\\
\addlinespace
$\Delta\delta_B=-0.100$ 
& 0.003 & -0.022 & -0.003 
& 0.004 & -0.034 & -0.004 
& 0.004 & -0.035 & -0.004\\
$\Delta\delta_B=-0.200$ 
& 0.005 & -0.043 & -0.005 
& 0.008 & -0.065 & -0.007 
& 0.008 & -0.066 & -0.007\\
$\Delta\delta_B=-0.500$ 
& 0.012 & -0.098 & -0.013 
& 0.018 & -0.136 & -0.017 
& 0.018 & -0.136 & -0.017\\

\bottomrule
\end{tabular}
}

\vspace{0.4cm}

\begin{tablenotes}
    \footnotesize 
    \item \textit{Notes:} Model-implied changes in adoption shares at horizons $h=1,5,10$ following permanent decreases in baseline utilities $\delta_A$ (PGY), $\delta_B$ (CMA). Entries report changes in the share adopting only $A$ ($\Delta p_A$), only $B$ ($\Delta p_B$), and both norms ($\Delta p_{AB}$), relative to the baseline predicted path under the estimated model. Total adoption is $Q_{\mathrm{tot}}=p_A+p_B+p_{AB}$, so $\Delta Q_{\mathrm{tot}}=\Delta p_A+\Delta p_B+\Delta p_{AB}$.
\end{tablenotes}
\end{table}

Table~\ref{tab:ng_pgy_counterfactual_social} reports counterfactuals that weaken social interactions by reducing $s_A$ and adjusting $s_{AB}$ accordingly. In Nigeria, the effects of these spillover interventions are economically small. For example, a 10\% reduction in $s_A$ implies $\Delta Q_{\mathrm{tot}}\approx -0.002$ at $h=1$ and $-0.002$ at $h=10$, while even halving $s_A$ yields only $\Delta Q_{\mathrm{tot}}\approx -0.008$ at $h=1$ and $-0.009$ at $h=10$. This contrasts with the earlier analysis for Sierra Leone, where comparable reductions in the FGC social multiplier generated large and growing declines in overall adoption.

\begin{table}[!htbp]
\caption{Nigeria (PGY$\times$CMA): social-interaction counterfactuals.}
\label{tab:ng_pgy_counterfactual_social}
\centering
\resizebox{\textwidth}{!}{%

\begin{tabular}{lccccccccc}
\toprule
\multicolumn{1}{c}{ } & \multicolumn{3}{c}{$h=1$} & \multicolumn{3}{c}{$h=5$} & \multicolumn{3}{c}{$h=10$} \\
\cmidrule(l{3pt}r{3pt}){2-4} \cmidrule(l{3pt}r{3pt}){5-7} \cmidrule(l{3pt}r{3pt}){8-10}
 & $\Delta p_A$ & $\Delta p_B$ & $\Delta p_{AB}$ & $\Delta p_A$ & $\Delta p_B$ & $\Delta p_{AB}$ & $\Delta p_A$ & $\Delta p_B$ & $\Delta p_{AB}$\\
\midrule
$s_A\leftarrow 0.90\,s_A$ & -0.002 & 0.002 & -0.002 & -0.002 & 0.002 & -0.002 & -0.002 & 0.002 & -0.002\\
$s_A\leftarrow 0.70\,s_A$ & -0.005 & 0.005 & -0.005 & -0.006 & 0.006 & -0.005 & -0.006 & 0.006 & -0.005\\
$s_A\leftarrow 0.50\,s_A$ & -0.008 & 0.008 & -0.008 & -0.009 & 0.008 & -0.008 & -0.009 & 0.008 & -0.008\\
\bottomrule
\end{tabular}
}

\vspace{0.4cm}

\begin{tablenotes}
    \footnotesize \item \textit{Notes:} Model-implied changes in adoption shares at horizons $h=1,5,10$ when weakening social interactions for norm $A$ (PGY) by scaling $s_A$ and simultaneously reducing $s_{AB}$ as indicated, holding all other parameters and initial conditions fixed. Total adoption is $Q_{\mathrm{tot}}=p_A+p_B+p_{AB}$, so $\Delta Q_{\mathrm{tot}}=\Delta p_A+\Delta p_B+\Delta p_{AB}$.
\end{tablenotes}
\end{table}

\section{Conclusions}

We propose to model social norm adherence  as a choice among \textit{bundles}, extending the framework of discrete choice with social interactions to allow for non-exclusive actions. We model two dimensions of complementarities and substitutabilities among norms: a `technological' and a `social' one. Technological complementarities occur when the intrinsic utility from joint adoption of two norms is super-additive, while social complementarities arise when the conformity returns to joint adoption exceed the sum of the social returns of conforming to each norm in isolation. We characterize the conditions for overall complementarity or substitutability in a way that nests \citet{Gentzkow_AER_2007} as a special case without social spillovers.

We bring the model to the data using repeated cross sections from the Demographic Health Surveys for Sierra Leone and Nigeria, covering cohorts born between 1960 and 2010. We consider three norms---child marriage (CMA), female genital cutting (FGC) and polygyny (PGY)---and focus on adoption bundles (FGC, CMA) and (PGY, CMA). For identification, we exploit child marriage bans imposed in the two countries and, for Nigeria, also variation in droughts, taking these as exogenous shifters of the intrinsic utility of child marriage. 

In both countries, we find positive and significant social multipliers for each of the three norms, confirming that social conformity is a significant determinant of norm adoption. For the bundle (FGC, CMA), we estimate a positive and significant technological complementarity---especially in Sierra Leone---and linear social returns, resulting in overall \textit{complementarity} among these norms. We interpret the strong complementarity found in Sierra Leone as arising from the tradition of ``Bondo," a collective initiation ritual that culminates in the cutting ceremony and marks adolescent girls' transition to adulthood and potentially marriage. On the other hand, for the (PGY, CMA) bundle, we find insignificant technological complementarities and sub-additive social returns in both countries, implying that the two norms are \textit{substitutes} in our framework. This result can be understood by considering the lower status associated with entering a polygynous union as a non-first wife, which can make child marriage less attractive.

For each of these norm bundles, we conduct policy counterfactual exercises that speak directly to the design of norm-change interventions. We consider two types of policies: (i) shifts to the intrinsic utility of each norm, as would occur with a legal ban or an information campaign targeting intrinsic preferences; and (ii) reductions in the social multipliers, as would result from targeting the institutional mechanisms through which conformity sanctions are imposed.  In the case of (FGC. CMA) in Sierra Leone, we find that targeting FGC generates large reductions in equilibrium adherence to \textit{both} norms, due to the complementarity and to the high social multiplier estimated for FGC. In the case of (PGY, CMA),  substitutability limits the cross-norm benefits from intervening on any of the two norms in isolation and may even generate unintended reallocation effects.

Our paper opens several directions for future research. Although we study pairs of norms, adherence decisions in practice may involve larger sets of norms. Extending the model to higher-dimensional bundles presents computational challenges, but may generate interesting dynamics. Also, extending the empirical analysis to a larger number of countries, and mapping the parameter estimates to sociological and anthropological evidence on the significance of apparently similar norms across different societies would be fascinating. Finally, while our empirical applications concern gender norms in Africa, our framework can be applied to different types of norms and---in a much broader sense---to any choice setting where individual utility includes intrinsic and social payoffs and where social spillovers are non-negligible.

\bibliography{ref}

@incollection{kline2020econometric,
  title={Econometric analysis of models with social interactions},
  author={Kline, Brendan and Tamer, Elie},
  booktitle={The Econometric Analysis of Network Data},
  pages={149--181},
  year={2020},
  publisher={Elsevier}
}

@article{guiteras2019demand,
  title={Demand estimation with strategic complementarities: Sanitation in Bangladesh},
  author={Guiteras, Raymond and Levinsohn, James and Mobarak, Ahmed Mushfiq},
  year={2019},
  publisher={CEPR Discussion Paper No. DP13498}
}

@article{jackson2007diffusion,
  title={Diffusion of behavior and equilibrium properties in network games},
  author={Jackson, Matthew O and Yariv, Leeat},
  journal={American Economic Review},
  volume={97},
  number={2},
  pages={92--98},
  year={2007},
  publisher={American Economic Association}
}

@article{bhattacharya2024demand,
  title={Demand and welfare analysis in discrete choice models with social interactions},
  author={Bhattacharya, Debopam and Dupas, Pascaline and Kanaya, Shin},
  journal={Review of Economic Studies},
  volume={91},
  number={2},
  pages={748--784},
  year={2024},
  publisher={Oxford University Press US}
}

@misc{athey1998empirical,
  title={An empirical framework for testing theories about complimentarity in organizational design},
  author={Athey, Susan and Stern, Scott},
  year={1998},
  publisher={National Bureau of Economic Research Cambridge, Mass., USA}
}

@article{berry2014structural,
  title={Structural models of complementary choices},
  author={Berry, Steve and Khwaja, Ahmed and Kumar, Vineet and Musalem, Andres and Wilbur, Kenneth C and Allenby, Greg and Anand, Bharat and Chintagunta, Pradeep and Hanemann, W Michael and Jeziorski, Przemek and others},
  journal={Marketing Letters},
  volume={25},
  number={3},
  pages={245--256},
  year={2014},
  publisher={Springer}
}

@book{EvansGariepy2015,
  author    = {Evans, Lawrence C. and Gariepy, Ronald F.},
  title     = {Measure Theory and Fine Properties of Functions},
  edition   = {Revised},
  publisher = {CRC Press},
  year      = {2015},
  note      = {See Theorem 3.10 (Coarea formula).}
}

@article{shell2011dynamics,
  title={Dynamics of change in the practice of female genital cutting in Senegambia: Testing predictions of social convention theory},
  author={Shell-Duncan, Bettina and Wander, Katherine and Hernlund, Ylva and Moreau, Amadou},
  journal={Social science \& medicine},
  volume={73},
  number={8},
  pages={1275--1283},
  year={2011},
  publisher={Elsevier}
}

@article{Gentzkow_AER_2007,
Author = {Gentzkow, Matthew},
Title = {Valuing New Goods in a Model with Complementarity: Online Newspapers},
Journal = {American Economic Review},
Volume = {97},
Number = {3},
Year = {2007},
Month = {June},
Pages = {713–744}
}

@article{Gulesci_etal_AER_2025,
Author = {Gulesci, Selim and Jindani, Sam and La Ferrara, Eliana and Smerdon, David and Sulaiman, Munshi and Young, Peyton},
Title = {A Stepping Stone Approach to Norm Transitions},
Journal = {American Economic Review},
Volume = {115},
Number = {7},
Year = {2025},
Month = {July},
Pages = {2237–66}}

@article{Brock_Durlauf_Restud_2001,
    author = {Brock, William A. and Durlauf, Steven N.},
    title = {Discrete Choice with Social Interactions},
    journal = {The Review of Economic Studies},
    volume = {68},
    number = {2},
    pages = {235-260},
    year = {2001},
    month = {04}
}

@article{fehr_2000,
  title={Cooperation and punishment in public goods experiments},
  author={Fehr, Ernst and G{\"a}chter, Simon},
  journal={American Economic Review},
  volume={90},
  number={4},
  pages={980--994},
  year={2000},
  publisher={American Economic Association}
}

@article{kandori_1992,
  title={Social norms and community enforcement},
  author={Kandori, Michihiro},
  journal={The Review of Economic Studies},
  volume={59},
  number={1},
  pages={63--80},
  year={1992},
  publisher={Wiley-Blackwell}
}

@article{alesina2002trusts,
  title={Who trusts others?},
  author={Alesina, Alberto and La Ferrara, Eliana},
  journal={Journal of Public Economics},
  volume={85},
  number={2},
  pages={207--234},
  year={2002},
  publisher={Elsevier}
}

@article{guiso2004role,
  title={The role of social capital in financial development},
  author={Guiso, Luigi and Sapienza, Paola and Zingales, Luigi},
  journal={American Economic Review},
  volume={94},
  number={3},
  pages={526--556},
  year={2004},
  publisher={American Economic Association}
}

@article{bisin2001economics,
  title={The economics of cultural transmission and the dynamics of preferences},
  author={Bisin, Alberto and Verdier, Thierry},
  journal={Journal of Economic Theory},
  volume={97},
  number={2},
  pages={298--319},
  year={2001},
  publisher={Elsevier}
}

@article{fernandez2009culture,
  title={Culture: An empirical investigation of beliefs, work, and fertility},
  author={Fern{\'a}ndez, Raquel and Fogli, Alessandra},
  journal={American Economic Journal: Macroeconomics},
  volume={1},
  number={1},
  pages={146--177},
  year={2009},
  publisher={American Economic Association}
}

@article{jakiela2016does,
  title={Does Africa need a rotten kin theorem? Experimental evidence from village economies},
  author={Jakiela, Pamela and Ozier, Owen},
  journal={The Review of Economic Studies},
  volume={83},
  number={1},
  pages={231--268},
  year={2016},
  publisher={Oxford University Press}
}

@article{fehr1999theory,
  title={A theory of fairness, competition, and cooperation},
  author={Fehr, Ernst and Schmidt, Klaus M},
  journal={The Quarterly Journal of Economics},
  volume={114},
  number={3},
  pages={817--868},
  year={1999},
  publisher={MIT press}
}

@article{field2008early,
  title={Early marriage, age of menarche, and female schooling attainment in Bangladesh},
  author={Field, Erica and Ambrus, Attila},
  journal={Journal of Political Economy},
  volume={116},
  number={5},
  pages={881--930},
  year={2008},
  publisher={The University of Chicago Press}
}

@article{corno2020age,
  title={Age of marriage, weather shocks, and the direction of marriage payments},
  author={Corno, Lucia and Hildebrandt, Nicole and Voena, Alessandra},
  journal={Econometrica},
  volume={88},
  number={3},
  pages={879--915},
  year={2020},
  publisher={Wiley Online Library}
}

@article{buchmann2023signal,
  title={A signal to end child marriage: Theory and experimental evidence from Bangladesh},
  author={Buchmann, Nina and Field, Erica and Glennerster, Rachel and Nazneen, Shahana and Wang, Xiao Yu},
  journal={American Economic Review},
  volume={113},
  number={10},
  pages={2645--2688},
  year={2023},
  publisher={American Economic Association 2014 Broadway, Suite 305, Nashville, TN 37203}
}

@article{jayachandran2021microentrepreneurship,
  title={Microentrepreneurship in developing countries},
  author={Jayachandran, Seema},
  journal={Handbook of Labor, Human resources and Population Economics},
  pages={1--31},
  year={2021},
  publisher={Springer}
}

@article{efferson2015fgc,
  title={Female genital cutting is not a social coordination norm},
  author={Efferson, Charles and Vogt, S and Elhadi, A and Ahmed, H and Fehr, Ernst},
  journal={Science},
  volume={349},
  number={6255},
  pages={1446--7},
  year={2015}
}

@article{tabellini_2008scope,
  title={The scope of cooperation: Values and incentives},
  author={Tabellini, Guido},
  journal={The Quarterly Journal of Economics},
  volume={123},
  number={3},
  pages={905--950},
  year={2008},
  publisher={MIT Press}
}

@article{townsend1994risk,
  title={Risk and insurance in village India},
  author={Townsend, Robert M},
  journal={Econometrica},
  pages={539--591},
  year={1994},
  publisher={JSTOR}
}

@article{kandel1992peer,
  title={Peer pressure and partnerships},
  author={Kandel, Eugene and Lazear, Edward P},
  journal={Journal of Political Economy},
  volume={100},
  number={4},
  pages={801--817},
  year={1992},
  publisher={The University of Chicago Press}
}

@article{ichino2000work,
  title={Work environment and individual background: Explaining regional shirking differentials in a large Italian firm},
  author={Ichino, Andrea and Maggi, Giovanni},
  journal={The Quarterly Journal of Economics},
  volume={115},
  number={3},
  pages={1057--1090},
  year={2000},
  publisher={MIT Press}
}

@article{almaas2020cutthroat,
  title={Cutthroat capitalism versus cuddly socialism: Are Americans more meritocratic and efficiency-seeking than Scandinavians?},
  author={Alm{\aa}s, Ingvild and Cappelen, Alexander W and Tungodden, Bertil},
  journal={Journal of Political Economy},
  volume={128},
  number={5},
  pages={1753--1788},
  year={2020},
  publisher={The University of Chicago Press Chicago, IL}
}

@article{cappelen2020fair,
  title={Fair and unfair income inequality},
  author={Cappelen, Alexander W and Falch, Ranveig and Tungodden, Bertil},
  journal={Handbook of Labor, Human Resources and Population Economics},
  pages={1--25},
  year={2020},
  publisher={Springer}
}

@article{henrich2001search,
  title={In search of homo economicus: behavioral experiments in 15 small-scale societies},
  author={Henrich, Joseph and Boyd, Robert and Bowles, Samuel and Camerer, Colin and Fehr, Ernst and Gintis, Herbert and McElreath, Richard},
  journal={American Economic Review},
  volume={91},
  number={2},
  pages={73--78},
  year={2001},
  publisher={American Economic Association}
}

@article{greif1993contract,
  title={Contract enforceability and economic institutions in early trade: The Maghribi traders' coalition},
  author={Greif, Avner},
  journal={The American Economic Review},
  pages={525--548},
  year={1993},
  publisher={JSTOR}
}

@article{goolsbee2004consumer,
  title={The consumer gains from direct broadcast satellites and the competition with cable TV},
  author={Goolsbee, Austan and Petrin, Amil},
  journal={Econometrica},
  volume={72},
  number={2},
  pages={351--381},
  year={2004},
  publisher={Wiley Online Library}
}

@article{bursztyn2018status,
  title={Status goods: experimental evidence from platinum credit cards},
  author={Bursztyn, Leonardo and Ferman, Bruno and Fiorin, Stefano and Kanz, Martin and Rao, Gautam},
  journal={The Quarterly Journal of Economics},
  volume={133},
  number={3},
  pages={1561--1595},
  year={2018},
  publisher={Oxford University Press}
}

@article{imas2024superiority,
  title={Superiority-seeking and the preference for exclusion},
  author={Imas, Alex and Madar{\'a}sz, Krist{\'o}f},
  journal={Review of Economic Studies},
  volume={91},
  number={4},
  pages={2347--2386},
  year={2024},
  publisher={Oxford University Press UK}
}

@techreport{Bargain_2026,
  author      = {Aminjonov, Ulugbek and Bargain, Olivier},
  title       = {Ancestral Norm Mixtures and Women’s Empowerment},
  year        = {2026},
  month       = {April},
  type        = {Unpublished},
  institution = {Bordeaux School of Economics}
}

@book{sandholm_2010,
	address = {Cambridge, MA},
	author = {Sandholm, William H},
	publisher = {MIT Press},
	title = {Population games and evolutionary dynamics},
	year = {2010}
}

@article{bursztyn_socialmedia2025,
  title={When Product Markets Become Collective Traps: The Case of Social Media},
  author={Bursztyn, Leonardo and Handel, Benjamin and Jiménez-Durán, Rafael and Roth, Chris},
  journal={American Economic Review},
  volume={115},
  number={12},
  pages={4105–-36},
  year={2025}
}

@incollection{LaFerrara&Yanagizawa_2026,
  author    = {La Ferrara, Eliana and Yanagizawa-Drott, David},
  title     = {Changing Culture and Norms in Developing Countries},
  booktitle = {The Handbook of Culture and Economic Behavior},
  editor    = {Enke, Benjamin and Giuliano, Paola and Nunn, Nathan and Wantchekon, Leonard},
  year      = {forthcoming},
  publisher = {Elsevier}
}

@incollection{Field_etal2026,
  author    = {Field, Erica and McKelway, Madeline and Voena, Alessandra},
  title     = {Gender Norms and Development},
  booktitle = {THandbook of Development Economics Vol. 6},
  editor    = {Dupas, Pascaline and Goldberg, Penny and Pande, Rohini},
  year      = {forthcoming},
  publisher = {Elsevier}
}

\appendix

\counterwithin{figure}{section}
\counterwithin{table}{section}

\section{Additional Empirical Results}

\subsection{FGC $\times$ CMA in Nigeria} \label{appendix_section_NG_FGC_CMA}

\paragraph{Additional exclusion restrictions: CMA ban and drought} A key advantage of the Nigerian setting is that exposure to drought provides an additional shifter to CMA adoption. Unlike in Sierra Leone, where limited spatial variation makes drought largely uninformative, droughts in Nigeria  plausibly operate through local economic conditions and the marriage market, affecting incentives for early marriage. Our use of drought as part of the exclusion restriction assumes that these shocks have no effect on the baseline utility of FGC. We therefore report estimates under two identification strategies: exploiting state-level variation in the timing of CMA regulation and using drought variation alone.

Table~\ref{tab:ng_main_merged_threegroups} reports the estimated model parameters for Nigeria under three identification strategies: (i) using only drought variation (excluding the CMA-ban shifter),  (ii) using the CMA-ban shifter alone, and (iii) using both the CMA-ban and drought shifters. Across all three columns we report the unrestricted specification that allows for non-additive sanctions by estimating $(s_A,s_B,s_{AB})$. Results across all specifications are consistent both in significance and point estimates. 
 The CMA-ban shifter enters only the CMA utility and is estimated to be negative and precisely estimated when included: the ban coefficient is $-0.08$ with drought controls and $-0.09$ without drought, both statistically significant at the 1\% level, consistent with the results found in Sierra Leone. Drought exposure over the age of twelve-fifteen enters positively in the CMA equation when included with an effect of about $0.04$--$0.05$ (s.e.\ $0.01$), and is also highly statistically significant. A positive drought–CMA relationship captures settings where droughts are negative income shocks (e.g. due to reliance on agriculture as the main source of income) and where the incentive to cash in the bride price leads to anticipating the timing of a daughter's marriage \citep{corno2020age}. 

\paragraph{Complementarities} Relative to Sierra Leone, where $\Gamma\approx 1$ for FGC$\times$CMA, complementarities in Nigeria are positive but significantly smaller.\footnote{Similarly, the estimated \(\rho\) is modestly negative, ranging from about
\(-0.08\) to \(-0.12\). It is statistically significant at the five percent level
only in the drought-only specification.} 
This is consistent with the differences in institutional settings highlighted in Section \ref{sec:norms_background}: in Sierra Leone the institution of the Bondo society de facto `bundles' the practice of cutting and the preparation for marriage around adolescence, and cutting precedes marriage because Bondo is the ritual that marks the transition to adulthood. In Nigeria, on the other hand, the two decisions are temporally more distant and there is significant variation in the geographic prevalence of the practice. For instance, in Nigeria the ethnic groups with the highest prevalence of child marriage are the Hausa-Fulani (in the North), while the groups with the highest FGC rates are the Yoruba and the Igbo (in the South).\footnote{The shares of compliers with each norm, $p_A, p_B, p_AB$ are computed for each cohort at the ethnic group $\times$ region level, hence we are accounting for geographic and group-level differences in prevalence.} 

Given that, at least in Nigeria, ideals of purity are often associated with the practice of FGC, one may wonder why FGC and CMA are not substitutes, to the extent that child marriage also contributes to purity or chastity by reducing the risk of pre-marital sex. Several factors may explain this. First, the two norms act on different margins of the `purity' ideal: FGC is supposed to control \textit{desire} and reduce the risk of extra-marital sex (before and after the marriage), while CMA controls \textit{opportunity} by `transferring' the girl to the husband's household and effectively relieving the parents of the risk of honor's breach. Second, FGC and CMA serve other functions that overlap only partially. For example, the cultural identity component associated with FGC is not served by CMA, and the role of CMA in securing bride price for the family of origin, or helping forge alliances with other lineages, is not served by FGC.

\paragraph{Social multipliers} Finally, note that the estimated social multipliers are large for both norms in Nigeria: $s_A$ is approximately $2.05$ in all three specifications, while $s_B$ ranges from $0.92$ to $1.09$. Compared to Sierra Leone---where FGC exhibits a large multiplier while CMA’s multiplier is comparatively small---Nigeria exhibits stronger social feedback for CMA as well. Allowing for non-additive sanctions yields estimates of $s_{AB}$ that are close to $s_A+s_B$, suggesting also in this context close to linear social spillovers.

\begin{table}[h]
\centering
\caption{Nigeria (FGC$\times$CMA): estimates with CMA-ban and drought shifters.}
\label{tab:ng_main_merged_threegroups}
\resizebox{\ifdim\width>\linewidth\linewidth\else\width\fi}{!}{%
\begin{tabular}{lcccccc}
\toprule
 \multicolumn{1}{c}{\textbf{Nigeria}}  & \multicolumn{6}{c}{FGC $\times$ CMA} \\
\cmidrule(l{3pt}r{3pt}){2-7}
 & \multicolumn{2}{c}{Drought only (no CMA ban)} & \multicolumn{2}{c}{CMA ban (no drought)} & \multicolumn{2}{c}{CMA ban + drought} \\
\cmidrule(l{3pt}r{3pt}){2-3} \cmidrule(l{3pt}r{3pt}){4-5} \cmidrule(l{3pt}r{3pt}){6-7}
 & (1) & (2) & (3) & (4) & (5) & (6) \\
\midrule
CMA Ban &  &  &  & $-0.09^{***}$ &  & $-0.08^{***}$ \\
 &  &  &  & $(0.01)$ &  & $(0.01)$ \\
Drought &  & $0.05^{***}$ &  &  &  & $0.04^{***}$ \\
 &  & $(0.01)$ &  &  &  & $(0.01)$ \\
$\Gamma$ & \multicolumn{2}{c}{$0.19^{***}$} & \multicolumn{2}{c}{$0.16^{***}$} & \multicolumn{2}{c}{$0.15^{**}$} \\
 & \multicolumn{2}{c}{$(0.06)$} & \multicolumn{2}{c}{$(0.06)$} & \multicolumn{2}{c}{$(0.06)$} \\
$s$ & $2.05^{***}$ & $1.09^{***}$ & $2.05^{***}$ & $0.95^{***}$ & $2.05^{***}$ & $0.92^{***}$ \\
 & $(0.05)$ & $(0.09)$ & $(0.05)$ & $(0.10)$ & $(0.05)$ & $(0.10)$ \\
$s_{AB}$ &  \multicolumn{2}{c}{$3.05^{***}$}   & \multicolumn{2}{c}{$2.93^{***}$} &  \multicolumn{2}{c}{$2.90^{***}$} \\
 &  \multicolumn{2}{c}{$(0.11)$} &  \multicolumn{2}{c}{$(0.12)$} &  \multicolumn{2}{c}{$(0.12)$} \\
$\rho$ & \multicolumn{2}{c}{$-0.12^{**}$} & \multicolumn{2}{c}{$-0.10^{*}$} & \multicolumn{2}{c}{$-0.08$} \\
 & \multicolumn{2}{c}{$(0.06)$} & \multicolumn{2}{c}{$(0.06)$} & \multicolumn{2}{c}{$(0.06)$} \\
 \hline
Primary Education & $0.06^{***}$ & $-0.11^{***}$ & $0.06^{***}$ & $-0.11^{***}$ & $0.06^{***}$ & $-0.11^{***}$ \\
 & $(0.02)$ & $(0.02)$ & $(0.02)$ & $(0.02)$ & $(0.02)$ & $(0.02)$ \\
Secondary Education & $-0.09^{***}$ & $-0.81^{***}$ & $-0.10^{***}$ & $-0.79^{***}$ & $-0.10^{***}$ & $-0.79^{***}$ \\
 & $(0.03)$ & $(0.02)$ & $(0.03)$ & $(0.02)$ & $(0.03)$ & $(0.02)$ \\
Muslim & $0.23^{***}$ & $0.22^{***}$ & $0.24^{***}$ & $0.23^{***}$ & $0.24^{***}$ & $0.23^{***}$ \\
 & $(0.02)$ & $(0.02)$ & $(0.02)$ & $(0.02)$ & $(0.02)$ & $(0.02)$ \\
Christian & $0.10^{***}$ & $-0.03$ & $0.10^{***}$ & $-0.03$ & $0.10^{***}$ & $-0.03$ \\
 & $(0.02)$ & $(0.02)$ & $(0.02)$ & $(0.02)$ & $(0.02)$ & $(0.02)$ \\
Urban & $0.04^{***}$ & $-0.13^{***}$ & $0.03^{***}$ & $-0.13^{***}$ & $0.03^{**}$ & $-0.13^{***}$ \\
 & $(0.01)$ & $(0.01)$ & $(0.01)$ & $(0.01)$ & $(0.01)$ & $(0.01)$ \\
\midrule
Ethnicity FE & Yes & Yes & Yes & Yes & Yes & Yes \\
Region FE & Yes & Yes & Yes & Yes & Yes & Yes \\
Wealth FE & Yes & Yes & Yes & Yes & Yes & Yes \\
\midrule  p-value for $H_0: s_{AB} = s_A + s_B$ & \multicolumn{2}{c}{ 0.23 } & \multicolumn{2}{c}{  0.35} & \multicolumn{2}{c}{ 0.37 } \\
\midrule
Sample size & \multicolumn{2}{c}{ 71067 } & \multicolumn{2}{c}{ 71067 } & \multicolumn{2}{c}{ 71067 } \\
\bottomrule
\end{tabular}
}
\end{table}

\subsection{PGY $\times$ CMA in Sierra Leone} \label{appendix_section_SL_PGY_CMA}

Table~\ref{tab:sl_main1_nodrought_pgy} reports the estimated parameters for the joint adoption of polygyny (PGY) and child marriage (CMA) in Sierra Leone. The CMA-ban shifter enters negatively in the CMA equation, with an effect of about $-0.04$ standard deviations in both specifications. Turning to complementarities, once non-additive sanctions are allowed, $\Gamma$ becomes smaller in magnitude and statistically indistinguishable from zero ($\Gamma=-0.07$, s.e.\ $0.06$). 
Social interactions, however, are economically relevant: both PGY and CMA feature positive and statistically significant social multipliers (e.g., $s_A\approx 1.30$--$1.45$ for PGY and $s_B\approx 0.55$--$0.56$ for CMA), and we reject additive sanctions($p$-value $0.02$), with $s_{AB}=1.44$ (s.e.\ $0.29$) differing from $s_A+s_B$. 

As for the case of Nigeria, the fact that $\Gamma \approx 0$ and $s_{AB}<s_A+s_B$ indicates that CMA and PGY are social (and overall) \textit{substitutes}. The rationale offered for this result in Nigeria (due to the undesirability of entering a polygynous union as a second or third wife for a young girl) also applies to Sierra Leone. However, compared to the estimated coefficients for Nigeria  in Table~\ref{tab:ng_main1_nodrought_pgy}, the Sierra Leone estimates point to smaller social spillovers: in Nigeria we estimate substantially larger multipliers for both PGY and CMA (e.g., $s_A\approx 2.3$--$2.45$ and $s_B\approx 1.25$), while in Sierra Leone both are roughly half as large. A possible reason may be that in Sierra Leone societal pressure -whether as rewards or as sanctions- is highest for compliance with the Bondo practice (FGC); norms like CMA and PGY are less formally sanctioned by the traditional system.

\begin{table}[!htbp]
\caption{Sierra Leone (PGY$\times$CMA): Estimated parameters}
\label{tab:sl_main1_nodrought_pgy}
\centering
\resizebox{\ifdim\width>\linewidth\linewidth\else\width\fi}{!}{
\begin{tabular}{lcccc}
\toprule
 \multicolumn{1}{c}{\textbf{Sierra Leone}} & \multicolumn{4}{c}{PGY $\times$ CMA} \\
\cmidrule(l{3pt}r{3pt}){2-5}
 & \multicolumn{2}{c}{$s_A, s_B$} & \multicolumn{2}{c}{$s_A, s_B, s_{AB}$} \\
\cmidrule(l{3pt}r{3pt}){2-3} \cmidrule(l{3pt}r{3pt}){4-5}
 &  PGY & CMA  & PGY  & CMA  \\
 & (1) & (2) & (3) & (4) \\
\midrule
CMA Ban & \multicolumn{2}{c}{$-0.04^{*}$} & \multicolumn{2}{c}{$-0.04^{**}$} \\
 & \multicolumn{2}{c}{$(0.02)$} & \multicolumn{2}{c}{$(0.02)$} \\
$\Gamma$ & \multicolumn{2}{c}{$-0.17^{**}$} & \multicolumn{2}{c}{$-0.07$} \\
 & \multicolumn{2}{c}{$(0.07)$} & \multicolumn{2}{c}{$(0.06)$} \\
$s$ & $1.30^{***}$ & $0.55^{***}$ & $1.45^{***}$ & $0.56^{***}$ \\
 & $(0.21)$ & $(0.13)$ & $(0.22)$ & $(0.13)$ \\
$s_{AB}$& \multicolumn{2}{c}{} & \multicolumn{2}{c}{$1.44^{***}$} \\
 & \multicolumn{2}{c}{} & \multicolumn{2}{c}{$(0.29)$} \\
$\rho$ & \multicolumn{2}{c}{$0.24^{***}$} & \multicolumn{2}{c}{$0.19^{***}$} \\
 & \multicolumn{2}{c}{$(0.07)$} & \multicolumn{2}{c}{$(0.06)$} \\
 \hline
Primary Education & $-0.14^{***}$ & $-0.08^{***}$ & $-0.14^{***}$ & $-0.08^{***}$ \\
 & $(0.03)$ & $(0.02)$ & $(0.03)$ & $(0.02)$ \\
Secondary Education & $-0.57^{***}$ & $-0.70^{***}$ & $-0.56^{***}$ & $-0.70^{***}$ \\
 & $(0.03)$ & $(0.02)$ & $(0.03)$ & $(0.02)$ \\
Muslim & $0.01$ & $-0.03$ & $0.02$ & $-0.03$ \\
 & $(0.13)$ & $(0.11)$ & $(0.13)$ & $(0.11)$ \\
Christian & $-0.26^{*}$ & $-0.15$ & $-0.26^{*}$ & $-0.14$ \\
 & $(0.13)$ & $(0.11)$ & $(0.13)$ & $(0.11)$ \\
Urban & $-0.27^{***}$ & $-0.08^{***}$ & $-0.27^{***}$ & $-0.08^{***}$ \\
 & $(0.03)$ & $(0.02)$ & $(0.03)$ & $(0.02)$ \\
\midrule
Ethnicity FE & Yes & Yes & Yes & Yes \\
Region FE & Yes & Yes & Yes & Yes \\
Wealth FE & Yes & Yes & Yes & Yes \\
\midrule  p-value for $H_0: s_{AB} = s_A + s_B$ & \multicolumn{2}{c}{} &  \multicolumn{2}{c}{ $0.02^{**} $  } \\
\midrule
Sample size & \multicolumn{2}{c}{ 33591} & \multicolumn{2}{c}{ 33591} \\
\bottomrule
\end{tabular}}
\vspace{0.4cm}

\begin{tablenotes}
    \footnotesize \item \textit{Notes:} Maximum-likelihood estimates from the two-norm adoption model for PGY$\times$CMA. Columns labeled $s_A,s_B$ impose additive sanctions ($s_{AB}=s_A+s_B$); columns labeled $s_A,s_B,s_{AB}$ additionally estimate $s_{AB}$. The exclusion restriction is the Child Marriage (CMA) ban for cohorts born after 1990, which enters only the CMA utility but not the PGY utility. Reference groups for social interactions are defined by region$\times$ethnicity. All specifications control for education, religion, urban residence, and ethnicity, region, and wealth fixed effects. Standard errors in parentheses.
\end{tablenotes}
\end{table}

\newpage

\setcounter{section}{0}
\renewcommand{\thesection}{S\arabic{section}}

\section*{\centering Supplemental Appendix}

 \section{Proof of Theorem \ref{prop_gamma_fixed_point_extended}} \label{proof_gamma_fixed_point_extended}


\paragraph{Step 1: Notation} 
Recall the definition of $\tilde{p}(p,\delta) = \Big(\tilde p_A(p), \tilde p_B(p), \tilde p_{AB} (p)\Big)^\top$, which we write here as an explicit function of $\delta$ in addition to $p$. Recall also that  $p^*$ is implicitly defined by
$$
\begin{bmatrix} \tilde{p}_A(p^*) \\ 
\tilde{p}_B(p^*) \\ 
\tilde{p}_{AB}(p^*)
\end{bmatrix}
= p^*, 
$$
denoting the equilibrium adherence proportions of each bundle.  Let $\tilde{\Lambda}(p, \delta)$ denote the Jacobian of $\tilde{p}(p, \delta)$ with respect to $p$.

Write $p^*(\delta)$ as an explicit function of $\delta$, such that  the implicit function theorem implies (see Lemma \ref{lem:diff_line_integrals} for the characterization of the derivative)
\begin{align}
    \frac{\partial p_i^*(\delta)}{\partial \delta_j}   =  \left[ I_3 - \tilde{\Lambda}(p^*) \right]^{-1}   \frac{\partial \tilde{p}_i(p^*, \delta)}{\partial \delta_j}  (p^*, \delta) , \label{eq_implicit_fxn_2}
\end{align}
for any stable $p^*$ (Definition \ref{ass:stability2}). Note that the above expression is well-defined because $\det(I_3 - \tilde{\Lambda}(p^*))  > 0$ for stable $p^*$.

Denote $||\cdot||$ the $l_2$-norm. Write $u_A, u_B, u_{AB}$ as implicit functions of $\delta = (\delta_A, \delta_B), p = (p_A, p_B, p_{AB}), z= (z_A,z_B)$, $s = (s_A, s_B, s_{AB})$, and $\Gamma$, omitting their dependence on these variables to ease notation whenever it is clear from the context.

\paragraph{Step 2: Defining line integrals}
We now provide a characterization of the Jacobian $\tilde{\Lambda}$ using line integrals. For each $v \in \mathcal{V}$, define $f_v(z) \equiv u_v - \max_{v' \in \mathcal{V} \setminus \{v\} } u_{v'}$ as a explicit function of $z$ (alongside the model parameters), and let
\begin{align}
    L_1^A  & \equiv \{u_A =0, \, u_A \geq u_B, \, u_A \geq u_{AB} \}, \qquad L_1^{AB}  \equiv \{u_{AB} =0, \, u_{AB} \geq u_A, \, u_{AB} \geq u_{B} \} \nonumber \\
    L_2^A & \equiv  \{u_A \geq 0, \, u_A = u_B, \, u_A \geq u_{AB} \}, \qquad  L_2^{AB}  \equiv \{u_{AB} \geq 0, \, u_{AB} = u_A, \, u_{AB} \geq u_{B} \} \nonumber  \\
    L_3^A & \equiv  \{u_A \geq 0, \, u_A \geq u_B, \, u_A =  u_{AB} \}, \qquad L_3^{AB} \equiv \{u_{AB} \geq 0, \, u_{AB} \geq u_A, \, u_{AB} = u_{B} \},  \nonumber 
\end{align}
such that $f_A^{-1}(0) = L_1^A \cup L_2^A \cup L_3^A$ and $f_{AB}^{-1}(0) = L_1^{AB} \cup L_2^{AB} \cup L_3^{AB}$.  Define  $L_i^B$ similarly by exchanging $A$s and $B$s in the definition of $L_i^A$, such that $f_B^{-1}(0) = L_1^B \cup L_2^B \cup L_3^B$. 

Denote $\sigma_i^A$ the surface measure over the set $L_i^A$, $\sigma_i^B$ the surface measure over the set $L_i^B$, and $\sigma_i^{AB}$ the surface measure over $L_i^{AB}$, where the measure is taken over the randomness of $z$. The following lemma is a basic characterization of the line integral following directly from Assumption ~\ref{ass:z_density} and the coarea formula (e.g., Theorem 3.10 in \cite{EvansGariepy2015}). 

\begin{lemma}[Differentiation of adoption probabilities]\label{lem:diff_line_integrals}
Suppose Assumption~\ref{ass:z_density} holds. Fix $(p,\delta)$ and $j\in\{A,B,AB\}$.
Then, $\tilde p_k(p,\delta)$ is differentiable in $p_j$ and
$$  
\begin{aligned} 
\frac{\partial \tilde p_k(p,\delta)}{\partial p_j}
=
\sum_{i=1}^3 \int_{L_i^k}
\frac{\partial f_k/\partial p_j}{\|\nabla_z f_k\|}\,r(z)\,d\sigma_i^k, \qquad \forall k \in \{A, B, AB\}. 
\end{aligned} 
$$ 
\end{lemma}

\begin{proof}

We prove the statement for $k = A$; the case when $k = B, AB$ follows verbatim after appropriate relabeling.  Note that we can write
\[
\tilde p_A(p,\delta)=\int_{\mathbb R^2}
\mathbbm 1\{u_A\ge 0,\ u_A\ge u_B,\ u_A\ge u_{AB}\}\,r(z)\,dz.
\]
Because $r$ is continuous and the utilities are affine in $z$, the probability that $u_A$ equals at least two elements of $\{0,u_B,u_{AB}\}$ is zero. Hence, the intersection of at least  two of the three sets
$L_1^A,L_2^A,L_3^A$ has measure zero with respect to $\sigma_i^A$, and can be ignored from our calculations.

We now study the value of the integral within a given region $L_i^A$. Since $z$ enters utilities only through the linear terms $z_A$, $z_B$, and $z_A+z_B$, $f_A$ equals one of $u_A$, $u_A-u_B$, or $u_A-u_{AB}$ over each $L_i^A$. Hence $f_A(z)=a+b_A z_A+b_B z_B$ with $(b_A,b_B)\in\{(1,0),(1,-1),(0,-1)\}$, so $\nabla_z f_A=(b_A,b_B)$ is constant and nonzero on $L_i^A$. Therefore, the coarea formula yields\footnote{The tail condition in Assumption \ref{ass:z_density} guarantees that the density is
integrable along each affine boundary that defines a choice-region margin. Indeed,
for any affine line $L=\{z_0+tu:t\in\mathbb R\}$  with $\|u\|=1$,
$ 
\int_L r(z)\,d\sigma(z)
=
\int_{\mathbb R} r(z_0+tu)\,dt
\le
C\int_{\mathbb R}(1+\|z_0+tu\|)^{-2-\eta}\,dt
<\infty.
$ 
}
\[
\int_{\{0\le f_A(z)\le \varepsilon\}} r(z)\,dz
=\int_0^\varepsilon\left(\int_{\{f_A(z)=t\}}\frac{r(z)}{\|\nabla_z f_A\|}\,d\sigma(z)\right)\,dt.
\]
Applying this with $\varepsilon=h\,\partial f_A/\partial p_j$, dividing by $h$, and letting $h\to 0$ implies
\[
\frac{\partial}{\partial p_j}\int \mathbbm 1\{f_A(z)\ge 0\}\,r(z)\,dz
=
\int_{\{f_A(z)=0\}}\frac{\partial f_A/\partial p_j}{\|\nabla_z f_A\|}\,r(z)\,d\sigma(z),
\]
and summing over the three boundary pieces $L_1^A,L_2^A,L_3^A$ yields the line-integral expression.\footnote{By Assumption \ref{ass:z_density}, the corresponding probability at the intersection boundary is zero.} 
\end{proof}

\paragraph{Step 3: Evaluating the Jacobian using line integrals}  Using Lemma \ref{lem:diff_line_integrals},  each term in the first row of the Jacobian $\tilde{\Lambda}(p,\delta)$ can be written as
$$  
\begin{aligned}
   \tilde{\Lambda}_{A,j} \equiv \frac{\partial \tilde{p}_A(p , \delta)}{\partial p_j} & = \frac{\partial}{\partial p_j} \int_z \1_{ \{  f_A(z) \geq 0 \} } \, dR(z)  \\
   & = \frac{\partial}{ \partial p_j } \sum_{i \in \{B,AB, \emptyset\}} \int_z \1_{ \{f_A(z) \geq 0, \,\, \max_{v \in \{B, AB, \emptyset\} }  u_v = u_i\}}  \, d R(z) \\ & = \sum_{i=1}^3 \int_{L_i^A} \frac{\partial f_A / \partial p_j}{||\nabla f_A||} r(z) \, d\sigma_i^A \nonumber
\end{aligned}
$$ 
for all $j \in \{A,B,AB\}$. Note that the last equation follows from the definition of the line integral (Lemma \ref{lem:diff_line_integrals}) and $L_i^A$.

 \paragraph{Step 4: Evaluating line segments} Next, we evaluate the integrands over each line segment. On the segment $L_1^A$, 
$ 
f_A = u_A - 0 = \delta_A + s_A(p_A + p_{AB}) + z_A,\quad 
\nabla_z f_A = (1,0),\quad \|\nabla f_A\| = 1,
$ 
so that 
$ 
\frac{\partial f_A}{\partial p_A} = s_A,\quad
\frac{\partial f_A}{\partial p_B} = 0,\quad
\frac{\partial f_A}{\partial p_{AB}} = s_A.
$  
Similarly,  on the segment $L_2^A$, 
$ 
f_A = u_A - u_B
= (\delta_A - \delta_B) + s_A(p_A + p_{AB}) - s_B(p_B + p_{AB})
+ z_A - z_B,
$  
so that 
$  
\nabla_z f_A = (1,-1),\quad \|\nabla f_A\| = \sqrt{2},
$ and 
$ 
\frac{\partial f_A}{\partial p_A} = s_A,\quad
\frac{\partial f_A}{\partial p_B} = -s_B,\quad
\frac{\partial f_A}{\partial p_{AB}} = s_A - s_B.
$  
Finally, on the segment $L^A_3$, $ 
f_A = u_A - u_{AB}
= -\delta_B + (s_A - s_{AB})p_{AB} - s_B p_B - z_B - \Gamma,
$ 
so that 
$  
\nabla_z f_A = (0,-1),\quad \|\nabla f_A\| = 1,
$  
and
$ 
\frac{\partial f_A}{\partial p_A} = 0,\quad
\frac{\partial f_A}{\partial p_B} = -s_B,\quad
\frac{\partial f_A}{\partial p_{AB}} = s_A - s_{AB}.
$  

Collecting terms, it follows that
$$  
\begin{aligned} 
\footnotesize 
\frac{\frac{\partial f_A}{\partial p_A}}{\|\nabla f_A\|}=
\begin{cases}
s_A & \text{over } L^A_1,\\[0.3em]
\dfrac{s_A}{\sqrt{2}} & \text{over } L^A_2,\\[0.3em]
0 & \text{over } L^A_3,
\end{cases}
\qquad
\frac{\frac{\partial f_A}{\partial p_B}}{\|\nabla f_A\|}=
\begin{cases}
0 & \text{over } L^A_1,\\[0.3em]
-\dfrac{s_B}{\sqrt{2}} & \text{over } L^A_2,\\[0.3em]
- s_B & \text{over } L^A_3,
\end{cases}, \qquad 
\frac{\frac{\partial f_A}{\partial p_{AB}}}{\|\nabla f_A\|}=
\begin{cases}
s_A & \text{over } L^A_1,\\[0.3em]
\dfrac{s_A - s_B}{\sqrt{2}} & \text{over } L^A_2,\\[0.3em]
s_A - s_{AB} & \text{over } L^A_3.
\end{cases}
\end{aligned} 
$$ 
By the definition of line integrals, we then have
\[
\begin{aligned}
\tilde{\Lambda}_{A,A}
&= s_A \int_{L^A_1} r(z)\, d\sigma^A_1
+ \frac{s_A}{\sqrt{2}} \int_{L^A_2} r(z)\, d\sigma^A_2,\\[0.5em]
\tilde{\Lambda}_{A,B}
&= -\frac{s_B}{\sqrt{2}} \int_{L^A_2} r(z)\, d\sigma^A_2
- s_B \int_{L^A_3} r(z)\, d\sigma^A_3,\\[0.5em]
\tilde{\Lambda}_{A,AB}
&= s_A \int_{L^A_1} r(z)\, d\sigma^A_1
+ \frac{s_A - s_B}{\sqrt{2}} \int_{L^A_2} r(z)\, d\sigma^A_2
+ (s_A - s_{AB}) \int_{L^A_3} r(z)\, d\sigma^A_3.
\end{aligned}
\]
The calculations for the second row of \(\tilde{\Lambda}\) (derivatives of \(\tilde{p}_B\)) follow analogously by swapping the roles of \(A\) and \(B\), i.e., interchanging \((A,s_A)\) with \((B,s_B)\).

For the third row, we write
\[
\tilde{\Lambda}_{AB,j} \equiv
\frac{\partial \tilde{p}_{AB}(p,\delta)}{\partial p_j}
= \frac{\partial}{\partial p_j} \int 1\{f_{AB}(z)\ge 0\}\, dR(z)
= \sum_{i=1}^3 \int_{L^{AB}_i} 
\frac{\partial f_{AB}/\partial p_j}{\|\nabla f_{AB}\|}\,
r(z)\, d\sigma^{AB}_i.
\]
On the segment $L_1^{AB}$, we have $ 
f_{AB} = u_{AB} - 0
= \delta_A + \delta_B + s_A p_A + s_B p_B + s_{AB} p_{AB} + z_A + z_B + \Gamma,
$  
so that 
$  
\nabla_z f_{AB} = (1,1),\quad \|\nabla f_{AB}\| = \sqrt{2},
$  
and
$  
\frac{\partial f_{AB}}{\partial p_A} = s_A,\quad
\frac{\partial f_{AB}}{\partial p_B} = s_B,\quad
\frac{\partial f_{AB}}{\partial p_{AB}} = s_{AB}.
$  
On the segment $L^{AB}_2$, we have $  
f_{AB} = u_{AB} - u_A
= \delta_B + s_B p_B + (s_{AB} - s_A)p_{AB} + z_B + \Gamma,
$  
so that 
$  
\nabla_z f_{AB} = (0,1),\quad \|\nabla f_{AB}\| = 1,
$  
and 
$  
\frac{\partial f_{AB}}{\partial p_A} = 0,\quad
\frac{\partial f_{AB}}{\partial p_B} = s_B,\quad
\frac{\partial f_{AB}}{\partial p_{AB}} = s_{AB} - s_A.
$  Finally, on the segment $L^{AB}_3$,  we have 
$  
f_{AB} = u_{AB} - u_B
= \delta_A + s_A p_A + (s_{AB} - s_B)p_{AB} + z_A + \Gamma,
$  
hence
$  
\nabla_z f_{AB} = (1,0),\quad \|\nabla f_{AB}\| = 1,
$  
and
$  
\frac{\partial f_{AB}}{\partial p_A} = s_A,\quad
\frac{\partial f_{AB}}{\partial p_B} = 0,\quad
\frac{\partial f_{AB}}{\partial p_{AB}} = s_{AB} - s_B.
$  

Collecting terms, it follows that 
$$ 
\footnotesize 
\begin{aligned} 
\frac{\frac{\partial f_{AB}}{\partial p_A}}{\|\nabla f_{AB}\|}=
\begin{cases}
\dfrac{s_A}{\sqrt{2}} & \text{over } L^{AB}_1,\\[0.3em]
0 & \text{over } L^{AB}_2,\\[0.3em]
s_A & \text{over } L^{AB}_3,
\end{cases}, \quad 
\frac{\frac{\partial f_{AB}}{\partial p_B}}{\|\nabla f_{AB}\|}=
\begin{cases}
\dfrac{s_B}{\sqrt{2}} & \text{over } L^{AB}_1,\\[0.3em]
s_B & \text{over } L^{AB}_2,\\[0.3em]
0 & \text{over } L^{AB}_3,
\end{cases}, \quad 
\frac{\frac{\partial f_{AB}}{\partial p_{AB}}}{\|\nabla f_{AB}\|}=
\begin{cases}
\dfrac{s_{AB}}{\sqrt{2}} & \text{over } L^{AB}_1,\\[0.3em]
s_{AB} - s_A & \text{over } L^{AB}_2,\\[0.3em]
s_{AB} - s_B & \text{over } L^{AB}_3,
\end{cases}
\end{aligned} 
$$ 
and the definition of line integrals yield 
\[
\begin{aligned}
\tilde{\Lambda}_{AB,A}
&= \frac{s_A}{\sqrt{2}}\int_{L^{AB}_1} r(z)\, d\sigma^{AB}_1
+ s_A \int_{L^{AB}_3} r(z)\, d\sigma^{AB}_3,\\[0.5em]
\tilde{\Lambda}_{AB,B}
&= \frac{s_B}{\sqrt{2}}\int_{L^{AB}_1} r(z)\, d\sigma^{AB}_1
+ s_B \int_{L^{AB}_2} r(z)\, d\sigma^{AB}_2,\\[0.5em]
\tilde{\Lambda}_{AB,AB}
&= \frac{s_{AB}}{\sqrt{2}}\int_{L^{AB}_1} r(z)\, d\sigma^{AB}_1
+ (s_{AB} - s_A) \int_{L^{AB}_2} r(z)\, d\sigma^{AB}_2
+ (s_{AB} - s_B) \int_{L^{AB}_3} r(z)\, d\sigma^{AB}_3.
\end{aligned}
\]

\paragraph{Step 5: Explicit expression for the Jacobian} For $
i\in\{1,2,3\}$ and $v\in\{A,B,AB\}$, define
\[
T^v_i \equiv \int_{L^v_i} r(z)\, d\sigma^v_i.
\]
Collecting terms from Step 4, it follows  that 
\begin{equation} \label{eqn:J_tilde} 
\footnotesize 
\begin{aligned}
\tilde{\Lambda} =
\begin{pmatrix}
s_A T^A_1 + \dfrac{s_A}{\sqrt{2}}T^A_2  &
-\dfrac{s_B}{\sqrt{2}}T^A_2 - s_B T^A_3 &
s_A T^A_1 + \dfrac{s_A - s_B}{\sqrt{2}}T^A_2 + (s_A - s_{AB})T^A_3\\[0.6em]
-\dfrac{s_A}{\sqrt{2}}T^B_2 - s_A T^B_3 &
s_B T^B_1 + \dfrac{s_B}{\sqrt{2}}T^B_2 &
s_B T^B_1 + \dfrac{s_B - s_A}{\sqrt{2}}T^B_2 + (s_B - s_{AB})T^B_3\\[0.6em]
\dfrac{s_A}{\sqrt{2}}T^{AB}_1 + s_A T^{AB}_3 &
\dfrac{s_B}{\sqrt{2}}T^{AB}_1 + s_B T^{AB}_2 &
\dfrac{s_{AB}}{\sqrt{2}}T^{AB}_1 + (s_{AB}-s_A)T^{AB}_2 + (s_{AB}-s_B)T^{AB}_3
\end{pmatrix}.
\end{aligned} 
\end{equation} 

\paragraph{Step 6: Properties of the Jacobian} We now characterize properties of the Jacobian. 
Define $\tilde{\Gamma} \equiv \Gamma + p_{AB}(s_{AB} - s_A - s_B)$ as an implicit function of $p$. 
\begin{lemma} \label{lemma_useful_facts}
   The following expressions are true.
    \begin{enumerate}
        \item $T_2^A = T_2^B$, $T_3^A = T_2^{AB}$, and $T_3^B = T_3^{AB}$. 
        \item If $\tilde{\Gamma} > 0$, then $T_2^A = 0$.
        \item If $\tilde{\Gamma} < 0$, then $T_1^{AB} = 0$.
        \item If $\tilde{\Gamma} = 0, T_2^A = T_2^B = T_1^{AB} = 0$. 
    \end{enumerate}
\end{lemma}

\begin{proof}

\medskip
\noindent\textit{Claim 1.}
We show set equalities, which immediately imply equality of the corresponding $T_i^v$. First, on the set $\{u_A = u_B\}$:
\[
L_2^A = \{u_A \ge 0,\ u_A = u_B,\ u_A \ge u_{AB}\}, \quad 
L_2^B = \{u_B \ge 0,\ u_B = u_A,\ u_B \ge u_{AB}\}.
\]
But on $\{u_A = u_B\}$, the inequalities $u_A \ge 0$ and $u_B \ge 0$ are equivalent, and so are
$u_A \ge u_{AB}$ and $u_B \ge u_{AB}$. Hence $L_2^A = L_2^B$ as sets, and therefore
$T_2^A = T_2^B$.

Second, on the set $\{u_A = u_{AB}\}$:
\[
L_3^A = \{u_A \ge 0,\ u_A \ge u_B,\ u_A = u_{AB}\}, \quad 
L_2^{AB} = \{u_{AB} \ge 0,\ u_{AB} = u_A,\ u_{AB} \ge u_B\}.
\]
On $\{u_A = u_{AB}\}$, the inequalities $u_A \ge 0$ and $u_{AB} \ge 0$ coincide, and so do
$u_A \ge u_B$ and $u_{AB} \ge u_B$. Thus $L_3^A = L_2^{AB}$, so $T_3^A = T_2^{AB}$.

Third, on the set $\{u_B = u_{AB}\}$:
\[
L_3^B = \{u_B \ge 0,\ u_B \ge u_A,\ u_B = u_{AB}\},
\quad 
L_3^{AB} = \{u_{AB} \ge 0,\ u_{AB} \ge u_A,\ u_{AB} = u_B\}.
\]
Again, on $\{u_B = u_{AB}\}$ the inequalities $u_B \ge 0$ and $u_{AB} \ge 0$ coincide, as do
$u_B \ge u_A$ and $u_{AB} \ge u_A$, so $L_3^B = L_3^{AB}$ and hence $T_3^B = T_3^{AB}$. This establishes Claim 1.

\medskip
\noindent\textit{Claims 2 and 3.} 
For Claim 2, if $u_A = u_B$ it follows that $u_{AB} = u_A + u_B + \tilde{\Gamma} = 2 u_A + \tilde{\Gamma}$. Clearly $u_A \ge u_{AB}$ and $u_A \ge 0$ is impossible when $\tilde{\Gamma} > 0$. Hence in this case $L_2^A = \emptyset \Rightarrow T_2^A = 0$.

For Claim 3, on $u_{AB} = 0$ with $u_{AB} \ge u_A, u_B$ we need that $u_A \le 0, u_B \le 0 \Rightarrow u_A + u_B \le 0$. If $\tilde{\Gamma} < 0$, since $u_{AB} = u_A + u_B + \tilde{\Gamma}, L_1^{AB} = \emptyset \Rightarrow T_1^{AB} = 0$.

\medskip \noindent\textit{Claim 4.}  When $\tilde{\Gamma} = 0$, $L_2^A = L_2^B = L_1^{AB} = \{u_A = u_B = 0\}$ are degenerate points, which implies, by Assumption \ref{ass:z_density}, Claim 4. 
\end{proof}

\paragraph{Step 7: Derivative characterization}
Note that $\frac{\partial \tilde{p}_i(p,\delta)}{\partial p_j} = s_j \frac{\partial \tilde{p}_i(p,\delta)}{\partial \delta_j}$ for $i \in \{A,B,AB\}$ and $j \in \{A,B\}$ because  $\delta_j$ enters the utilities $u_A, u_B$ as an affine shift, while $p_j$ does so scaled by $s_j$. Substituting Equation \eqref{eqn:J_tilde} into the implicit function Equation \eqref{eq_implicit_fxn_2} and applying Lemma \ref{lemma_useful_facts}, we obtain
\begin{equation} \label{eqn:dd}
\footnotesize 
\begin{aligned}
  \det(I_3 - \tilde{\Lambda}(p^*))  \frac{\partial Q_A(p^*)}{\partial \delta_B} 
   & =    \begin{cases}
       \frac{T_1^{AB}}{\sqrt{2}} +  T_3^A T_3^B(s_{AB} - s_A - s_B)   & \text{if } \tilde \Gamma > 0 \\
       (s_{AB} - s_A - s_B) T_3^A T_3^B & \text{if } \tilde \Gamma = 0 \\
          T_3^A T_3^B(s_{AB} - s_A - s_B)   + \frac{ T_2^B ((s_{AB}-s_A - s_B) (T_3^A+T_3^B)-1)}{\sqrt{2}}  & \text{if } \tilde \Gamma < 0,
   \end{cases}
\end{aligned}
\end{equation} 
by direct computation at the stable equilibrium $p^*$.\footnote{To see this let 
$ 
M(\delta) \equiv I_3 - \tilde{\Lambda}\bigl(p^*(\delta),\delta\bigr),
r_B(\delta) \equiv 
\left.\frac{\partial \tilde{p}(p,\delta)}{\partial \delta_B}\right|_{p=p^*(\delta)}. 
$ 
Since $p^*(\delta)$ is a stable equilibrium, by the implicit function theorem 
$ 
\frac{\partial p^*(\delta)}{\partial \delta_B}
=
M(\delta)^{-1} r_B(\delta).
$  Define $Q_A^*(\delta) = p_A^*(\delta) + p_{AB}^*(\delta)$, so if we let
$e \equiv (1,0,1)^\top$ we can write
$
\frac{\partial Q_A^*(\delta)}{\partial \delta_B}
=
e^\top \frac{\partial p^*(\delta)}{\partial \delta_B}
=
e^\top M(\delta)^{-1} r_B(\delta).
$ 
Now use the adjugate formula for the inverse of a $3\times 3$ matrix:
$ 
M(\delta)^{-1}
=
\frac{\operatorname{adj}\bigl(M(\delta)\bigr)}{\det M(\delta)}.
$  Therefore, $ 
\det\bigl(I_3 - \tilde{\Lambda}\bigr)\,
\frac{\partial Q_A^*(\delta)}{\partial \delta_B}
=
e^\top \operatorname{adj}\bigl(M(\delta)\bigr) r_B(\delta).
$  
Next, because 
$ 
\frac{\partial \tilde{p}(p,\delta)}{\partial p_j}
=
s_j \frac{\partial \tilde{p}(p,\delta)}{\partial \delta_j},
$ 
since $\delta_j$ enters $u_A,u_B$ as an affine shift, while $p_j$ enters
scaled by $s_j$, it follows 
$ 
r_B(\delta)
=
\left.\frac{\partial \tilde{p}(p,\delta)}{\partial \delta_B}\right|_{p=p^*(\delta)}
=
\frac{1}{s_B}\,\tilde{\Lambda}_{\cdot,B}\bigl(p^*(\delta),\delta\bigr),
$  
where $\tilde{\Lambda}_{\cdot,B}$ denotes the $B$-column of $\tilde{\Lambda}$.  Plugging
this into the previous display gives
$ 
\det\bigl(I_3 - \tilde{\Lambda}\bigr)\,
\frac{\partial Q_A^*(\delta)}{\partial \delta_B}
=
\frac{1}{s_B}\,e^\top \operatorname{adj}\bigl(I_3 - \tilde{\Lambda}\bigr)\,
\tilde{\Lambda}_{\cdot,B}.
$ 
By letting  the $1\times 3$ row vector
$ 
V_B^\top(\delta) \equiv
\frac{1}{s_B}\,e^\top \operatorname{adj}\bigl(I_3 - \tilde{\Lambda}(p^*(\delta),\delta)\bigr),
$ 
$
\det\bigl(I_3 - \tilde{\Lambda}\bigr)\,
\frac{\partial Q_A^*(\delta)}{\partial \delta_B}
=
V_B^\top(\delta)\,\tilde{\Lambda}_{\cdot,B}.
$  Computing $V_B$ directly leads to the final expression.
}

\paragraph{Step 8: Strict positivity of $T_3^A,T_3^B$ for non-degenerate distributions} In the following lemma we establish strict positivity of $T_3^A, T_3^B$. 

\begin{lemma}[Geometry of the adoption margins]\label{lemma:T-signs}
Let Assumption \ref{ass:z_density} hold. Then 
\begin{itemize}
\item[1.] For all $\tilde\Gamma\in\mathbb{R}$, we have $T_3^A > 0$ and $T_3^B > 0$.
\item[2.] If in addition $s_{AB} = s_A + s_B$ and $\Gamma > 0$, then $T_1^{AB} > 0$.
\item[3.] If in addition $s_{AB} = s_A + s_B$ and $\Gamma < 0$, then $T_2^B > 0$. 
\end{itemize} 
\end{lemma}

\begin{proof} \textit{Claim 1.}  Let $\mu_A = \delta_A + s_A (p_A + p_{AB})$ and $\mu_B = \delta_B + s_B (p_B + p_{AB})$. 
Under Assumption \ref{ass:z_density},  $T_i^v$ is the integral of this strictly positive density over the one-dimensional set $L_i^v$ with respect to its induced Lebesgue measure. In particular,
$ 
T_i^v > 0
$
if and only $L_i^v$ if has positive Lebesgue measure, and it equals zero otherwise.  

By definition,
$ 
L_3^A = \{u_A \ge 0,\ u_A \ge u_B,\ u_A = u_{AB}\}.
$ 
On the set $\{u_A = u_{AB}\}$, we have
$ 
\mu_A + z_A
=
\mu_A + \mu_B + z_A + z_B + \tilde\Gamma
\quad \Rightarrow \quad
\mu_B + z_B + \tilde\Gamma = 0
\quad \Rightarrow \quad
z_B = -\mu_B - \tilde\Gamma.
$ 
The remaining inequalities become:
(i) $u_A \ge 0$, which implies 
$ 
\mu_A + z_A \ge 0
\quad\Longleftrightarrow\quad
z_A \ge -\mu_A
$; (ii) $u_A \ge u_B$, which implies 
$ 
\mu_A + z_A \ge \mu_B + z_B
=
\mu_B - \mu_B - \tilde\Gamma
=
-\tilde\Gamma
\quad\Longleftrightarrow\quad
z_A \ge -\tilde\Gamma - \mu_A.
$ 
Hence, 
\[
L_3^A
=
\left\{z_B = -\mu_B - \tilde\Gamma,\ 
z_A \ge \max(-\mu_A,\,-\mu_A - \tilde\Gamma)\right\},
\]
which is a half-line in $\mathbb{R}^2$ of the form $\{(z_A,z_B): z_B=c,\ z_A\ge c_0\}$ for some finite constants $c,c_0$ (depending on $\mu_A,\mu_B,\tilde\Gamma$, but not on $z$). This half-line has infinite one-dimensional Lebesgue measure, and since $r(z) > 0$ everywhere and the line integral is finite by Assumption \ref{ass:z_density}, $L_3^A$ has positive one-dimensional measure under $r(z)$, and thus
$ 
T_3^A = \int_{L_3^A} r(z)\,d\sigma_3^A > 0
$ 
for all $\tilde\Gamma\in\mathbb{R}$.
An analogous argument applies to $L_3^B$, such that $T_3^B > 0$.

\medskip 
\noindent \textit{Claim 2.} Recall that $L^{AB}_1=\{u_{AB}=0,\ u_{AB}\ge u_A,\ u_{AB}\ge u_B\}$.  Using $u_{AB}=u_A+u_B+\Gamma$ (since $s_{AB} = s_A + s_B$), the constraint $u_{AB}=0$ becomes $u_A+u_B=-\Gamma$.
The inequalities $0=u_{AB}\ge u_A$ and $0=u_{AB}\ge u_B$ are equivalent to $u_A\le 0$ and $u_B\le 0$.
Hence, 
\[
L^{AB}_1=\{u_A+u_B=-\Gamma,\ u_A\le 0,\ u_B\le 0\}.
\]
If $\Gamma>0$, this set is a nondegenerate line segment: for any $t\in[-\Gamma,0]$,
taking $u_A=t$ and $u_B=-\Gamma-t$ satisfies $u_A\le 0$, $u_B\le 0$, and $u_A+u_B=-\Gamma$.
Since $(u_A,u_B)$ is an affine translate of $(z_A,z_B)$ (see \eqref{eqn:utility_basics}),
$L^{AB}_1$ is also a nondegenerate line segment in $z$-space, hence has strictly positive
measure $\sigma^{AB}_1(L^{AB}_1)>0$. This implies that $T_1^{AB} > 0$ completing the proof of Claim 2. 
\medskip

\noindent\textit{Claim 3.}  Recall that $L^B_2=\{u_B\ge 0,\ u_B=u_A,\ u_B\ge u_{AB}\}$.  On the set $\{u_A=u_B\}$ and using $u_{AB}=u_A+u_B+\Gamma$ (since $s_{AB} = s_A + s_B$), we have $u_{AB}=2u_B+\Gamma$.
Thus, $u_B\ge u_{AB}$ is equivalent to $u_B\ge 2u_B+\Gamma$, i.e.\ $u_B\le -\Gamma$, and it follows that  
\[
L^B_2=\{u_A=u_B,\ 0\le u_B\le -\Gamma\}.
\]
If $\Gamma<0$, the interval $[0,-\Gamma]$ has positive length, so $L^B_2$ is a nondegenerate line segment
(and again an affine translate of a segment in $z$-space). By the same positivity argument as above, $T_2^B > 0$, as desired.  
\end{proof}

\paragraph{Step 9: Proof of the claim for $\Gamma \neq 0$} Whenever $\Gamma \neq 0$, 
since $T_i^v \geq 0$ for all $i \in \{1,2,3\}$ and $v \in \{A,B,AB\}$, we see that if $\Gamma > 0$ and $s_{AB} \ge s_A + s_B$ (and thus $\tilde \Gamma > 0$ for any value of equilibrium shares $p^*$), then  $\sgn(\frac{\partial Q_A^*}{\partial \delta_B} ) > 0$ by Lemma \ref{lemma:T-signs}.  The cases for $\Gamma < 0$, $s_{AB} \le s_A + s_B$ proceeds similarly.  Using the fact that $p^*$ is a stable equilibrium, we have $\mathrm{det}(I_3 - \tilde{\Lambda}) > 0$, completing the proof for the case where $\Gamma \neq 0$.

\paragraph{Step 10: Proof of the claim for $\Gamma = 0$ and $s_{AB} = s_A + s_B$} Follows immediately from the second case in Equation \eqref{eqn:dd}. 

\paragraph{Step 11: Proof  of the claim for $\Gamma = 0$ and $s_{AB} \neq s_A + s_B$} If $p_{AB}^*(\delta) = 0$, the result then follows immediately from the second case in Equation \eqref{eqn:dd}. Else, $p_{AB}^*(\delta) \neq 0$, and the result follows directly from the first or third condition in Equation  \eqref{eqn:dd}, combined with Lemma \ref{lemma:T-signs}.

\section{Additional details on empirical specification}
\label{sec:stacked_se}

\subsection{Intuition behind identification}

A useful way to see the source of identification is to compare the deterministic components of the three utilities in \eqref{eqn:dynamic1}. Conditional on $X_{it}$ and the lagged group state $\big(p_{A(t-1)}^{g(i)},p_{B(t-1)}^{g(i)},p_{AB(t-1)}^{g(i)}\big)$, we have
\[
u_{ABit}-u_{Ait}-u_{Bit}
=
\Gamma+\bigl(s_{AB}-s_A-s_B\bigr)p_{AB(t-1)}^{g(i)}.
\]
In other words, $\Gamma$ enters as an intercept in the surplus from joint adoption relative to the sum of the two single-norm utilities, while $s_{AB}-s_A-s_B$ enters as the slope of that surplus with respect to the lagged probability of joint adoption.

Consider identifying these two components separately. Suppose first that $s_{AB}=s_A=s_B=0$, so that we are in the model of \cite{Gentzkow_AER_2007}. For simplicity, assume there are no additional covariates other than a binary norm-$B$ shifter $X_t^B\in\{0,1\}$ that switches on at time $t$ and remains equal to one thereafter; this is only for expositional convenience, and the same logic applies to more general paths of $X_\tau^B$. Let $\Delta Y_\tau$ denote the difference between the path with the switch and the baseline path without the switch. In this static case,
\[
u_{B\tau}-u_{\emptyset \tau}=\beta X_\tau^B+\tilde z_{B\tau},
\qquad
u_{AB,\tau}-u_{A,\tau}=\beta X_\tau^B+\tilde z_{B\tau}+\Gamma.
\]
Hence, on the switch date,
\[
\Delta(u_{B,t}-u_{\emptyset,t})=\beta,
\qquad
\Delta(u_{AB,t}-u_{A,t})=\beta.
\]
The shifter therefore moves the margins $B$ -vs. $\emptyset$ and $AB$ -vs. $A$. In the data, this appears as a common variation in the corresponding choice probabilities. Since the shifter does not enter the utility of $A$, the likelihood attributes this common shift to $\beta$, while the remaining intercept gap  is attributed to $\Gamma$.

Now suppose that social spillovers are non-zero, and consider the identification of $s_{AB}-s_A-s_B$. In this case, because each $s_j$ multiplies by past adoption shares, we can treat $p_{t-1}^g$ as a simple control variable and estimate $s_A$ (resp. $s_B$) through variation in past norm adoption $p_{A,t-1} + p_{AB,t-1}$ (resp. $p_{B,t-1}+ p_{AB,t-1}$) affecting $p_{A,t}, p_{AB,t}$ (resp. $p_{B,t}, p_{AB,t}$). Similarly, $s_{AB}$ is estimated through the relationship of $u_{AB,t}$ with $p_{AB,t-1}, p_{A,t-1}, p_{B,t-1}$. Using this intuition, we can therefore jointly identify $\Gamma$ (which does not vary over time) with the social components by propagating the system over time in the likelihood function. 

Note that an important assumption is that the unobserved components $\tilde{z}_{t}$ are not correlated with omitted institutional features that also directly affect the current utility. If that occurs, then $s_A$ or $s_B$ are not identified separately, and likewise $\Gamma$. Our maintained identifying assumption is therefore that, conditional on the included covariates, fixed effects, and excluded shifters, there are no omitted factors that enter current utilities and are correlated with lagged shares. 


\subsection{Estimation and Inference}

Estimation and inference proceeds as described in Algorithm \ref{alg:estimation}. Note that a key object entering the right-hand side of the choice utilities is the vector of lagged group shares of each norm/adoption state. Because these shares are themselves estimated, their sampling variation propagates into the sampling distribution of the structural estimator. This subsection formalizes estimation and standard errors taking into account survey sampling uncertainty.  We observe repeated cross-sections over finite periods $T$ with $N \equiv T n$.

\begin{algorithm}[!ht]
\caption{Estimation and inference}
\label{alg:estimation}
\begin{algorithmic}[1] 
\State \textbf{Input:} Micro-data $\{(v_{it}, X_i, g(i))\}_{t=1}^T$ (panel or repeated cross sections), with $v_{it}$ denoting the adopted norm, $X_i$ denoting the covariates (including the intercept) and $g(i)$ denoting the reference group generating social interactions for individual $i$; tuning parameter $\varepsilon>0$ for smoothing; grid size $G$ (set to $G=200$) and bounds $[-5,5]^2$; index sets $(\mathcal{I}_A,\mathcal{I}_B)$ for the validity restrictions in \eqref{eqn:validity}; set of starting values for $\theta$. Group-level shares $p_{A(t-1)}^{g(i)}, p_{B(t-1)}^{g(i)}, p_{AB(t-1)}^{g(i)}$ from the data. 
\State \textbf{Impose exclusion restrictions:} Restrict $\theta$ so that $\beta_A^{(\mathcal{I}_A)} = 0$ and $\beta_B^{(\mathcal{I}_B)} = 0$.
\State  \textbf{Construct quadrature grid:} Build a $G \times G$ tensor-product grid $\{(\tilde{z}_A^g,\tilde{z}_B^h)\}$ on $[-5,5]^2$ with associated weights $\{w_{gh}\}$ and bivariate normal density $\phi_\rho(\tilde{z}_A^g,\tilde{z}_B^h)$.

            \State \textbf{Compute model-implied choices:} Compute utilities as functions of $\theta$ and $z$:
            \[
            \begin{aligned}
            u_{Ait}(\theta; \tilde{z}) &= X_{it}^\top \beta_A + s_A (p_{A(t-1)}^{g(i)} + p_{AB(t-1)}^{g(i)}) + \tilde{z}_A^g, \\
            u_{Bit}(\theta; \tilde{z}) &= X_{it}^\top \beta_B + s_B (p_{B(t-1)}^{g(i)} + p_{AB(t-1)}^{g(i)}) + \tilde{z}_B^h, \\
            u_{ABit}(\theta; \tilde{z}) &= u_{Ait}(\theta) + u_{Bit}(\theta) + \Gamma + (s_{AB} - s_A - s_B)p_{AB(t-1)}^{g(i)}, \\
            u_{\emptyset it}(\theta; \tilde{z}) &= 0.
            \end{aligned}
            \]
            and the smoothed individual choice probabilities  
            $$
            \pi_{it}(v;\theta,\tilde{z}) =  \sigma_\varepsilon\left(u_{vit}(\theta;\tilde{z}) - \max_{v' \neq v} u_{v'it}(\theta;\tilde{z}) \right)
            $$
           where $ 
            \sigma_\varepsilon(x) = \big(1 + e^{-x/\varepsilon}\big)^{-1}.
            $ For $\varepsilon \rightarrow 0$, these correspond to the indicator functions. 
            
        \State \textbf{Compute individual probabilities:} Approximate the model choice probabilities by numerical integration:
        \[
        \hat{p}_{it}^{g(i)}(v;\theta) := \sum_{g,h} 
        \pi_{it}(v;\theta,\tilde{z})\,\phi_\rho(\tilde{z}_A^g,\tilde{z}_B^h)\,w_{gh},
        \quad v \in \{\emptyset,A,B,AB\}.
        \]
          
    \State \textbf{Compute log-likelihood:} Compute the log-likelihood
    \[
      \ell(\theta) = \sum_{i,t} \sum_{v \in \{\emptyset,A,B,AB\}}
      \mathbf{1}\{v_{it}=v\}\,\log(\hat{p}_{it}^{g(i)}(v;\theta)).
    \]

\State \textbf{Optimization:} With multiple random starting values for $\theta$, maximize $\ell(\theta)$ subject to the validity restrictions, obtaining $\hat{\theta}$.
\State \textbf{Confidence intervals:} Construct confidence intervals as in Supplemental Appendix \ref{sec:stacked_se} using the plug-in estimate of the Fisher information with respect to both $\theta$ and the estimated share of adopters. 
\State \textbf{Policy counterfactuals:} For each policy scenario, estimate policy counterfactuals as in Definition \ref{defn:policy1}. 
\end{algorithmic}
\end{algorithm}

\paragraph{Extended parameter vector.}
Let $v_{it}\in\{\emptyset,A,B,AB\}$ denote the realized choice. For each group $r$ and
time $t$, define the (population) share vector
\[
q_{r,t}\;\equiv\;
\begin{pmatrix}
p^{r}_{A,t}\\ p^{r}_{B,t}\\ p^{r}_{AB,t}
\end{pmatrix}\in\mathbb{R}^3,
\qquad
p^{r}_{\emptyset,t}=1-\mathbf{1}^\top q_{r,t}.
\]
Let $\theta\in\Theta\subset\mathbb{R}^{d_\theta}$ collect the structural parameters, and
define $ 
q\equiv\{q_{r,t}\}_{r=1,\dots,R;\;t=1,\dots,T}.
$  

Let $\hat{p}_{it}^{g(i)}(v;\theta, q)$ denote the model-implied (smoothed, integrated) probability
that individual $(i,t)$ chooses $v$ as in Step 6 of Algorithm \ref{alg:estimation}. Define the individual log-likelihood contribution and
its $\theta$-score
\[
\ell_{it}(\theta,q)\;\equiv\;\log \hat{p}_{it}^{g(i)}(v;\theta, q),
\qquad
s_{it}(\theta,q)\;\equiv\;\partial_\theta \ell_{it}(\theta,q)\in\mathbb{R}^{d_\theta}.
\]

\paragraph{Share-moment conditions.}
Index group--time cells by $c=(r,t)$ and let $\mathcal{I}_c$ be the set of observations
in cell $c$, with $n_c\equiv|\mathcal{I}_c|$. Let
$y_{it} = (1\{v_{it} = A\}, 1\{v_{it} = B\}, 1\{v_{it} = AB\})$.
Define the sample share in cell $c$ as $\bar y_c \equiv n_c^{-1}\sum_{(i,t)\in\mathcal{I}_c} y_{it}$, and 
$ 
G_{N,c}(q)\;\equiv\;\frac{n_c}{N}\,(\bar y_c-q_c)
\;=\;\frac{1}{N}\sum_{(i,t)\in\mathcal{I}_c}(y_{it}-q_c), \quad 
G_N(q)\equiv\{G_{N,c}(q)\}_c.
$  

\paragraph{Stacked estimating equations.}
Let $\alpha = (\theta, q)$ and  define 
\[
S_N(\theta,q)\;\equiv\;\frac{1}{N}\sum_{i,t} s_{it}(\theta,q),
\qquad
\Psi_N(\alpha)\;\equiv\;
\begin{pmatrix}
S_N(\theta,q)\\
G_N(q)
\end{pmatrix}.
\]
The stacked estimator solves $\Psi_N(\hat{\theta}, \hat{q})=0$, which coincides with the estimator in Algorithm \ref{alg:estimation} (since $G_{N,c}(q)=0$ implies $q_c=\bar y_c$ for each cell $c$).

\paragraph{Asymptotic linear representation.}
Assume standard regularity conditions for smooth M-estimation and that
$n_c/N\to\kappa_c\in(0,1)$ for each cell $c$ (so $\hat q_c-q_{0,c}=O_p(N^{-1/2})$).
Let $\alpha_0=(\theta_0,q_0)$ denote the true parameter values. A first-order expansion of the stacked
equations gives
\[
0
=
\Psi_N(\alpha_0)
+
A\,(\hat\alpha-\alpha_0)
+
o_p(N^{-1/2}),
\qquad
A\equiv \partial_\alpha \Psi(\alpha_0)
=
\begin{pmatrix}
A_{\theta\theta} & A_{\theta q}\\
0 & A_{qq}
\end{pmatrix},
\]
where
\[
A_{\theta\theta}\equiv \partial_\theta S(\theta_0,q_0),\qquad
A_{\theta q}\equiv \partial_q S(\theta_0,q_0),\qquad
A_{qq}\equiv \partial_q G(q_0).
\]
Because $G_{N,c}(q)$ is linear in $q_c$, $A_{qq}$ is block-diagonal with blocks
$-(\frac{n_c}{N} I_3)$, hence $A_{qq}^{-1}$ exists and is also block-diagonal. Solving the block system yields the influence-function form
\begin{align*}
\sqrt{N}(\hat\theta-\theta_0)
&=
-A_{\theta\theta}^{-1}\frac{1}{\sqrt{N}}\sum_{i,t}
\underbrace{\Big(s_{it}(\theta_0,q_0)-A_{\theta q}A_{qq}^{-1}m_{it}(q_0)\Big)}_{\displaystyle \tilde s_{it}}
+o_p(1),
\end{align*}
where $m_{it}(q_0)$ is the per-observation contribution to the share moments: it is a
three-dimensional vector that places $(y_{it}-q_{c,0})$ into the coordinates of the
cell $c=(g(i),t)$ and zeros elsewhere.

\paragraph{Sandwich variance and clustering.}
Let $\hat A_{\theta\theta}$ and $\widehat{A_{\theta q}A_{qq}^{-1}}$ denote sample analogues
(e.g., obtained by automatic differentiation). Define the adjusted score
\[
\hat{\tilde s}_{it}\;\equiv\; s_{it}(\hat\theta,\hat q)-\widehat{A_{\theta q}A_{qq}^{-1}}\,m_{it}(\hat q).
\]
Then a heteroskedasticity-robust sandwich estimator is
\[
\widehat{\mathbb{V}}(\hat\theta)
=
\hat A_{\theta\theta}^{-1}
\left(\frac{1}{N}\sum_{i,t}\hat{\tilde s}_{it}\hat{\tilde s}_{it}^\top\right)
\hat A_{\theta\theta}^{-1\top}.
\]


\section{Additional tables and figures}

\subsection{Covariate composition}

\subsubsection{Sierra Leone}



In Sierra Leone, the sample contains 33{,}658 women, with cohort counts ranging from 3{,}528 in the oldest bin (1958--1970) to 11{,}991 in the 1981--1990 bin. Respondents span five regions with heterogeneous sizes and are grouped into three ethnicity groups with similar shares (33.6\% Mende, 31.4\% Temne, and 34.9\% grouped as ``Other''). Religion is stable across cohorts and predominantly Muslim (above $75\%$). Urban residence rises modestly over cohorts, from about 35\% in the oldest bins to 46.1\% among cohorts born between 1991--2001. Education and literacy exhibit pronounced cohort changes. Literacy increases in parallel, from 14.6\% among the oldest cohorts to 41.9\% among the youngest. These changes can be explained by several factors, including the post-civil war recovery period (the civil war ended in 2002 and therefore plausibly affecting post-1985 cohorts' education level).
Because education can be strongly correlated with behaviors related to both norms' adoption, the cohort trends in Table~\ref{tab:sl_summary1} underscore the importance of controlling for these characteristics in our model.

\begin{table}[H]
\centering
\small 
\caption{Sierra Leone: Summary statistics by birth-cohort ten years bins.} \label{tab:sl_summary1}

\scriptsize
\setlength{\tabcolsep}{4pt}
\renewcommand{\arraystretch}{0.75}
\setlength{\aboverulesep}{0pt}
\setlength{\belowrulesep}{1pt}
\setlength{\cmidrulesep}{0pt}

\centering
\begin{tabular}[t]{lllllll}
\toprule
\small
Panel & Statistic & 1958--1970 & 1971--1980 & 1981--1990 & 1991--2001 & Full sample\\
\midrule
 & Observations & 3,528 & 9,015 & 11,991 & 9,124 & 33,658\\
\cmidrule{1-7}
 & High education & 14.3\% & 13.9\% & 23.1\% & 55.6\% & 28.5\%\\
 & Literate & 14.6\% & 12.8\% & 20.5\% & 41.9\% & 23.6\%\\
 & Low education & 8.4\% & 10.3\% & 12.3\% & 13.5\% & 11.7\%\\
 & Muslim & 77.6\% & 77.1\% & 77.2\% & 75.6\% & 76.8\%\\
 & Christian & 21.8\% & 22.5\% & 22.3\% & 24.1\% & 22.8\%\\
{\raggedright\arraybackslash Covariates} & Urban & 35.7\% & 35.2\% & 39.7\% & 46.1\% & 39.8\%\\
\cmidrule{1-7}
 & East & 21.7\% & 22.0\% & 20.9\% & 18.2\% & 20.6\%\\
 & North & 34.3\% & 32.6\% & 31.0\% & 29.6\% & 31.4\%\\
 & South & 26.5\% & 24.9\% & 24.5\% & 21.8\% & 24.1\%\\
{\raggedright\arraybackslash Region} & West & 17.5\% & 20.4\% & 23.6\% & 30.3\% & 23.9\%\\
\cmidrule{1-7}
 & Mende & 36.6\% & 34.7\% & 34.3\% & 30.5\% & 33.6\%\\
 & Other & 33.0\% & 34.9\% & 34.8\% & 35.9\% & 34.9\%\\
{\raggedright\arraybackslash Ethnicity} & Temne & 30.4\% & 30.4\% & 30.9\% & 33.6\% & 31.4\%\\
\bottomrule
\end{tabular}
\end{table}

\subsubsection{Nigeria}


Nigeria’s covariate composition changes across birth cohorts (Table~\ref{tab:ng_summary1}). In particular,  educational attainment rises over time: the share with high education increases from 34.9\% among the 1958--1970 cohort to 60.3\% among the 1991--2001 cohort, while low education falls from 27.0\% to 9.8\%. Literacy also increases over time, from 38.2\% to 51.3\% for the 1981--1990 cohort, and remains high thereafter (48.2\% in 1991--2001).

\begin{table}[H]
\centering
\caption{Nigeria: summary statistics by birth-cohort ten years bins.} \label{tab:ng_summary1}
\centering

\scriptsize
\setlength{\tabcolsep}{4pt}
\renewcommand{\arraystretch}{0.75}
\setlength{\aboverulesep}{0pt}
\setlength{\belowrulesep}{1pt}
\setlength{\cmidrulesep}{0pt}

\begin{tabular}[t]{lllllll}
\toprule
Panel & Statistic & 1958--1970 & 1971--1980 & 1981--1990 & 1991--2001 & Full sample\\
\midrule
 & Observations & 8,814 & 18,785 & 26,321 & 17,285 & 71,205\\
\cmidrule{1-7}
 & High education & 34.9\% & 45.5\% & 56.7\% & 60.5\% & 52.0\%\\
 & Literate & 38.2\% & 43.6\% & 51.3\% & 48.3\% & 46.9\%\\
 & Low education & 27.0\% & 23.4\% & 14.7\% & 9.7\% & 17.3\%\\
 & Muslim & 36.5\% & 39.8\% & 43.0\% & 51.9\% & 43.5\%\\
 & Christian & 13.2\% & 11.8\% & 11.8\% & 10.9\% & 11.8\%\\
{\raggedright\arraybackslash Covariates} & Urban & 42.1\% & 46.9\% & 47.5\% & 47.3\% & 46.6\%\\
\cmidrule{1-7}
 & North Central & 10.9\% & 11.5\% & 12.0\% & 12.3\% & 11.8\%\\
 & North East & 12.1\% & 12.6\% & 13.8\% & 15.7\% & 13.7\%\\
 & South East & 19.0\% & 16.9\% & 16.1\% & 15.0\% & 16.4\%\\
 & South South & 18.2\% & 18.7\% & 19.1\% & 15.7\% & 18.0\%\\
 & South West & 21.9\% & 20.8\% & 17.8\% & 13.3\% & 18.0\%\\
{\raggedright\arraybackslash Region} & Other region & 17.9\% & 19.6\% & 21.2\% & 28.0\% & 22.0\%\\
\cmidrule{1-7}
 & Hausa & 20.8\% & 23.6\% & 26.6\% & 35.4\% & 27.2\%\\
 & Igbo & 22.3\% & 20.5\% & 19.9\% & 17.6\% & 19.8\%\\
 & Yoruba & 21.8\% & 20.3\% & 17.2\% & 12.8\% & 17.5\%\\
{\raggedright\arraybackslash Ethnicity} & Other & 35.1\% & 35.5\% & 36.3\% & 34.2\% & 35.4\%\\
\bottomrule
\end{tabular}
\end{table}
 
\begin{table}[H]
\centering
\caption{Child's Rights Act assent date by State in Nigeria}
\label{cra:list}

\begin{threeparttable}
\footnotesize
\setlength{\tabcolsep}{8pt}
\renewcommand{\arraystretch}{0.9}

\begin{tabular}{@{}lclclclc@{}}
\toprule
State & Year & State & Year & State & Year & State & Year \\
\midrule
FCT         & 2003 & Anambra   & 2004 & Delta    & 2008 & Imo      & 2004 \\
Abia        & 2006 & Bauchi    & 2023 & Ebonyi   & 2010 & Jigawa   & 2021 \\
Adamawa     & 2022 & Bayelsa   & 2016 & Edo      & 2007 & Kaduna   & 2018 \\
Akwa-Ibom   & 2008 & Benue     & 2008 & Ekiti    & 2006 & Kano     & 2023 \\
Katsina     & 2021 & Kwara     & 2005 & Ondo     & 2007 & Rivers   & 2009 \\
Kebbi       & 2022 & Lagos     & 2007 & Osun     & 2007 & Sokoto   & 2021 \\
Kogi        & 2007 & Nasarawa  & 2005 & Oyo      & 2006 & Taraba   & 2005 \\
Cross River & 2008 & Borno     & 2022 & Enugu    & 2016 & Gombe    & Not assented \\
Niger       & 2010 & Ogun      & 2006 & Plateau  & 2005 & Yobe     & 2022 \\
Zamfara     & 2022 &           &      &          &      &          &      \\
\bottomrule
\end{tabular}

\vspace{0.15cm}

\begin{tablenotes}
\scriptsize
\item \textit{Source:} The Rule of Law and Empowerment Initiative (Partners West Africa Nigeria). Child's Rights Act Tracker
\href{https://www.partnersnigeria.org/child-rights-act-tracker/}{https://www.partnersnigeria.org/child-rights-act-tracker/}.
\end{tablenotes}

\end{threeparttable}
\end{table}

\subsection{Descriptive analysis} \label{S_Appendix_descriptives}

 Table \ref{tab:desc_stats} reports aggregate adoption rates for FGC, CMA, PGY, and their joint prevalence across major socio-demographic sub-groups in Nigeria and Sierra Leone. Observe that Sierra Leone exhibits a uniformly high prevalence of FGC in all groups, often above 90 percent; this is consistent with the fact that FGC being almost universal in most of the country. By contrast, Nigeria displays much stronger heterogeneity between groups: CMA and PGY are substantially more prevalent among rural, Muslim and low-education populations, while FGC is relatively more prevalent in the South and among Yoruba and Igbo groups.  This generates an important asymmetry in Nigeria, where the groups with the highest prevalence of CMA are not the same as those with the highest prevalence of FGC, suggesting a more limited overlap between the two norms than in Sierra Leone. 

In both countries, the urban and more educated groups exhibit a markedly lower prevalence of each of the norms we consider. Moreover, PGY$\times$CMA adoption is concentrated in the same groups where child marriage is the most common, particularly in Northern Nigeria, while in Sierra Leone it remains positive but substantially smaller than adoption of FGC$\times$CMA throughout.

\begin{table}[ht!]
  \centering
  \caption{Descriptive statistics by country and sub-group}
  \label{tab:desc_stats}
  \scalebox{0.9}{
   \begin{tabular}{llcccccc}
    \toprule
    &  &   $N$ & FGC & CMA & PGY & FGC$\times$CMA & PGY$\times$CMA \\
    \midrule
    \multicolumn{8}{l}{\textbf{Panel A: Nigeria}} \\
    \midrule
    Residence & Rural & 79,901 & 32.7\% & 52.9\% & 17.1\% & 15.2\% & 11.3\% \\
     & Urban & 54,182 & 37.9\% & 25.7\% & 8.7\% & 8.9\% & 4.2\% \\
    \addlinespace[4pt]
    Religion & Muslim & 65,068 & 29.7\% & 62.3\% & 21.2\% & 16.5\% & 14.7\% \\
     & Non-Muslim & 69,015 & 39.4\% & 22.6\% & 6.6\% & 8.9\% & 2.6\% \\
    \addlinespace[4pt]
    Education & High education & 63,042 & 35.9\% & 15.8\% & 5.3\% & 5.6\% & 1.7\% \\
     & Low education & 71,041 & 34.2\% & 65.1\% & 21.1\% & 19.6\% & 14.4\% \\
    \addlinespace[4pt]
    Region & North East & 23,422 & 6.9\% & 61.2\% & 19.4\% & 4.8\% & 14.0\% \\
     & North West & 31,882 & 29.6\% & 71.9\% & 21.2\% & 23.0\% & 16.5\% \\
     & North Central & 24,569 & 26.1\% & 35.8\% & 13.5\% & 8.5\% & 6.5\% \\
     & South West & 19,018 & 56.4\% & 16.7\% & 10.3\% & 10.6\% & 2.9\% \\
     & South South & 18,733 & 32.6\% & 23.0\% & 5.1\% & 9.4\% & 2.1\% \\
     & South East & 16,459 & 52.6\% & 16.2\% & 4.5\% & 11.3\% & 1.4\% \\
    \addlinespace[4pt]
    Ethnicity & Hausa-Fulani & 41,468 & 26.5\% & 72.5\% & 22.3\% & 20.2\% & 17.1\% \\
     & Igbo & 20,001 & 50.1\% & 15.4\% & 4.0\% & 10.4\% & 1.2\% \\
     & Yoruba & 18,124 & 61.4\% & 15.6\% & 10.5\% & 11.5\% & 2.8\% \\
     & Others & 54,349 & 20.7\% & 37.0\% & 11.6\% & 7.4\% & 6.4\% \\
    \addlinespace[4pt]
    \multicolumn{8}{l}{\textbf{Panel B: Sierra Leone}} \\
    \midrule
    Residence & Rural & 20,641 & 96.6\% & 49.1\% & 17.9\% & 48.1\% & 9.5\% \\
     & Urban & 13,796 & 87.9\% & 33.0\% & 8.4\% & 32.1\% & 4.3\% \\
    \addlinespace[4pt]
    Religion & Muslim & 26,382 & 96.1\% & 45.2\% & 16.0\% & 44.5\% & 8.5\% \\
     & Non-Muslim & 8,055 & 83.4\% & 34.1\% & 8.0\% & 32.5\% & 3.8\% \\
    \addlinespace[4pt]
    Education & High education & 9,936 & 82.5\% & 21.1\% & 5.1\% & 20.0\% & 2.2\% \\
     & Low education & 24,501 & 97.3\% & 51.4\% & 17.8\% & 50.4\% & 9.6\% \\
    \addlinespace[4pt]
    Region & Eastern & 7,046 & 94.7\% & 47.1\% & 13.5\% & 45.7\% & 7.3\% \\
     & North Western & 2,166 & 97.2\% & 39.2\% & 17.4\% & 38.9\% & 6.4\% \\
     & Northern & 10,701 & 97.7\% & 48.7\% & 19.5\% & 48.3\% & 10.9\% \\
     & Southern & 8,482 & 92.6\% & 42.1\% & 13.1\% & 41.2\% & 6.9\% \\
     & Western & 6,042 & 82.2\% & 28.6\% & 5.4\% & 27.0\% & 2.6\% \\
    \addlinespace[4pt]
    Ethnicity & Mende & 11,755 & 92.0\% & 42.3\% & 12.6\% & 41.2\% & 6.7\% \\
     & Other & 11,942 & 91.5\% & 42.0\% & 14.2\% & 40.9\% & 7.7\% \\
     & Temne & 10,740 & 96.1\% & 43.6\% & 15.6\% & 43.2\% & 8.0\% \\
    \bottomrule
  \end{tabular}
  }

    \vspace{0.4cm}

 \begin{tablenotes}
    \footnotesize \item \textit{Notes:} FGC, CMA, and PGY indicate adoption of each norm (with or without the other). FGC$\times$CMA and PGY$\times$CMA denote joint adoption.
\end{tablenotes} 

  \label{tab:summarystats:all}
\end{table}

\subsection{Additional results and tables}

\begin{figure}[H]
    \centering
    \caption{Age distribution of CMA and FGC in Sierra Leone}
    \includegraphics[width=0.75\linewidth]{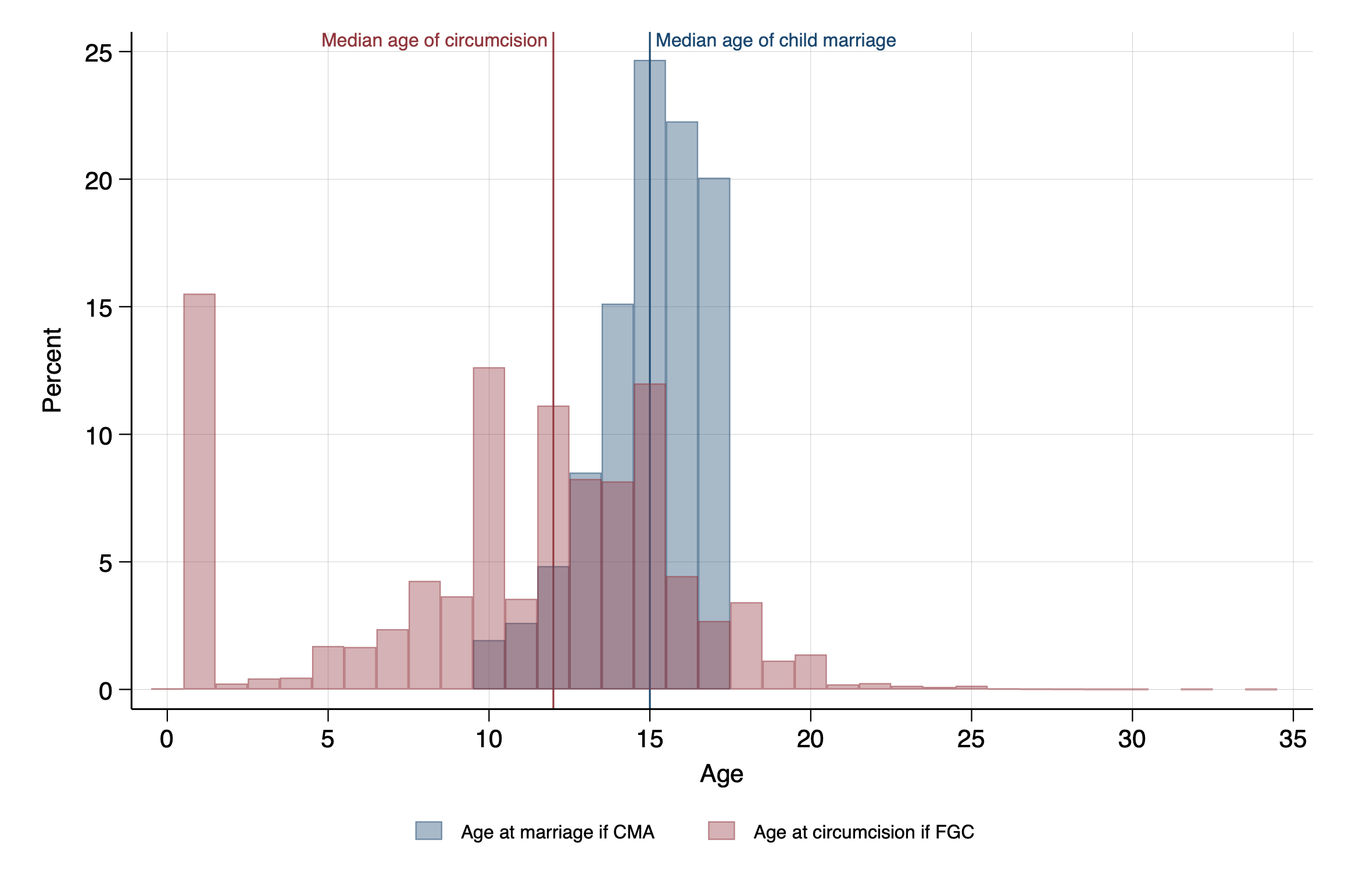}
    \begin{flushleft}
        \footnotesize \textit{Notes:} The figure reports for Sierra Leone the distribution of age at first marriage for brides who married as a child (blue) and age at circumcision for respondents who answer being circumcised (red).
    \end{flushleft}\label{SL:agedist}
\end{figure}

\begin{figure}[H]
    \centering
     \caption{Age distribution of CMA and FGC in Nigeria}
    \includegraphics[width=0.75\linewidth]{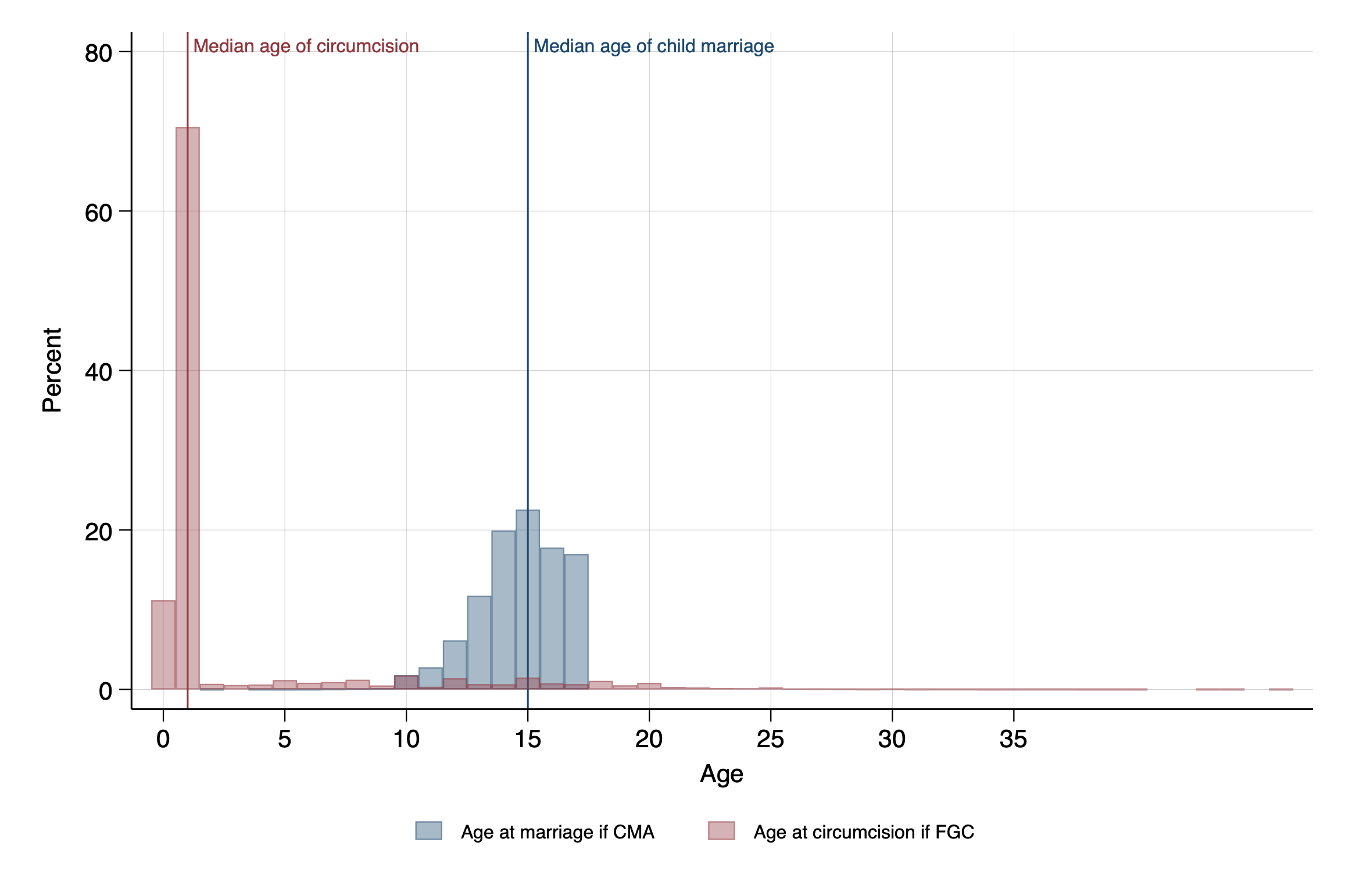}
    \begin{flushleft}
        \footnotesize \textit{Notes:} The figure reports for Sierra Leone the distribution of age at first marriage for brides who married as a child (blue) and age at circumcision for respondents who answer being circumcised (red).
    \end{flushleft}\label{NG:age}
\end{figure}

\begin{figure}[H]
    \centering
     \caption{Prevalence of drought in Sierra Leone by year}
    \includegraphics[width=0.6\linewidth]{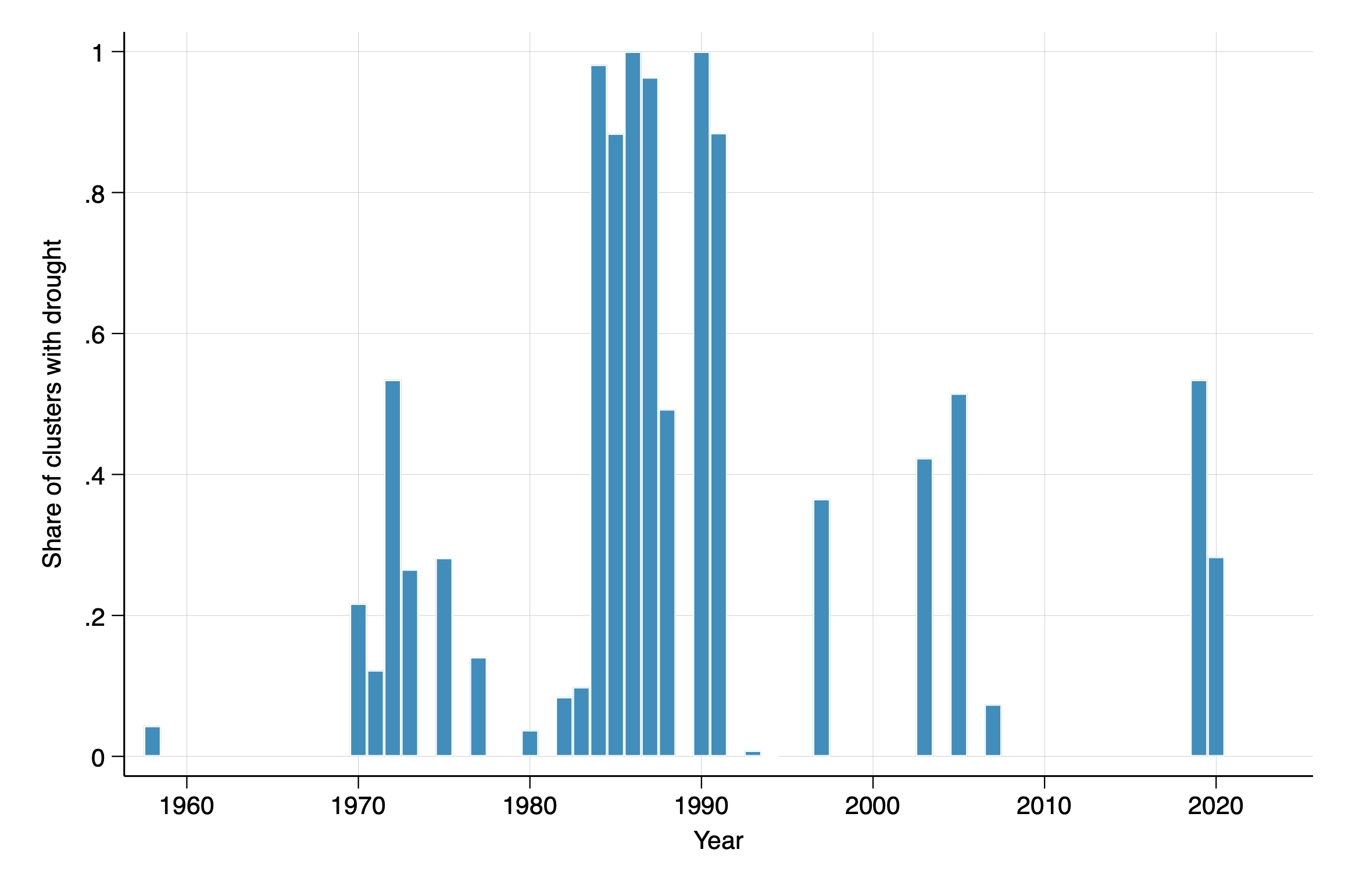}
    \begin{flushleft}
        \footnotesize \textit{Notes:} The figure reports the prevalence of drought in Sub-Saharan Africa and India, presented as the percentage of grid cells with drought in each calendar year in Sierra Leone. For all the analyses in this paper, for any grid cell, we define a drought as having rainfall lower than the 15th percentile of the long-run rainfall distribution.
    \end{flushleft}\label{SL:drought}
\end{figure}

\vspace{-10pt}
 
 \section{Robustness exercises}

\subsection{FGC Robustness} \label{sec:fgc_robustness}

\begin{table}[H]
\centering
\caption{Sierra Leone: dropping women cut before age 5}
\renewcommand{\arraystretch}{0.7}
\small
\begin{tabular}{lcccccccc}
\hline
& \multicolumn{2}{c}{Baseline} & \multicolumn{2}{c}{Robustness} & \multicolumn{2}{c}{Baseline} & \multicolumn{2}{c}{Robustness} \\
Variable & FGC & CMA & FGC & CMA & FGC & CMA & FGC & CMA \\
& (1) & (2) & (3) & (4) & (5) & (6) & (7) & (8) \\
\hline
CMA Ban &  & $-0.07^{***}$ &  & $-0.09^{***}$ &  & $-0.07^{***}$ &  & $-0.09^{***}$ \\
 &  & $(0.02)$ &  & $(0.02)$ &  & $(0.02)$ &  & $(0.02)$ \\
$\Gamma$ & \multicolumn{2}{c}{$1.05^{***}$} & \multicolumn{2}{c}{$1.12^{***}$} & \multicolumn{2}{c}{$1.08^{***}$} & \multicolumn{2}{c}{$1.12^{***}$} \\
 & \multicolumn{2}{c}{$(0.07)$} & \multicolumn{2}{c}{$(0.06)$} & \multicolumn{2}{c}{$(0.08)$} & \multicolumn{2}{c}{$(0.08)$} \\
$s$ & $2.46^{***}$ & $0.36^{***}$ & $1.97^{***}$ & $0.30^{**}$ & $2.47^{***}$ & $0.44^{*}$ & $2.01^{***}$ & $0.38$ \\
 & $(0.24)$ & $(0.13)$ & $(0.24)$ & $(0.13)$ & $(0.25)$ & $(0.24)$ & $(0.25)$ & $(0.24)$ \\
$s_{AB}$ & \multicolumn{2}{c}{} & \multicolumn{2}{c}{} & \multicolumn{2}{c}{$2.83^{***}$} & \multicolumn{2}{c}{$2.32^{***}$} \\
 & \multicolumn{2}{c}{} & \multicolumn{2}{c}{} & \multicolumn{2}{c}{$(0.28)$} & \multicolumn{2}{c}{$(0.29)$} \\
$\rho$ & \multicolumn{2}{c}{$-0.70^{***}$} & \multicolumn{2}{c}{$-0.75^{***}$} & \multicolumn{2}{c}{$-0.70^{***}$} & \multicolumn{2}{c}{$-0.73^{***}$} \\
 & \multicolumn{2}{c}{$(0.04)$} & \multicolumn{2}{c}{$(0.03)$} & \multicolumn{2}{c}{$(0.04)$} & \multicolumn{2}{c}{$(0.04)$} \\ \hline
 N  & \multicolumn{2}{c}{33161}  &  \multicolumn{2}{c}{ 26474} & \multicolumn{2}{c}{33161} & \multicolumn{2}{c}{ 26474 }  \\ 
\hline
\end{tabular}

\vspace{0.4cm}

\begin{minipage}{1\textwidth}
\footnotesize
\emph{Notes:} Columns 1 and 3 report the baseline Sierra Leone FGC--CMA estimates. Columns 2 and 4 report the corresponding robustness exercise that excludes women reported as cut before age 5. Columns 1--2 impose the specification without a separate joint-adoption spillover term $s_{AB}$, while Columns 3--4 allow for a distinct joint spillover parameter.  All specifications additionally control for indicators for primary education, secondary education, Muslim, urban residence, ethnicity (Mende, Temne, with Other as the reference level), region (South, East, West, with North as the reference level), and wealth quintile (2--5).  $^{***}p<0.01$, $^{**}p<0.05$, $^{*}p<0.1$.
\end{minipage}
 \label{table:FGC_SL_first}
\end{table}


\begin{table}[H]
\centering
\caption{Sierra Leone: keeping only women cut before age 5}
\renewcommand{\arraystretch}{0.7}
\small
\begin{tabular}{lcccccccc}
\hline
& \multicolumn{2}{c}{Baseline} & \multicolumn{2}{c}{Robustness} & \multicolumn{2}{c}{Baseline} & \multicolumn{2}{c}{Robustness} \\
Variable & FGC & CMA & FGC & CMA & FGC & CMA & FGC & CMA \\
& (1) & (2) & (3) & (4) & (5) & (6) & (7) & (8) \\
\hline
CMA Ban &  & $-0.07^{***}$ &  & $-0.09^{*}$ &  & $-0.07^{***}$ &  & $-0.10^{**}$ \\
 &  & $(0.02)$ &  & $(0.05)$ &  & $(0.02)$ &  & $(0.05)$ \\
$\Gamma$ & \multicolumn{2}{c}{$1.05^{***}$} & \multicolumn{2}{c}{$0.39$} & \multicolumn{2}{c}{$1.08^{***}$} & \multicolumn{2}{c}{$0.48^{**}$} \\
 & \multicolumn{2}{c}{$(0.07)$} & \multicolumn{2}{c}{$(0.32)$} & \multicolumn{2}{c}{$(0.08)$} & \multicolumn{2}{c}{$(0.21)$} \\
$s$ & $2.46^{***}$ & $0.36^{***}$ & $1.57^{***}$ & $-0.00$ & $2.47^{***}$ & $0.44^{*}$ & $1.61^{***}$ & $0.18$ \\
 & $(0.24)$ & $(0.13)$ & $(0.17)$ & $(0.18)$ & $(0.25)$ & $(0.24)$ & $(0.15)$ & $(0.20)$ \\
$s_{AB}$ & \multicolumn{2}{c}{} & \multicolumn{2}{c}{} & \multicolumn{2}{c}{$2.83^{***}$} & \multicolumn{2}{c}{$1.53^{***}$} \\
 & \multicolumn{2}{c}{} & \multicolumn{2}{c}{} & \multicolumn{2}{c}{$(0.28)$} & \multicolumn{2}{c}{$(0.24)$} \\
$\rho$ & \multicolumn{2}{c}{$-0.70^{***}$} & \multicolumn{2}{c}{$-0.17$} & \multicolumn{2}{c}{$-0.70^{***}$} & \multicolumn{2}{c}{$-0.18$} \\
 & \multicolumn{2}{c}{$(0.04)$} & \multicolumn{2}{c}{$(0.31)$} & \multicolumn{2}{c}{$(0.04)$} & \multicolumn{2}{c}{$(0.20)$} \\
 \hline
 N  & \multicolumn{2}{c}{33161}  &  \multicolumn{2}{c}{7240} & \multicolumn{2}{c}{33161} & \multicolumn{2}{c}{ 7240}  \\ 
\hline
\end{tabular}

\vspace{0.4cm}

\begin{minipage}{1\textwidth}
\footnotesize
\emph{Notes:} Columns 1 and 3 report the baseline Sierra Leone FGC--CMA estimates. Columns 2 and 4 report the corresponding robustness exercise that performs the estimation only on women cut before age 5. Columns 1--2 impose the specification without a separate joint-adoption spillover term $s_{AB}$, while Columns 3--4 allow for a distinct joint spillover parameter. All specifications additionally control for the same indicators as in Table \ref{table:FGC_SL_first}. Standard errors are reported in parentheses. $^{***}p<0.01$, $^{**}p<0.05$, $^{*}p<0.1$.

\end{minipage}
\end{table}

\begin{table}[H]
\centering
\small
\caption{Sierra Leone: coding ``don't know FGC" as 0 instead of NA} 
\renewcommand{\arraystretch}{0.7}
\begin{tabular}{lcccccccc}
\hline
& \multicolumn{2}{c}{Baseline} & \multicolumn{2}{c}{Robustness} & \multicolumn{2}{c}{Baseline} & \multicolumn{2}{c}{Robustness} \\
Variable & FGC & CMA & FGC & CMA & FGC & CMA & FGC & CMA \\
& (1) & (2) & (3) & (4) & (5) & (6) & (7) & (8) \\

\hline
CMA Ban &  & $-0.07^{***}$ &  & $-0.07^{***}$ &  & $-0.07^{***}$ &  & $-0.07^{***}$ \\
 &  & $(0.02)$ &  & $(0.02)$ &  & $(0.02)$ &  & $(0.02)$ \\
$\Gamma$ & \multicolumn{2}{c}{$1.05^{***}$} & \multicolumn{2}{c}{$1.04^{***}$} & \multicolumn{2}{c}{$1.08^{***}$} & \multicolumn{2}{c}{$1.05^{***}$} \\
 & \multicolumn{2}{c}{$(0.07)$} & \multicolumn{2}{c}{$(0.06)$} & \multicolumn{2}{c}{$(0.08)$} & \multicolumn{2}{c}{$(0.08)$} \\
$s$ & $2.46^{***}$ & $0.36^{***}$ & $2.37^{***}$ & $0.36^{***}$ & $2.47^{***}$ & $0.44^{*}$ & $2.39^{***}$ & $0.40^{*}$ \\
 & $(0.24)$ & $(0.13)$ & $(0.22)$ & $(0.12)$ & $(0.25)$ & $(0.24)$ & $(0.22)$ & $(0.22)$ \\
$s_{AB}$ & \multicolumn{2}{c}{} & \multicolumn{2}{c}{} & \multicolumn{2}{c}{$2.83^{***}$} & \multicolumn{2}{c}{$2.75^{***}$} \\
 & \multicolumn{2}{c}{} & \multicolumn{2}{c}{} & \multicolumn{2}{c}{$(0.28)$} & \multicolumn{2}{c}{$(0.26)$} \\
$\rho$ & \multicolumn{2}{c}{$-0.70^{***}$} & \multicolumn{2}{c}{$-0.70^{***}$} & \multicolumn{2}{c}{$-0.70^{***}$} & \multicolumn{2}{c}{$-0.69^{***}$} \\
 & \multicolumn{2}{c}{$(0.04)$} & \multicolumn{2}{c}{$(0.04)$} & \multicolumn{2}{c}{$(0.04)$} & \multicolumn{2}{c}{$(0.04)$} \\
  \hline
  N  & \multicolumn{2}{c}{33161}  &  \multicolumn{2}{c}{33416} & \multicolumn{2}{c}{33161} & \multicolumn{2}{c}{33416}  \\ 
\hline
\end{tabular}

\vspace{0.4cm}

\begin{minipage}{1\textwidth}
\footnotesize
\emph{Notes:} Columns 1 and 3 report the baseline Sierra Leone FGC--CMA estimates. Columns 2 and 4 report the corresponding robustness exercise that codes women who responded that they do not know what FGC is as not adhering to FGC, rather than NA. Columns 1--2 impose the specification without a separate joint-adoption spillover term $s_{AB}$, while Columns 3--4 allow for a distinct joint spillover parameter. All specifications additionally control for the same indicators as in Table \ref{table:FGC_SL_first}. Standard errors are reported in parentheses. $^{***}p<0.01$, $^{**}p<0.05$, $^{*}p<0.1$.

\end{minipage}
\end{table}

\begin{table}[H]
\centering
\renewcommand{\arraystretch}{0.7}
\small
\caption{Sierra Leone: Drought as additional exclusion restriction}
\begin{tabular}{lcccccccc}
\hline
& \multicolumn{2}{c}{Baseline} & \multicolumn{2}{c}{Robustness} & \multicolumn{2}{c}{Baseline} & \multicolumn{2}{c}{Robustness} \\
Variable & FGC & CMA & FGC & CMA & FGC & CMA & FGC & CMA \\
& (1) & (2) & (3) & (4) & (5) & (6) & (7) & (8) \\

\hline
CMA Ban &  & $-0.07^{***}$ &  & $-0.07^{***}$ &  & $-0.07^{***}$ &  & $-0.07^{***}$ \\
 &  & $(0.02)$ &  & $(0.02)$ &  & $(0.02)$ &  & $(0.02)$ \\
Drought &  &  &  & $0.01$ &  &  &  & $0.01$ \\
 &  &  &  & $(0.01)$ &  &  &  & $(0.01)$ \\
$\Gamma$ & \multicolumn{2}{c}{$1.05^{***}$} & \multicolumn{2}{c}{$1.05^{***}$} & \multicolumn{2}{c}{$1.08^{***}$} & \multicolumn{2}{c}{$1.08^{***}$} \\
 & \multicolumn{2}{c}{$(0.07)$} & \multicolumn{2}{c}{$(0.07)$} & \multicolumn{2}{c}{$(0.08)$} & \multicolumn{2}{c}{$(0.08)$} \\
$s$ & $2.46^{***}$ & $0.36^{***}$ & $2.45^{***}$ & $0.36^{***}$ & $2.47^{***}$ & $0.44^{*}$ & $2.47^{***}$ & $0.45^{*}$ \\
 & $(0.24)$ & $(0.13)$ & $(0.24)$ & $(0.13)$ & $(0.25)$ & $(0.24)$ & $(0.25)$ & $(0.24)$ \\
$s_{AB}$ & \multicolumn{2}{c}{} & \multicolumn{2}{c}{} & \multicolumn{2}{c}{$2.83^{***}$} & \multicolumn{2}{c}{$2.83^{***}$} \\
 & \multicolumn{2}{c}{} & \multicolumn{2}{c}{} & \multicolumn{2}{c}{$(0.28)$} & \multicolumn{2}{c}{$(0.29)$} \\
$\rho$ & \multicolumn{2}{c}{$-0.70^{***}$} & \multicolumn{2}{c}{$-0.70^{***}$} & \multicolumn{2}{c}{$-0.70^{***}$} & \multicolumn{2}{c}{$-0.70^{***}$} \\
 & \multicolumn{2}{c}{$(0.04)$} & \multicolumn{2}{c}{$(0.04)$} & \multicolumn{2}{c}{$(0.04)$} & \multicolumn{2}{c}{$(0.04)$} \\

\hline
 N  & \multicolumn{2}{c}{33161}  &  \multicolumn{2}{c}{33161} & \multicolumn{2}{c}{33161} & \multicolumn{2}{c}{33161}  \\ 
\hline
\end{tabular}

\vspace{0.4cm}

\begin{minipage}{1\textwidth}
\footnotesize
\emph{Notes:} Columns 1 and 3 report the baseline Sierra Leone FGC--CMA estimates. Columns 2 and 4 report the corresponding robustness exercise that includes drought as an exclusion restriction. Columns 1--2 impose the specification without a separate joint-adoption spillover term $s_{AB}$, while Columns 3--4 allow for a distinct joint spillover parameter. All specifications additionally control for the same indicators as in Table \ref{table:FGC_SL_first}. Standard errors are reported in parentheses. $^{***}p<0.01$, $^{**}p<0.05$, $^{*}p<0.1$.

\end{minipage}

\end{table}

\vspace{-10pt}

 \begin{table}[H]
\centering
\renewcommand{\arraystretch}{0.7}
\small
\caption{Nigeria: dropping women cut before age 5}
\begin{tabular}{lcccccccc}
\hline
& \multicolumn{2}{c}{Baseline} & \multicolumn{2}{c}{Robustness} & \multicolumn{2}{c}{Baseline} & \multicolumn{2}{c}{Robustness} \\
Variable & FGC & CMA & FGC & CMA & FGC & CMA & FGC & CMA \\
& (1) & (2) & (3) & (4) & (5) & (6) & (7) & (8) \\

\hline
CMA Ban &  & $-0.09^{***}$ &  & $-0.12^{***}$ &  & $-0.09^{***}$ &  & $-0.12^{***}$ \\
 &  & $(0.01)$ &  & $(0.02)$ &  & $(0.01)$ &  & $(0.02)$ \\
$\Gamma$ & \multicolumn{2}{c}{$0.14^{***}$} & \multicolumn{2}{c}{$0.00$} & \multicolumn{2}{c}{$0.16^{***}$} & \multicolumn{2}{c}{$-0.01$} \\
 & \multicolumn{2}{c}{$(0.05)$} & \multicolumn{2}{c}{$(0.04)$} & \multicolumn{2}{c}{$(0.06)$} & \multicolumn{2}{c}{$(0.04)$} \\
$s$ & $2.04^{***}$ & $0.94^{***}$ & $2.38^{***}$ & $0.81^{***}$ & $2.05^{***}$ & $0.95^{***}$ & $2.37^{***}$ & $0.81^{***}$ \\
 & $(0.05)$ & $(0.10)$ & $(0.15)$ & $(0.11)$ & $(0.05)$ & $(0.10)$ & $(0.14)$ & $(0.11)$ \\
$s_{AB}$ & \multicolumn{2}{c}{} & \multicolumn{2}{c}{} & \multicolumn{2}{c}{$2.93^{***}$} & \multicolumn{2}{c}{$3.26^{***}$} \\
 & \multicolumn{2}{c}{} & \multicolumn{2}{c}{} & \multicolumn{2}{c}{$(0.12)$} & \multicolumn{2}{c}{$(0.27)$} \\
$\rho$ & \multicolumn{2}{c}{$-0.08$} & \multicolumn{2}{c}{$0.12^{***}$} & \multicolumn{2}{c}{$-0.10^{*}$} & \multicolumn{2}{c}{$0.13^{***}$} \\
 & \multicolumn{2}{c}{$(0.06)$} & \multicolumn{2}{c}{$(0.04)$} & \multicolumn{2}{c}{$(0.06)$} & \multicolumn{2}{c}{$(0.04)$} \\\hline
 N &  \multicolumn{2}{c}{71067}   &  \multicolumn{2}{c}{49836}    & \multicolumn{2}{c}{71067}   &  \multicolumn{2}{c}{49836}  \\ 
  \hline 
\end{tabular}

\vspace{0.4cm}

\begin{minipage}{1\textwidth}
\footnotesize
\emph{Notes:} Columns 1 and 3 report the baseline Nigeria FGC--CMA estimates. Columns 2 and 4 report the corresponding robustness exercise that excludes women reported as cut before age 5. Columns 1--2 impose the specification without a separate joint-adoption spillover term $s_{AB}$, while Columns 3--4 allow for a distinct joint spillover parameter.  All specifications additionally control for indicators for primary education, secondary education, Muslim, urban residence, ethnicity (Hausa-Fulani, Igbo, Yoruba, with Other as the reference), region (North Central, North East, South East, South South, and South West, with North West as the reference), and wealth quintile (2--5). Standard errors are reported in parentheses. $^{***}p<0.01$, $^{**}p<0.05$, $^{*}p<0.1$.
\end{minipage}
\label{table:NG_FGC_first}
\end{table}

\begin{table}[H]
\centering
\caption{Nigeria: keeping only women cut before age 5}
\renewcommand{\arraystretch}{0.7}
\small
\begin{tabular}{lcccccccc}
\hline
& \multicolumn{2}{c}{Baseline} & \multicolumn{2}{c}{Robustness} & \multicolumn{2}{c}{Baseline} & \multicolumn{2}{c}{Robustness} \\
Variable & FGC & CMA & FGC & CMA & FGC & CMA & FGC & CMA \\
& (1) & (2) & (3) & (4) & (5) & (6) & (7) & (8) \\
\hline
CMA Ban &  & $-0.09^{***}$ &  & $-0.10^{***}$ &  & $-0.09^{***}$ &  & $-0.10^{***}$ \\
 &  & $(0.01)$ &  & $(0.01)$ &  & $(0.01)$ &  & $(0.01)$ \\
$\Gamma$ & \multicolumn{2}{c}{$0.14^{***}$} & \multicolumn{2}{c}{$0.25^{***}$} & \multicolumn{2}{c}{$0.16^{***}$} & \multicolumn{2}{c}{$0.18^{**}$} \\
 & \multicolumn{2}{c}{$(0.05)$} & \multicolumn{2}{c}{$(0.07)$} & \multicolumn{2}{c}{$(0.06)$} & \multicolumn{2}{c}{$(0.09)$} \\
$s$ & $2.04^{***}$ & $0.94^{***}$ & $2.03^{***}$ & $0.86^{***}$ & $2.05^{***}$ & $0.95^{***}$ & $2.03^{***}$ & $0.85^{***}$ \\
 & $(0.05)$ & $(0.10)$ & $(0.06)$ & $(0.10)$ & $(0.05)$ & $(0.10)$ & $(0.06)$ & $(0.10)$ \\
$s_{AB}$ & \multicolumn{2}{c}{} & \multicolumn{2}{c}{} & \multicolumn{2}{c}{$2.93^{***}$} & \multicolumn{2}{c}{$2.99^{***}$} \\
 & \multicolumn{2}{c}{} & \multicolumn{2}{c}{} & \multicolumn{2}{c}{$(0.12)$} & \multicolumn{2}{c}{$(0.14)$} \\
$\rho$ & \multicolumn{2}{c}{$-0.08$} & \multicolumn{2}{c}{$-0.23^{***}$} & \multicolumn{2}{c}{$-0.10^{*}$} & \multicolumn{2}{c}{$-0.18^{**}$} \\
 & \multicolumn{2}{c}{$(0.06)$} & \multicolumn{2}{c}{$(0.07)$} & \multicolumn{2}{c}{$(0.06)$} & \multicolumn{2}{c}{$(0.08)$} \\
 \hline
 N  & \multicolumn{2}{c}{71067}  &  \multicolumn{2}{c}{66198} & \multicolumn{2}{c}{71067} & \multicolumn{2}{c}{ 66198}  \\ 
\hline
\end{tabular}

\vspace{0.4cm}

\begin{minipage}{1\textwidth}
\footnotesize
\emph{Notes:} Columns 1 and 3 report the baseline Nigeria FGC--CMA estimates. Columns 2 and 4 report the corresponding robustness exercise that performs the estimation only on women cut before age 5. Columns 1--2 impose the specification without a separate joint-adoption spillover term $s_{AB}$, while Columns 3--4 allow for a distinct joint spillover parameter.   All specifications additionally control for the same indicators as in Table \ref{table:NG_FGC_first}. Standard errors are reported in parentheses. $^{***}p<0.01$, $^{**}p<0.05$, $^{*}p<0.1$.

\end{minipage}
\end{table}

\begin{table}[H]
\centering
\renewcommand{\arraystretch}{0.7}
\small 
\caption{Nigeria: coding ``don't know FGC" as 0 instead of NA}
\begin{tabular}{lcccccccc}
\hline
& \multicolumn{2}{c}{Baseline} & \multicolumn{2}{c}{Robustness} & \multicolumn{2}{c}{Baseline} & \multicolumn{2}{c}{Robustness} \\
Variable & FGC & CMA & FGC & CMA & FGC & CMA & FGC & CMA \\
& (1) & (2) & (3) & (4) & (5) & (6) & (7) & (8) \\

\hline
CMA Ban &  & $-0.09^{***}$ &  & $-0.07^{***}$ &  & $-0.09^{***}$ &  & $-0.06^{***}$ \\
 &  & $(0.01)$ &  & $(0.01)$ &  & $(0.01)$ &  & $(0.01)$ \\
$\Gamma$ & \multicolumn{2}{c}{$0.14^{***}$} & \multicolumn{2}{c}{$0.11^{***}$} & \multicolumn{2}{c}{$0.16^{***}$} & \multicolumn{2}{c}{$0.18^{***}$} \\
 & \multicolumn{2}{c}{$(0.05)$} & \multicolumn{2}{c}{$(0.04)$} & \multicolumn{2}{c}{$(0.06)$} & \multicolumn{2}{c}{$(0.05)$} \\
$s$ & $2.04^{***}$ & $0.94^{***}$ & $2.41^{***}$ & $1.12^{***}$ & $2.05^{***}$ & $0.95^{***}$ & $2.43^{***}$ & $1.14^{***}$ \\
 & $(0.05)$ & $(0.10)$ & $(0.05)$ & $(0.08)$ & $(0.05)$ & $(0.10)$ & $(0.05)$ & $(0.08)$ \\
$s_{AB}$ & \multicolumn{2}{c}{} & \multicolumn{2}{c}{} & \multicolumn{2}{c}{$2.93^{***}$} & \multicolumn{2}{c}{$3.37^{***}$} \\
 & \multicolumn{2}{c}{} & \multicolumn{2}{c}{} & \multicolumn{2}{c}{$(0.12)$} & \multicolumn{2}{c}{$(0.13)$} \\
$\rho$ & \multicolumn{2}{c}{$-0.08$} & \multicolumn{2}{c}{$-0.05$} & \multicolumn{2}{c}{$-0.10^{*}$} & \multicolumn{2}{c}{$-0.10^{**}$} \\
 & \multicolumn{2}{c}{$(0.06)$} & \multicolumn{2}{c}{$(0.04)$} & \multicolumn{2}{c}{$(0.06)$} & \multicolumn{2}{c}{$(0.05)$} \\
\hline
N &  \multicolumn{2}{c}{71067}  &  \multicolumn{2}{c}{113651} & \multicolumn{2}{c}{71067} & \multicolumn{2}{c}{113651} \\ 
\hline
\end{tabular}

\vspace{0.4cm}

\begin{minipage}{1\textwidth}
\footnotesize
\emph{Notes:} Columns 1 and 3 report the baseline Nigeria FGC--CMA estimates. Columns 2 and 4 report the corresponding robustness exercise that codes women who responded that they do not know what FGC is as not adhering to FGC, rather than NA. Columns 1--2 impose the specification without a separate joint-adoption spillover term $s_{AB}$, while Columns 3--4 allow for a distinct joint spillover parameter.    All specifications additionally control for the same indicators as in Table \ref{table:NG_FGC_first}. Standard errors are reported in parentheses. $^{***}p<0.01$, $^{**}p<0.05$, $^{*}p<0.1$.

\end{minipage}
\end{table}

\begin{table}[H]
\centering
\renewcommand{\arraystretch}{0.7}
\caption{Nigeria: Sharia law as additional covariate}
\small
\begin{tabular}{lcccccccc}
\hline
& \multicolumn{2}{c}{Baseline} & \multicolumn{2}{c}{Robustness} & \multicolumn{2}{c}{Baseline} & \multicolumn{2}{c}{Robustness} \\
Variable & FGC & CMA & FGC & CMA & FGC & CMA & FGC & CMA \\
& (1) & (2) & (3) & (4) & (5) & (6) & (7) & (8) \\

\hline
CMA Ban &  & $-0.09^{***}$ &  & $-0.12^{***}$ &  & $-0.09^{***}$ &  & $-0.12^{***}$ \\
 &  & $(0.01)$ &  & $(0.01)$ &  & $(0.01)$ &  & $(0.01)$ \\
Sharia Law &  &  & $-0.02$ & $0.09^{***}$ &  &  & $-0.02$ & $0.09^{***}$ \\
 &  &  & $(0.01)$ & $(0.02)$ &  &  & $(0.01)$ & $(0.02)$ \\
$\Gamma$ & \multicolumn{2}{c}{$0.14^{***}$} & \multicolumn{2}{c}{$0.12^{**}$} & \multicolumn{2}{c}{$0.16^{***}$} & \multicolumn{2}{c}{$0.14^{**}$} \\
 & \multicolumn{2}{c}{$(0.05)$} & \multicolumn{2}{c}{$(0.05)$} & \multicolumn{2}{c}{$(0.06)$} & \multicolumn{2}{c}{$(0.06)$} \\
$s$ & $2.04^{***}$ & $0.94^{***}$ & $2.04^{***}$ & $0.93^{***}$ & $2.05^{***}$ & $0.95^{***}$ & $2.05^{***}$ & $0.94^{***}$ \\
 & $(0.05)$ & $(0.10)$ & $(0.05)$ & $(0.10)$ & $(0.05)$ & $(0.10)$ & $(0.05)$ & $(0.10)$ \\
$s_{AB}$ & \multicolumn{2}{c}{} & \multicolumn{2}{c}{} & \multicolumn{2}{c}{$2.93^{***}$} & \multicolumn{2}{c}{$2.92^{***}$} \\
 & \multicolumn{2}{c}{} & \multicolumn{2}{c}{} & \multicolumn{2}{c}{$(0.12)$} & \multicolumn{2}{c}{$(0.12)$} \\
$\rho$ & \multicolumn{2}{c}{$-0.08$} & \multicolumn{2}{c}{$-0.06$} & \multicolumn{2}{c}{$-0.10^{*}$} & \multicolumn{2}{c}{$-0.08$} \\
 & \multicolumn{2}{c}{$(0.06)$} & \multicolumn{2}{c}{$(0.06)$} & \multicolumn{2}{c}{$(0.06)$} & \multicolumn{2}{c}{$(0.06)$} \\
 \hline
  N  & \multicolumn{2}{c}{71067}  &  \multicolumn{2}{c}{71067} & \multicolumn{2}{c}{71067} & \multicolumn{2}{c}{71067}  \\ 
\hline
\end{tabular}

\vspace{0.4cm}

\begin{minipage}{1\textwidth}
\footnotesize
\emph{Notes:} Columns 1 and 3 report the baseline Nigeria FGC--CMA estimates. Columns 2 and 4 report the corresponding robustness exercise that includes exposure to Sharia law legalization as additional covariates. Columns 1--2 impose the specification without a separate joint-adoption spillover term $s_{AB}$, while Columns 3--4 allow for a distinct joint spillover parameter.   All specifications additionally control for the same indicators as in Table \ref{table:NG_FGC_first}. Standard errors are reported in parentheses. $^{***}p<0.01$, $^{**}p<0.05$, $^{*}p<0.1$.

\end{minipage}

\end{table}

\begin{table}[H]
\centering
\renewcommand{\arraystretch}{0.7}
\caption{Nigeria: State-level CMA Ban}
\small
\begin{tabular}{lcccccccc}
\hline
& \multicolumn{2}{c}{Baseline} & \multicolumn{2}{c}{Robustness} & \multicolumn{2}{c}{Baseline} & \multicolumn{2}{c}{Robustness} \\
Variable & FGC & CMA & FGC & CMA & FGC & CMA & FGC & CMA \\
& (1) & (2) & (3) & (4) & (5) & (6) & (7) & (8) \\

\hline
CMA Ban &  & $-0.09^{***}$ &  & $-0.13^{***}$ &  & $-0.09^{***}$ &  & $-0.13^{***}$ \\
 &  & $(0.01)$ &  & $(0.02)$ &  & $(0.01)$ &  & $(0.02)$ \\
$\Gamma$ & \multicolumn{2}{c}{$0.14^{***}$} & \multicolumn{2}{c}{$0.13^{**}$} & \multicolumn{2}{c}{$0.16^{***}$} & \multicolumn{2}{c}{$0.15^{**}$} \\
 & \multicolumn{2}{c}{$(0.05)$} & \multicolumn{2}{c}{$(0.05)$} & \multicolumn{2}{c}{$(0.06)$} & \multicolumn{2}{c}{$(0.06)$} \\
$s$ & $2.04^{***}$ & $0.94^{***}$ & $2.04^{***}$ & $1.01^{***}$ & $2.05^{***}$ & $0.95^{***}$ & $2.05^{***}$ & $1.03^{***}$ \\
 & $(0.05)$ & $(0.10)$ & $(0.05)$ & $(0.09)$ & $(0.05)$ & $(0.10)$ & $(0.05)$ & $(0.09)$ \\
$s_{AB}$ & \multicolumn{2}{c}{} & \multicolumn{2}{c}{} & \multicolumn{2}{c}{$2.93^{***}$} & \multicolumn{2}{c}{$3.00^{***}$} \\
 & \multicolumn{2}{c}{} & \multicolumn{2}{c}{} & \multicolumn{2}{c}{$(0.12)$} & \multicolumn{2}{c}{$(0.12)$} \\
$\rho$ & \multicolumn{2}{c}{$-0.08$} & \multicolumn{2}{c}{$-0.07$} & \multicolumn{2}{c}{$-0.10^{*}$} & \multicolumn{2}{c}{$-0.09$} \\
 & \multicolumn{2}{c}{$(0.06)$} & \multicolumn{2}{c}{$(0.06)$} & \multicolumn{2}{c}{$(0.06)$} & \multicolumn{2}{c}{$(0.06)$} \\
 \hline
  N  & \multicolumn{2}{c}{71067}  &  \multicolumn{2}{c}{71067} & \multicolumn{2}{c}{71067} & \multicolumn{2}{c}{71067}  \\ 
\hline
\end{tabular}

\vspace{0.4cm}

\begin{minipage}{1\textwidth}
\footnotesize
\emph{Notes:} Columns 1 and 3 report the baseline Nigeria FGC--CMA estimates. Columns 2 and 4 report the corresponding robustness exercise that defines the CMA ban variable in terms of state-level ratifications of the child marriage ban. Columns 1--2 impose the specification without a separate joint-adoption spillover term $s_{AB}$, while Columns 3--4 allow for a distinct joint spillover parameter.   All specifications additionally control for the same indicators as in Table \ref{table:NG_FGC_first}. Standard errors are reported in parentheses. $^{***}p<0.01$, $^{**}p<0.05$, $^{*}p<0.1$.

\end{minipage}

\end{table}

\subsection{PGY Robustness} \label{sec:pgy_robustness}
 
 \begin{table}[H]
\centering
\renewcommand{\arraystretch}{0.7}
\small
\caption{Nigeria: Drought as additional exclusion restriction}
\begin{tabular}{lcccccccc}
\hline
& \multicolumn{2}{c}{Baseline} & \multicolumn{2}{c}{Robustness} & \multicolumn{2}{c}{Baseline} & \multicolumn{2}{c}{Robustness} \\
Variable & PGY & CMA & PGY & CMA & PGY & CMA & PGY & CMA \\
& (1) & (2) & (3) & (4) & (5) & (6) & (7) & (8) \\

\hline
CMA Ban &  & $-0.04^{***}$ &  & $-0.04^{***}$ &  & $-0.04^{***}$ &  & $-0.04^{***}$ \\
 &  & $(0.01)$ &  & $(0.01)$ &  & $(0.01)$ &  & $(0.01)$ \\
Drought &  &  &  & $0.00$ &  &  &  & $0.00$ \\
 &  &  &  & $(0.01)$ &  &  &  & $(0.01)$ \\
$\Gamma$ & \multicolumn{2}{c}{$-0.06^{*}$} & \multicolumn{2}{c}{$-0.06^{*}$} & \multicolumn{2}{c}{$0.05$} & \multicolumn{2}{c}{$0.05$} \\
 & \multicolumn{2}{c}{$(0.03)$} & \multicolumn{2}{c}{$(0.03)$} & \multicolumn{2}{c}{$(0.05)$} & \multicolumn{2}{c}{$(0.05)$} \\
$s$ & $2.27^{***}$ & $1.24^{***}$ & $2.27^{***}$ & $1.24^{***}$ & $2.45^{***}$ & $1.26^{***}$ & $2.45^{***}$ & $1.26^{***}$ \\
 & $(0.11)$ & $(0.07)$ & $(0.11)$ & $(0.07)$ & $(0.09)$ & $(0.07)$ & $(0.09)$ & $(0.07)$ \\
$s_{AB}$ & \multicolumn{2}{c}{} & \multicolumn{2}{c}{} & \multicolumn{2}{c}{$3.01^{***}$} & \multicolumn{2}{c}{$3.01^{***}$} \\
 & \multicolumn{2}{c}{} & \multicolumn{2}{c}{} & \multicolumn{2}{c}{$(0.11)$} & \multicolumn{2}{c}{$(0.11)$} \\
$\rho$ & \multicolumn{2}{c}{$0.11^{***}$} & \multicolumn{2}{c}{$0.11^{***}$} & \multicolumn{2}{c}{$0.07$} & \multicolumn{2}{c}{$0.07$} \\
 & \multicolumn{2}{c}{$(0.03)$} & \multicolumn{2}{c}{$(0.03)$} & \multicolumn{2}{c}{$(0.04)$} & \multicolumn{2}{c}{$(0.04)$} \\
\hline
 N  & \multicolumn{2}{c}{128785}  &  \multicolumn{2}{c}{128785} & \multicolumn{2}{c}{128785} & \multicolumn{2}{c}{128785}  \\ 
\hline
\end{tabular}

\vspace{0.4cm}

\begin{minipage}{1\textwidth}
\footnotesize
\emph{Notes:} Columns 1 and 3 report the baseline Nigeria PGY--CMA estimates. Columns 2 and 4 report the corresponding robustness exercise that includes dummies for drought exposure. Columns 1--2 impose the specification without a separate joint-adoption spillover term $s_{AB}$, while Columns 3--4 allow for a distinct joint spillover parameter.  All specifications additionally control for the same indicators as in Table \ref{table:NG_FGC_first}. Standard errors are reported in parentheses. $^{***}p<0.01$, $^{**}p<0.05$, $^{*}p<0.1$.

\end{minipage}

\end{table}

\vspace{-10pt}

\begin{table}[H]
\centering
\renewcommand{\arraystretch}{0.7}
\caption{Nigeria: Drought as only exclusion restriction}
\small
\begin{tabular}{lcccccccc}
\hline
& \multicolumn{2}{c}{Baseline} & \multicolumn{2}{c}{Robustness} & \multicolumn{2}{c}{Baseline} & \multicolumn{2}{c}{Robustness} \\
Variable & PGY & CMA & PGY & CMA & PGY & CMA & PGY & CMA \\
& (1) & (2) & (3) & (4) & (5) & (6) & (7) & (8) \\

\hline
CMA Ban &  & $-0.04^{***}$ &  &  &  & $-0.04^{***}$ &  &  \\
 &  & $(0.01)$ &  &  &  & $(0.01)$ &  &  \\
Drought &  &  &  & $0.01$ &  &  &  & $0.01$ \\
 &  &  &  & $(0.01)$ &  &  &  & $(0.01)$ \\
$\Gamma$ & \multicolumn{2}{c}{$-0.06^{*}$} & \multicolumn{2}{c}{$-0.06^{*}$} & \multicolumn{2}{c}{$0.05$} & \multicolumn{2}{c}{$0.05$} \\
 & \multicolumn{2}{c}{$(0.03)$} & \multicolumn{2}{c}{$(0.03)$} & \multicolumn{2}{c}{$(0.05)$} & \multicolumn{2}{c}{$(0.05)$} \\
$s$ & $2.27^{***}$ & $1.24^{***}$ & $2.27^{***}$ & $1.35^{***}$ & $2.45^{***}$ & $1.26^{***}$ & $2.46^{***}$ & $1.37^{***}$ \\
 & $(0.11)$ & $(0.07)$ & $(0.11)$ & $(0.06)$ & $(0.09)$ & $(0.07)$ & $(0.09)$ & $(0.06)$ \\
$s_{AB}$ & \multicolumn{2}{c}{} & \multicolumn{2}{c}{} & \multicolumn{2}{c}{$3.01^{***}$} & \multicolumn{2}{c}{$3.12^{***}$} \\
 & \multicolumn{2}{c}{} & \multicolumn{2}{c}{} & \multicolumn{2}{c}{$(0.11)$} & \multicolumn{2}{c}{$(0.10)$} \\
$\rho$ & \multicolumn{2}{c}{$0.11^{***}$} & \multicolumn{2}{c}{$0.11^{***}$} & \multicolumn{2}{c}{$0.07$} & \multicolumn{2}{c}{$0.07$} \\
 & \multicolumn{2}{c}{$(0.03)$} & \multicolumn{2}{c}{$(0.03)$} & \multicolumn{2}{c}{$(0.04)$} & \multicolumn{2}{c}{$(0.05)$} \\

\hline
 N  & \multicolumn{2}{c}{128785}  &  \multicolumn{2}{c}{128785} & \multicolumn{2}{c}{128785} & \multicolumn{2}{c}{128785}  \\ 
\hline
\end{tabular}

\vspace{0.4cm}

\begin{minipage}{1\textwidth}
\footnotesize
\emph{Notes:} Columns 1 and 3 report the baseline Nigeria PGY--CMA estimates. Columns 2 and 4 report the corresponding robustness exercise that includes dummies for drought exposure while removing the CMA ban dummy. Columns 1--2 impose the specification without a separate joint-adoption spillover term $s_{AB}$, while Columns 3--4 allow for a distinct joint spillover parameter.  All specifications additionally control for the same indicators as in Table \ref{table:NG_FGC_first}. Standard errors are reported in parentheses. $^{***}p<0.01$, $^{**}p<0.05$, $^{*}p<0.1$.

\end{minipage}

\end{table}

\begin{table}[H]
\centering
\renewcommand{\arraystretch}{0.7}
\small
\caption{Nigeria: State-level CMA Ban}
\begin{tabular}{lcccccccc}
\hline
& \multicolumn{2}{c}{Baseline} & \multicolumn{2}{c}{Robustness} & \multicolumn{2}{c}{Baseline} & \multicolumn{2}{c}{Robustness} \\
Variable & PGY & CMA & PGY & CMA & PGY & CMA & PGY & CMA \\
& (1) & (2) & (3) & (4) & (5) & (6) & (7) & (8) \\

\hline
CMA Ban &  & $-0.04^{***}$ &  & $-0.11^{***}$ &  & $-0.04^{***}$ &  & $-0.10^{***}$ \\
 &  & $(0.01)$ &  & $(0.01)$ &  & $(0.01)$ &  & $(0.01)$ \\
$\Gamma$ & \multicolumn{2}{c}{$-0.06^{*}$} & \multicolumn{2}{c}{$-0.05^{*}$} & \multicolumn{2}{c}{$0.05$} & \multicolumn{2}{c}{$0.07$} \\
 & \multicolumn{2}{c}{$(0.03)$} & \multicolumn{2}{c}{$(0.03)$} & \multicolumn{2}{c}{$(0.05)$} & \multicolumn{2}{c}{$(0.04)$} \\
$s$ & $2.27^{***}$ & $1.24^{***}$ & $2.27^{***}$ & $1.21^{***}$ & $2.45^{***}$ & $1.26^{***}$ & $2.44^{***}$ & $1.23^{***}$ \\
 & $(0.11)$ & $(0.07)$ & $(0.11)$ & $(0.06)$ & $(0.09)$ & $(0.07)$ & $(0.09)$ & $(0.06)$ \\
$s_{AB}$ & \multicolumn{2}{c}{} & \multicolumn{2}{c}{} & \multicolumn{2}{c}{$3.01^{***}$} & \multicolumn{2}{c}{$2.98^{***}$} \\
 & \multicolumn{2}{c}{} & \multicolumn{2}{c}{} & \multicolumn{2}{c}{$(0.11)$} & \multicolumn{2}{c}{$(0.10)$} \\
$\rho$ & \multicolumn{2}{c}{$0.11^{***}$} & \multicolumn{2}{c}{$0.09^{***}$} & \multicolumn{2}{c}{$0.07$} & \multicolumn{2}{c}{$0.05$} \\
 & \multicolumn{2}{c}{$(0.03)$} & \multicolumn{2}{c}{$(0.03)$} & \multicolumn{2}{c}{$(0.04)$} & \multicolumn{2}{c}{$(0.04)$} \\
\hline
 N &  \multicolumn{2}{c}{128785}   &  \multicolumn{2}{c}{128785}    & \multicolumn{2}{c}{128785}   &  \multicolumn{2}{c}{128785}  \\ 
  \hline 
\end{tabular}

\vspace{0.4cm}

\begin{minipage}{1\textwidth}
\footnotesize
\emph{Notes:} Columns 1 and 3 report the baseline Nigeria PGY--CMA estimates. Columns 2 and 4 report the corresponding robustness exercise that defines the CMA ban variable in terms of state-level ratifications of the child marriage ban. Columns 1--2 impose the specification without a separate joint-adoption spillover term $s_{AB}$, while Columns 3--4 allow for a distinct joint spillover parameter.  All specifications additionally control for the same indicators as in Table \ref{table:NG_FGC_first}. Standard errors are reported in parentheses. $^{***}p<0.01$, $^{**}p<0.05$, $^{*}p<0.1$.

\end{minipage}
\end{table}

\begin{table}[H]
\centering
\small
\caption{Sierra Leone: Drought as additional exclusion restriction}
\renewcommand{\arraystretch}{0.7}
\begin{tabular}{lcccccccc}
\hline
& \multicolumn{2}{c}{Baseline} & \multicolumn{2}{c}{Robustness} & \multicolumn{2}{c}{Baseline} & \multicolumn{2}{c}{Robustness} \\
Variable & PGY & CMA & PGY & CMA & PGY & CMA & PGY & CMA \\
& (1) & (2) & (3) & (4) & (5) & (6) & (7) & (8) \\

\hline
CMA Ban &  & $-0.04^{*}$ &  & $-0.04^{*}$ &  & $-0.04^{**}$ &  & $-0.04^{**}$ \\
 &  & $(0.02)$ &  & $(0.02)$ &  & $(0.02)$ &  & $(0.02)$ \\
Drought &  &  &  & $0.00$ &  &  &  & $0.00$ \\
 &  &  &  & $(0.01)$ &  &  &  & $(0.01)$ \\
$\Gamma$ & \multicolumn{2}{c}{$-0.17^{**}$} & \multicolumn{2}{c}{$-0.17^{**}$} & \multicolumn{2}{c}{$-0.07$} & \multicolumn{2}{c}{$-0.07$} \\
 & \multicolumn{2}{c}{$(0.07)$} & \multicolumn{2}{c}{$(0.06)$} & \multicolumn{2}{c}{$(0.06)$} & \multicolumn{2}{c}{$(0.06)$} \\
$s$ & $1.30^{***}$ & $0.55^{***}$ & $1.30^{***}$ & $0.55^{***}$ & $1.45^{***}$ & $0.56^{***}$ & $1.45^{***}$ & $0.56^{***}$ \\
 & $(0.21)$ & $(0.13)$ & $(0.21)$ & $(0.13)$ & $(0.22)$ & $(0.13)$ & $(0.22)$ & $(0.14)$ \\
$s_{AB}$ & \multicolumn{2}{c}{} & \multicolumn{2}{c}{} & \multicolumn{2}{c}{$1.44^{***}$} & \multicolumn{2}{c}{$1.44^{***}$} \\
 & \multicolumn{2}{c}{} & \multicolumn{2}{c}{} & \multicolumn{2}{c}{$(0.29)$} & \multicolumn{2}{c}{$(0.29)$} \\
$\rho$ & \multicolumn{2}{c}{$0.24^{***}$} & \multicolumn{2}{c}{$0.25^{***}$} & \multicolumn{2}{c}{$0.19^{***}$} & \multicolumn{2}{c}{$0.19^{***}$} \\
 & \multicolumn{2}{c}{$(0.07)$} & \multicolumn{2}{c}{$(0.07)$} & \multicolumn{2}{c}{$(0.06)$} & \multicolumn{2}{c}{$(0.06)$} \\

\hline
 N  & \multicolumn{2}{c}{33591}  &  \multicolumn{2}{c}{33591} & \multicolumn{2}{c}{33591} & \multicolumn{2}{c}{33591}  \\ 
 \hline
\end{tabular}

\vspace{0.4cm}

\begin{minipage}{1\textwidth}
\footnotesize
\emph{Notes:} Columns 1 and 3 report the baseline Sierra Leone Leone PGY--CMA estimates. Columns 2 and 4 report the corresponding robustness exercise that includes dummies for drought exposure. Columns 1--2 impose the specification without a separate joint-adoption spillover term $s_{AB}$, while Columns 3--4 allow for a distinct joint spillover parameter.   All specifications additionally control for the same indicators as in Table \ref{table:FGC_SL_first}. Standard errors are reported in parentheses. $^{***}p<0.01$, $^{**}p<0.05$, $^{*}p<0.1$.

\end{minipage}
\end{table}

\end{document}